\documentclass{amsart}[draft]

\usepackage{mathtools}
\usepackage{placeins}
\usepackage{slashed}
\usepackage{adjustbox}
\usepackage{amsmath}	
\usepackage{amssymb}
\usepackage{graphicx}
\usepackage{mathrsfs}	
\usepackage{multicol}
\usepackage{soul}	
\usepackage{color}	
\usepackage{appendix}
\usepackage{eufrak}
\usepackage[labelformat=simple]{subcaption}
\usepackage[backend=bibtex,style=numeric,citestyle=numeric,maxbibnames=99,giveninits,doi=false,isbn=false,url=false,eprint=false,sortcites=true]{biblatex}
\renewbibmacro{in:}{\ifentrytype{article}{}{\printtext{\bibstring{in}\intitlepunct}}}
 
\newcommand{\de}{\ensuremath{\partial}}
\newcommand{\dee}{\ensuremath{\textrm{d}}}

\newcommand{\inty}[4]{\ensuremath{ \int_{#1}^{#2} \! #3 \, \dee#4 }}

\newcommand{\field}[1]{\mathbb{#1}}
\newcommand{\ip}[2]{\ensuremath{ \left< \left. #1 \right| #2 \right> } }

\newcommand{\bra}[1]{\langle #1|}
\newcommand{\ket}[1]{| #1 \rangle}

\usepackage{booktabs}
\usepackage[font=footnotesize,labelfont=bf]{caption}
\newcommand{\comment}[1]{}
\newcommand{\la}{\langle}
\newcommand{\ra}{\rangle}
\newcommand{\cC}{\mathcal{C}}
\newcommand{\cD}{\mathcal{D}}

\newcommand{\cB}{\mathcal{B}}
\newcommand{\cJ}{\mathcal{J}}

\newcommand{\R}{\mathbb{R}}

\DeclareMathOperator{\range}{range}
\DeclareMathOperator{\diag}{diag}
\DeclareMathOperator{\im}{Im}
\DeclareMathOperator{\divd}{div}
\DeclareMathOperator{\loc}{loc}
\DeclareMathOperator{\uloc}{uloc}

\newcommand\numberthis{\addtocounter{equation}{1}\tag{\theequation}}
\usepackage{enumitem}
\usepackage{bm}

\DeclareMathOperator{\dist}{dist}
\DeclareMathOperator{\diam}{diam}

\newtheorem{definition}{Definition}
\newtheorem{assumption}{Assumption}
\newtheorem{remark}{Remark}
\newtheorem{theorem}{Theorem}
\newtheorem*{theorem*}{Theorem}
\newtheorem*{maintheorem*}{Main Theorem}
\newtheorem{lemma}{Lemma}
\newtheorem*{lemma*}{Lemma}
\newtheorem*{conjecture*}{Conjecture}
\newtheorem{proposition}{Proposition}
\newtheorem{corollary}[proposition]{Corollary}

\renewcommand{\vec}[1]{\boldsymbol{#1}}

\bibliography{bibliography}

\numberwithin{lemma}{section}
\numberwithin{example}{section}
\numberwithin{figure}{section}
\numberwithin{proposition}{section}
\numberwithin{equation}{section}
\numberwithin{theorem}{section}
\numberwithin{remark}{section}
\numberwithin{definition}{section}
\numberwithin{assumption}{section}

\usepackage{tabto}

\allowdisplaybreaks

\title[Generalized Wannier functions for non-periodic systems]{Existence and computation of generalized Wannier functions for non-periodic systems in two dimensions and higher}

\author{Kevin D. Stubbs, Alexander B. Watson, and Jianfeng Lu}

\usepackage{xcolor}
\definecolor{purp}{RGB}{160, 32, 240}
\definecolor{lightblue}{RGB}{32, 160, 240}

\begin{document}

\maketitle

\begin{abstract}
Exponentially-localized Wannier functions (ELWFs) are an orthonormal basis of the Fermi projection of a material consisting of functions which decay exponentially fast away from their maxima.
When the material is insulating and crystalline, conditions which guarantee existence of ELWFs in dimensions one, two, and three are well-known, and methods for constructing the ELWFs numerically are well-developed.
We consider the case where the material is insulating but not necessarily crystalline, where much less is known. In one spatial dimension, Kivelson and Nenciu-Nenciu have proved ELWFs can be constructed as the eigenfunctions of a self-adjoint operator acting on the Fermi projection. In this work, we identify an assumption under which we can generalize the Kivelson-Nenciu-Nenciu result to two dimensions and higher. Under this assumption, we prove that ELWFs can be constructed as the eigenfunctions of a sequence of self-adjoint operators acting on the Fermi projection.
\end{abstract}

\section{Introduction} \label{sec:introduction}
The starting point for understanding electronic properties of materials is the many-body ground state of the material's electrons.
In the independent electron approximation, electrons in the ground state occupy the eigenstates of the single-electron Hamiltonian with energy up to the Fermi level. 
The subspace of the single-electron Hilbert space occupied by electrons in the ground state is known as the Fermi projection \cite{ashcroft_mermin}. It is often desirable
to find orthonormal bases of the Fermi projection which are as spatially localized as possible. Such bases are important both for theoretical and numerical studies of materials. For example, they form the basis of the modern theory of polarization \cite{1993King-SmithVanderbilt,1994Resta,2012MarzariMostofiYatesSouzaVanderbilt}, and can dramatically speed up numerical calculations \cite{1999Goedecker,2005LeeNardelliMarzari,2006StengelSpaldin,2012MarzariMostofiYatesSouzaVanderbilt}.

For insulating \textit{crystalline} materials, it is natural to work with orthonormal bases which share the periodic structure of the material in the sense that the basis is closed under translation by Bravais lattice vectors. Such bases are known as Wannier bases and their elements are known as Wannier functions.
Wannier functions can be
constructed (in the simplest case) by integrating the Bloch functions 
of the occupied Bloch bands
with respect to the quasi-momentum over the Brillouin zone \cite{1937Wannier}. 
Since each Bloch function is only defined up to a complex phase, 
or ``gauge'', Wannier functions are not unique. 
By changing the gauge, one can change the spatial localization of the corresponding Wannier functions. 

In pioneering work, Kohn 
found that for non-degenerate Bloch bands of inversion-symmetric crystals in one spatial dimension it is always possible to choose the gauge of the Bloch functions such that the associated Wannier functions decay \emph{exponentially fast} in space \cite{1959Kohn}.
In the years since, many authors have worked to generalize Kohn's result. A summary of these results is as follows. 
\begin{itemize}
\item In one spatial dimension, the Fermi projection of an insulating crystalline material can always be represented by exponentially-localized Wannier functions \cite{1964DesCloizeaux,1964DesCloizeaux2,1983Nenciu,1988HelfferSjostrand,1991Nenciu}. 
\item In two dimensions, the same result holds if and only if the Chern number, a topological invariant associated with the occupied Bloch functions, vanishes \cite{1964DesCloizeaux,1964DesCloizeaux2,1983Nenciu,1988HelfferSjostrand,1991Nenciu,2007BrouderPanatiCalandraMourougane,2007Panati,2018MonacoPanatiPisanteTeufel}. 
\item In three dimensions, the result holds as long as three ``Chern-like'' numbers all vanish \cite{1964DesCloizeaux,1964DesCloizeaux2,1983Nenciu,1988HelfferSjostrand,1991Nenciu,2007BrouderPanatiCalandraMourougane,2007Panati}. 
\end{itemize}
An important special case of these results is that exponentially-localized Wannier bases always exist whenever the insulating crystal is symmetric under (Bosonic or Fermionic) time-reversal, since this implies that the Chern and Chern-like numbers vanish \cite{2007Panati}. When these symmetries hold, a further challenge is to find Wannier functions which respect the symmetries, this additional constraint requires more refined methods and can create further topological obstructions; see, for example \cite{2006FuKane,2016CorneanHerbstNenciu,2016FiorenzaMonacoPanati,2016FiorenzaMonacoPanati2,2017CorneanMonaco,2017CorneanMonacoTeufel}.

For insulating materials without crystalline atomic structure, Bloch functions do not exist and hence cannot be used to construct a spatially localized basis. In spite of this, Nenciu-Nenciu \cite{1998NenciuNenciu}, following 
Kivelson \cite{1982Kivelson}, proved by construction that exponentially-localized bases of the Fermi projection always exist in one spatial dimension, and that this same construction applied to a crystalline material yields a Wannier basis.
Exponentially-localized bases of the Fermi projection are conjectured to exist more generally, at least whenever time-reversal symmetry holds \cite{1993NenciuNenciu,1991Niu}. Beyond one spatial dimension, their existence has been proven in a few special cases \cite{2008CorneanNenciuNenciu,1993NenciuNenciu,1993GellerKohn,2015Prodan,2016CorneanHerbstNenciu} (see Section \ref{sec:previous_works} for discussion of these works); for more details of the one-dimensional case see \cite{1974RehrKohn,1974KohnOnffroy}. In what follows, we will refer to any localized basis of the Fermi projection of a non-periodic insulator as a generalized Wannier basis, and to the elements of such a basis as generalized Wannier functions.

In this work, we introduce and prove validity of a new construction of exponentially localized generalized Wannier functions for periodic and non-periodic insulators in two dimensions and higher. Our results can be thought of as generalizing Kivelson and Nenciu-Nenciu's ideas to higher dimensions, although our method relies on novel assumptions, one for each dimension higher than one, which are unnecessary in one dimension. We provide full details of our proofs for the case of an infinite material in two dimensions described by a continuum PDE model. Our proofs, being somewhat more operator-theoretic than those of Nenciu-Nenciu, extend in straightforward ways (up to assumptions and details which we explain in each case) to discrete, finite, and higher dimensional models. Our proof of exponential localization, in particular, requires several new ideas compared with that of Nenciu-Nenciu.


When applied to periodic insulating materials, our construction yields exponentially localized Wannier functions, and hence our assumptions imply vanishing of the Chern number in two dimensions, and of the three Chern-like numbers in three dimensions. Since our construction preserves time-reversal symmetries of the Hamiltonian, our assumptions also imply vanishing of the topological obstructions which appear when time-reversal symmetries are imposed. We conjecture that our assumptions imply vanishing of topological indices formulated for non-periodic insulators, and for insulators in other symmetry classes, as well. In the case of the two dimensional Chern marker, this follows from combining our results with \cite{Marcelli2020} (see also \cite{Lu2021}). For other examples of non-periodic formulations of topological invariants, see \cite{1994BellissardvanElstSchulz-Baldes,2015Loring,2006Kitaev,2011BiancoResta,2010Prodan,2018GrafShapiro}.

Our construction implies a new algorithm for numerically computing generalized Wannier functions in finite systems, both with periodic boundary conditions and otherwise. Numerical tests on the two dimensional Haldane model suggest that our assumption generally holds whenever the model is in its non-topological phase, and we find that our algorithm indeed yields exponentially localized generalized Wannier functions for that model, even with weak disorder.
The algorithm fails in the topological (Chern) insulator phase because our assumption does not hold in this case. We show that this algorithm can be used to compute two dimensional Wannier functions which respect different choices of boundary conditions, and model symmetries such as time-reversal symmetry, in \cite{Stubbs2020}. In three dimensions, our construction requires an additional assumption, and we leave numerical exploration of this case to future work, although we argue that our assumptions are always valid, at least, for insulators in arbitrary dimensions close to the ``atomic limit'', where onsite potentials dominate inter-site hopping.\footnote{The importance of this limit was emphasized in \cite{Bradlyn2017}.} We remark that in general, methods for numerical computation of \emph{generalized} Wannier functions are significantly less developed than methods for computing Wannier functions \cite{1997MarzariVanderbilt,2001SouzaMarzariVanderbilt,2012MarzariMostofiYatesSouzaVanderbilt,2010ELiLu,2015DamleLinYing,2017DamleLinYing}, although see Appendix A of \cite{1997MarzariVanderbilt} and \cite{1998SilvestrelliMarzariVanderbiltParrinello}. Numerical methods for computing generalized Wannier functions in finite systems can be viewed as Boys localization schemes \cite{1960Boys}.

\subsection{Paper Organization}
The remainder of our paper is organized as follows. We present our main theoretical result without making our assumptions completely explicit, and explain in what sense our result generalizes Kivelson-Nenciu-Nenciu's one dimensional result, in Section \ref{sec:kivelson_idea}. We summarize what is known about the validity of our assumptions in Section \ref{sec:outlook}. We review previous literature on constructing generalized Wannier functions for non-periodic materials in Section \ref{sec:previous_works}. We show results of implementing our numerical algorithm in Section \ref{sec:numerical-results}. We review some notations in Section \ref{sec:notations}, before presenting the strategy of the proof of the main theorem, and make our assumptions precise, in Section \ref{sec:results}. The proof of our main theorem is presented across Sections \ref{sec:self-adjointness}, \ref{sec:pjypj-disc-spec}, and \ref{sec:pjypj-exp-loc}. We explain how our results may be generalized to higher dimensional systems in Section \ref{sec:higher-d}, and to discrete systems in Section \ref{sec:discrete}. We defer proofs of key estimates required for the proof of our main theorem to Appendices \ref{sec:pj-props} and \ref{sec:shifting-proof}. We finish by proving in Appendix \ref{sec:gWFs_closed} that, for periodic models satisfying our assumptions, the generalized Wannier functions yielded by our construction are actually Wannier functions.

\subsection{The Kivelson-Nenciu-Nenciu Idea and our Main Theorem} \label{sec:kivelson_idea} 
In this section, we begin by reviewing Kivelson-Nenciu-Nenciu's construction of exponentially localized generalized Wannier functions in one spatial dimension. We will then present our main theorem without making our assumptions completely precise. We will make our assumptions precise in Section \ref{sec:results}. 

Consider the Hilbert space $\mathcal{H} = L^2(\field{R})$ and let $H$ be a Hamiltonian of the form
\[
  H = -\Delta + V(x),
\]
where $V$ is a real potential, not necessarily periodic, satisfying $V \in L^2_{\uloc}(\field{R})$. Assume that the Fermi level lies in a spectral gap of $H$, let $P$ denote the Fermi projection\footnote{Strictly speaking, the Fermi level is not well-defined unless $V$ satisfies further assumptions to ensure the system is spatially ergodic. More precisely, what is required here is that $P$ must be the projection onto an isolated part of the spectrum of $H$.}, and let $X$ denote the position operator $X f(x) = x f(x)$. The Kivelson-Nenciu-Nenciu idea is that \emph{in one spatial dimension, the eigenfunctions of the operator $PXP$ form an exponentially localized generalized Wannier basis}. A sketch of Nenciu-Nenciu's rigorous proof that this proposal works is as follows.

Using exponential decay of $P$ (Lemma 1 of \cite{1998NenciuNenciu}), $PXP$ is well-defined and self-adjoint $\mathcal{D}(X) \cap \range(P) \rightarrow \range(P)$.
Because of decay induced by $X$ (in the resolvent), $PXP$ has compact resolvent and hence purely real and discrete spectrum.
Since the spectral theorem implies that the eigenfunctions of $PXP$ form an orthonormal basis of $\range(P)$, it remains only to prove that the eigenfunctions of $PXP$ exponentially decay. This can be verified by a direct calculation from the eigenequation $PXP f = \lambda f$ which again relies on exponential decay of $P$. When $V$ is invariant under translations, $PXP$ commutes with the translation operator up to a constant factor, and hence the eigenfunctions of $PXP$ are Wannier functions.\footnote{This is the reason for working with $P X P$ rather than, say, $P X^2 P$, which would yield generalized Wannier functions that do not satisfy this property.}

In this work, we focus on the Hilbert space $\mathcal{H} = L^2(\field{R}^2)$ and let $H$ be a Hamiltonian of the form
\begin{equation} \label{eq:Ham}
    H = ( - i \nabla + A(\vec{x}) )^2 + V(\vec{x}).
\end{equation}
Suppose $V$ is a real scalar function and $A$ is a real vector function, not necessarily periodic, satisfying $V \in L^2_{\uloc}(\field{R}^2)$, $A \in L^4_{\loc}(\field{R}^2; \field{R}^2)$, and $\divd A \in L^2_{\loc}(\field{R}^2)$.
Assume again that the Fermi level lies in a spectral gap of $H$, and let $P$ denote the Fermi projection. 

The Kivelson-Nenciu-Nenciu idea does not generalize to this case in a straightforward way for the following reason. Suppose we let $X$ denote a two-dimensional position operator acting as $X f(\vec{x}) = x_1 f(\vec{x})$ with respect to some choice of non-parallel co-ordinate axes. Then the decay induced by $X$ is not enough for the resolvent of $PXP$ to be compact in two dimensions (unless the system is effectively one-dimensional, see \cite{2008CorneanNenciuNenciu}). To generalize the Kivelson-Nenciu-Nenciu idea to two dimensions, we make an additional assumption. The additional assumption we make is that the operator $PXP$ has \textit{uniform spectral gaps}, a notion we make precise in Assumption \ref{def:usg}. 
Numerical simulations on the Haldane model \cite{1988Haldane} suggest this assumption is equivalent to vanishing of the Chern number in this case: see Figure \ref{fig:pxp-evalues}. 

\begin{figure*}
  \centering
  \title{\textbf{Sorted Non-Zero Eigenvalues of $PXP$}}

  \includegraphics[width=.8\linewidth]{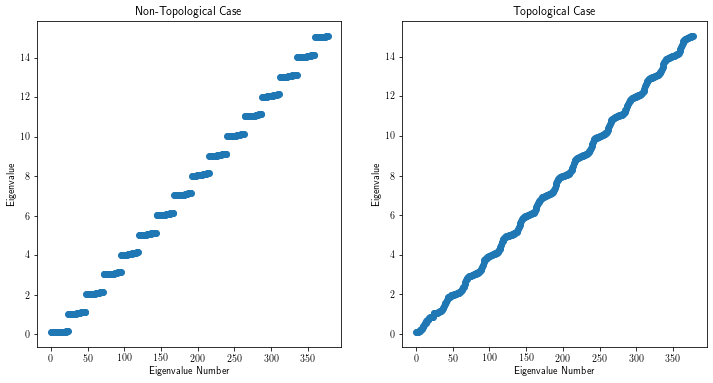}
    \caption{Detail from plot of the sorted non-zero eigenvalues of $PXP$ where $P$ is the Fermi projection and $X$ is the lattice position operator $\left[ X \psi \right]_{m,n} = \left[ m \psi \right]_{m,n}$ for the Haldane model on a $24 \times 24$ lattice with periodic boundary conditions. Parameters for the Haldane model are defined in Section \ref{sec:haldane-def}. The left plot corresponds to a non-topological phase (Chern = 0) with parameters $(t,t',v,\phi) = (1,0,1,0)$. The right plot corresponds to topological phase (Chern = 1) with parameters $(t,t',v,\phi) = (1,\frac{1}{4},1,\frac{\pi}{2})$. In the non-topological phase, the spectrum of $PXP$ shows clear gaps, while in the topological phase, the spectrum does not have clear gaps. We can gain additional insight into these results by Bloch transforming $PXP$ with respect to $x_2$; see Figure \ref{fig:charge_centers}.}
  \label{fig:pxp-evalues}
\end{figure*}

The uniform spectral gap assumption allows us to reduce the original problem of finding an exponentially-localized basis of $\range(P)$ to the problem of finding exponentially-localized bases of the set of subspaces $\range(P_j)$, where $P_j$ denotes the spectral projection onto each separated component of the spectrum of $PXP$. Crucially, functions in $\range(P_j)$ are quasi-one dimensional in the sense that they decay with respect to $X$ away from lines $x_1 = \eta_j$, where $\eta_j$ is a real constant, for each $j$. Using this property, we can apply the Kivelson-Nenciu-Nenciu idea to each $\range(P_j)$ in turn, and thereby build up a generalized Wannier basis of all of $\range(P)$.

We consider the family of operators $P_j Y P_j$, where $Y$ is a position operator acting in a non-parallel direction to $X$ as $Y f(\vec{x}) = x_2 f(\vec{x})$. We first prove, using exponential decay of $P_j$ (proved in Appendix \ref{sec:pj-props}), that with appropriate domains these operators are well-defined and self-adjoint in $L^2(\field{R}^2)$. We then prove, using decay induced by $Y$ combined with the fact that functions in $P_j$ decay in $X$, that each $P_j Y P_j$ has compact resolvent and hence purely real and discrete spectrum. We finally prove that the eigenfunctions of $P_j Y P_j$ exponentially decay by a direct calculation from the eigenequation $P_j Y P_j f = \lambda f$, again using exponential decay of $P_j$. It now follows immediately that the set of eigenfunctions of each of the $P_j Y P_j$ operators forms an exponentially localized basis of $\range(P)$.

We are now in a position to present the main result of this paper.
\begin{theorem}[Main Theorem] \label{th:main_theorem}
Let $H$ be the two dimensional continuum Hamiltonian \eqref{eq:Ham}, where the potentials $A$ and $V$ satisfy certain regularity conditions (Assumption \ref{as:H_assump}). Suppose that $H$ has a spectral gap containing the Fermi level (Assumption \ref{as:gap_assump}), and let $P$ be the Fermi projection. Let $X$ and $Y$ denote position operators $X f(\vec{x}) = x_1 f(\vec{x})$ and $Y f(\vec{x}) = x_2 f(\vec{x})$ with respect to a choice of non-parallel two-dimensional axes. Then, if $PXP$ has uniform spectral gaps (Assumption \ref{def:usg}), there exist functions $\{ \psi_{j,m} \}_{(j,m) \in \cJ \times \mathcal{M}}$ and points $\{ (a_j,b_m) \}_{(j, m) \in \cJ \times \mathcal{M}} \in \field{R}^2$ such that
  \begin{enumerate}
  \item The collection $\{ \psi_{j,m} \}_{(j,m) \in \cJ \times \mathcal{M}}$ is an orthonormal basis of $\range(P)$. 
  \item Each $\psi_{j,m}$ is exponentially localized at $(a_{j,m},b_{j,m})$ in the sense that
    \begin{equation}
      \int_{\field{R}^2} e^{2\gamma \sqrt{1 + (x_1-a_{j,m})^2 + (x_2-b_{j,m})^2}} \vert\psi_{j,m}(\vec{x})\vert^2 \,\emph{d}\vec{x} \leq C,
      \label{eq:exp-loc}
    \end{equation}
  where $(C,\gamma)$ are finite, positive constants which are independent of $j$ and $m$.
  \item The set of $\{ \psi_{j,m} \}_{(j,m) \in \cJ \times \mathcal{M}}$ are the set of eigenfunctions of the operators $P_j Y P_j \rvert_{\range(P_j)}$, where $P_j$ are the band projectors defined by Definition \ref{def:band-projectors}.
\end{enumerate}
Here $\cJ$ and $\mathcal{M}$ are the countable sets which index the projectors $P_j$, and the eigenfunctions of $P_j Y P_j$ for fixed $j$, respectively. 
\end{theorem}
We break up the proof of our main theorem into a series of lemmas to be proved, and make our assumptions precise, in Section \ref{sec:results}. We prove these lemmas in Sections \ref{sec:self-adjointness}, \ref{sec:pjypj-disc-spec}, and \ref{sec:pjypj-exp-loc}, using estimates on $P_j$ proved in Appendices \ref{sec:pj-props} and \ref{sec:shifting-proof}.

The generalization of our main theorem to discrete models is relatively simple because of the operator-theoretic structure of our proof. As long as the entries of the discrete Hamiltonian decay exponentially away from the diagonal, we can establish identical operator bounds on the Fermi projector $P$ as in the continuum case, and the rest of the proof goes through without modification; see Section \ref{sec:discrete}.

In finite systems with periodic or Dirichlet boundary conditions, the operators $P_j Y P_j$ which appear in the proof of our main theorem can be constructed and diagonalized numerically. Hence our theorem and its proof imply a simple algorithm for generating a generalized Wannier basis for a given $H$. Note that in our main theorem we allow for a possibly non-zero magnetic potential $A$ so that we can justify applying our construction to the Haldane model even with complex hopping. We test the effectiveness of this algorithm on the Haldane model in Section \ref{sec:numerical-results}. 

A sketch of the generalization of our results to three dimensions is as follows. Consider position operators $X$, $Y$, and $Z$ associated with a three-dimensional basis acting on $\mathcal{H}: = L^2(\field{R}^3)$, and consider the operator $PXP$. Assume $PXP$ has uniform spectral gaps, and let $P_j$ denote spectral projections onto each of the separated components of the spectrum of $PXP$. Now assume the operators $P_j Y P_j$ \emph{also} have uniform spectral gaps, and let $P_{j,k}$ denote spectral projections onto each of the separated components of the spectrum of $P_j Y P_j$. By analogous reasoning to the two dimensional case, functions in $\range(P_{j,k})$ are quasi-one dimensional. We therefore claim that the set of eigenfunctions of the operator $P_{j,k} Z P_{j,k}$ will form an exponentially localized basis of $\range(P_{j,k})$ for each $j, k$, and that the union of all of these eigenfunctions over $j$ and $k$ will form an exponentially-localized basis of $\range(P)$.  
For more detail, see Section \ref{sec:higher-d}.

The following result addresses the special case where the potentials $A$ and $V$ are periodic.
\begin{theorem} \label{th:periodic_case}
Assume the conditions of Theorem \ref{th:main_theorem}, that $A$ and $V$ are periodic with respect to a Bravais lattice $\Lambda$, that the position operators $X$ and $Y$ are chosen with respect to the basis vectors of the Bravais lattice (Assumption \ref{as:XY_periodic}), and that the $P_j$ are defined in a way which respects the lattice (Assumption \ref{as:Pj_periodic}). Then the index sets can be taken as $\mathcal{J} = \mathcal{M} = \field{Z}$, and the sets of functions $\{ \psi_{j,m} \}_{(j,m) \in \field{Z} \times \field{Z}}$, and centers $\{ (a_{j,m},b_{j,m}) \}_{(j,m) \in \field{Z} \times \field{Z}}$, whose existence is guaranteed by Theorem \ref{th:main_theorem}, can be chosen to be closed under translations in the lattice $\Lambda$. 
\end{theorem}
We prove Theorem \ref{th:periodic_case} in Appendix \ref{sec:gWFs_closed}; its generalizations to discrete, finite, and higher-dimensional models are clear. Theorem \ref{th:periodic_case} proves that, for periodic models, the uniform spectral gaps assumption on $PXP$ implies that $P$ has an exponentially localized Wannier basis. By the usual Fourier transform construction (see, for example, (4) of \cite{2012MarzariMostofiYatesSouzaVanderbilt}), Theorem \ref{th:periodic_case} implies the existence of an analytic and periodic Bloch frame for the Fermi projection, and hence vanishing of the Chern number \cite{2007Panati,2018MonacoPanatiPisanteTeufel}. In Appendix \ref{sec:chern-app}, we show how the uniform spectral gaps assumption allows for an explicit construction via parallel transport of an analytic and periodic Bloch frame in two dimensions.

The following is simple to prove (it follows immediately from the Riesz projection formula), but has important consequences.
\begin{theorem} \label{th:symmetries_preserved}
    Let $\mathcal{T}$ denote an anti-unitary operator $L^2(\field{R}^2) \rightarrow L^2(\field{R}^2)$ which squares to $1$ or $-1$. Assume that $\mathcal{T}$ commutes with $H$, $X$, and $Y$, and that $PXP$ has uniform spectral gaps so that the projectors $P_j$ are well-defined. Then $\mathcal{T}$ commutes with $P$, $PXP$, and, for each $j \in \mathcal{J}$, $P_j$ and $P_j Y P_j$.
\end{theorem}
The first important implication of Theorem \ref{th:symmetries_preserved} is: whenever the Hamiltonian is real, so are the operators $P_j Y P_j$ for every $j \in \mathcal{J}$, and hence (up to multiplying by a complex constant with modulus $1$) the generalized Wannier functions we construct are real. The second is that whenever the Hamiltonian commutes with a Fermionic time-reversal symmetry operator $\mathcal{T}$ (an anti-unitary operator which squares to $-1$), the set of eigenfunctions of each $P_j Y P_j$ for every $j \in \mathcal{J}$ is closed under $\mathcal{T}$. It follows immediately that, at least for periodic models in two dimensions, whenever $PXP$ has uniform spectral gaps, the $\field{Z}_2$ index must vanish. If it did not, our construction would yield exponentially localized Wannier functions respectful of the Fermionic time reversal symmetry, a contradiction \cite{2006FuKane,2017CorneanMonacoTeufel}.



\subsection{Outlook} \label{sec:outlook}

In this work we find that gap conditions on projected position operators ($PXP$ in two dimensions, $PXP$ and $P_j Y P_j$ in three dimensions, and so on) are sufficient conditions for constructing exponentially localized Wannier functions and generalized Wannier functions. This finding provides a new perspective, which does not rely fundamentally on periodicity, from which to study topological materials (although note \cite{1991Niu}, which emphasizes the role of eigenfunctions of $PXP$ in the quantum Hall effect). Our findings can be considered the generalization of ideas of Vanderbilt, Soluyanov, and co-authors, who have emphasized the importance of the spectrum of the operator $PXP$ for understanding topological properties in the periodic setting \cite{2011SoluyanovVanderbilt,2011SoluyanovVanderbilt2,2014TaherinejadGarrityVanderbilt,2017GreschYazyevTroyerVanderbiltBernevigSoluyanov,2018WuZhangSongTroyerSoluyanov}, to the non-periodic setting. 

Our work poses two immediate open questions: 
\begin{enumerate}[label=(\arabic*)]
\item Do general hypotheses exist, either for periodic or non-periodic materials, which guarantee that projected position operators have gaps?
\item How do these hypotheses relate to the topological obstructions which play such an important role in the periodic case? 
\end{enumerate}
Here we collect some observations relevant to these questions. We hope to give more complete answers in future work.

We start by noting that it is possible to construct relatively trivial models for which gap conditions on $PXP$ in two dimensions, and on $PXP$ and $P_j Y P_j$ for all $j$ in three dimensions, and so on, are guaranteed to hold. Start with a discrete periodic model, with two sites per unit cell, and set the onsite potentials equal to $\pm v$ for some non-zero and positive $v$. Set all other terms in the Hamiltonian equal to zero. The spectrum is then $\pm v$, and the associated eigenfunctions can be taken to be delta functions at each site. This can be understood as the ``atomic limit'', where the sites of the model are infinitely far from each other \cite{Bradlyn2017}. Let $P$ denote the projection onto the $-v$ eigenfunctions, and define $X, Y$ etc. to measure position relative to the lattice. The spectra of $PXP$, $P_j Y P_j$ etc. are now explicit: they all equal $\field{Z}$ (up to overall shifts and scaling factors), and hence all of these operators satisfy the uniform spectral gaps assumption. The generalized Wannier functions yielded by our construction in this case are simply deltas at each site with onsite potential $-v$. 

The models and Wannier functions constructed thus far are trivial. However, note that if we turn on inter-site hopping continuously, the spectrum of the operator $H$, the projector $P$, and the spectra of the projected position operators, will all change continuously as well. It follows that the uniform spectral gaps of the projected position operators will persist, even for non-zero (sufficiently small) inter-site hopping amplitudes. Note that this argument holds even if we allow for onsite potentials and/or hopping amplitudes which vary throughout the structure, so that in this way we can construct discrete models, both periodic and non-periodic, for which uniform spectral gaps assumptions are guaranteed to hold. We expect that the argument will also go through for continuous models in the ``strong-binding'' limit where tight-binding models describe the low-energy behavior (see, for example, \cite{Dimassi-Sjoestrand:99,Fefferman2018,Shapiro2020} and references therein).

The examples just constructed answer question (1) positively, albeit with hypotheses which are perhaps too strong to be generally useful. In our numerical experiments on the two dimensional Haldane model, the gap condition on $PXP$ appears to be roughly equivalent to vanishing of the Chern number for that model, even far from the atomic limit. However, the gap condition cannot be equivalent to vanishing of the Chern number in general, because Theorem \ref{th:symmetries_preserved} implies that (see also \cite{Stubbs2020} for numerical verification) $PXP$ cannot have gaps when the Kane-Mele model is in its topological ($\field{Z}_2$ index odd) phase, even though the Chern number vanishes in that case. It is therefore tempting to conjecture that $PXP$ has gaps whenever all topological obstructions, including those created by model symmetries such as time-reversal, vanish. Such a conjecture would still require care to state and prove, however, since the gap condition might not hold for totally arbitrary non-periodic non-topological systems, and/or be sensitive to the co-ordinate axes used to define the operator $X$. It is tempting to make similar conjectures in higher dimensions as well, with the same caveats. These conjectures are supported by the example of an insulator in the atomic limit, since this is precisely the paradigm of a topologically trivial insulator \cite{Bradlyn2017}.

We can gain further insight, at least in the periodic case, by revisiting the works of Vanderbilt and Soluyanov, so we now briefly recall their ideas. In periodic materials in $d$ dimensions, the operator $PXP$ can be decomposed using Bloch theory into operators $PXP(\vec{k})$ where $\vec{k}$ is a $(d-1)$-dimensional wave-vector. Topological data, such as Chern numbers, turn out to be encoded in winding numbers of the eigenvalues (band structure) of these operators. For more details, see \cite{2011SoluyanovVanderbilt,2011SoluyanovVanderbilt2,2014TaherinejadGarrityVanderbilt,2017GreschYazyevTroyerVanderbiltBernevigSoluyanov,2018WuZhangSongTroyerSoluyanov}. The associated eigenfunctions of these operators can be constructed by integrating Bloch functions in a particular gauge with respect to $(d-1)$ components of the wave-vector (as opposed to $d$ in the case of Wannier functions), resulting in ``hybrids'' between localized Wannier functions and extended Bloch functions, known as ``hybrid'', or ``hermaphrodite'', Wannier functions (see also \cite{2001SgiarovelloPeressiResta}).


Motivated by these works, we have investigated whether gap conditions can be directly (i.e. not through Theorem \ref{th:periodic_case}) related to vanishing of topological invariants. In Appendix \ref{sec:chern-app}, we give an alternative proof, using parallel transport, that, for two dimensional periodic materials, if $PXP$ has uniform spectral gaps, it is always possible to construct an analytic and periodic Bloch frame for the Fermi projection and hence the Chern number must vanish. In \cite{Stubbs2020} we give a similar direct proof that $PXP$ having uniform spectral gaps implies vanishing of the $\field{Z}_2$ index when Fermionic time-reversal symmetry holds and $P$ is the projection onto two bands (the minimum possible dimension of $\range P$ in the presence of Fermionic time-reversal symmetry). Ultimately, however, these calculations merely confirm that gap conditions imply vanishing of topological invariants, which is already clear from Theorems \ref{th:periodic_case} and \ref{th:symmetries_preserved}, and do not address questions (1) and (2).


We can also gain insight by slightly generalizing our construction. Suppose that, for some model and choice of position operator $X$, $PXP$ does not have gaps so that our construction fails. Then, we can to try to perturb $X$ to another operator $\widehat{X}$ so that $P \widehat{X} P$ does have gaps, define band projectors $P_j$ onto the spectral islands of $P \widehat{X} P$, and then obtain generalized Wannier functions by diagonalizing the operators $P_j Y P_j$. It turns out to be straightforward to make this proposal rigorous under appropriate technical assumptions on $\widehat{X}$; see Appendix \ref{sec:extension-to-other} for details.
We show the power of this more general perspective in \cite{Stubbs2020}, where we find that by constructing an $\widehat{X}$ which explicitly breaks Fermionic time-reversal symmetry, while satisfying the technical assumptions, we obtain an operator $P \widehat{X} P$ with gaps, from which we can obtain exponentially localized generalized Wannier functions, even in the topological phase of the Kane-Mele model. 

In \cite{Lu2021}, two of the authors have demonstrated the principle more abstractly by proving that, whenever $\range(P)$ has an orthonormal basis which decays at a sufficiently fast algebraic rate, there exists an operator $\widehat{X}$, defined in terms of the algebraically decaying orthonormal basis, and satisfying the necessary technical assumptions, such that $P \widehat{X} P$ has gaps. 
These results suggest that questions (1) and (2) may be easier to answer if we relax the condition that projected position operators have gaps to the condition that operators ``close'' to projected position operators have gaps. 
Indeed, if we take this perspective, \cite{Lu2021} answers question (1) in the affirmative, with the necessary hypothesis being existence of a sufficiently rapidly decaying orthonormal basis of $\range(P)$. This result is far from ideal, however, because existence of such a basis is not an easily verified condition in the non-periodic case. 

\subsection{Previous Works on generalized Wannier functions}
\label{sec:previous_works}
Before continuing to our numerical results, we pause to discuss existing literature on generalized Wannier functions, other than the works of Kivelson and Nenciu-Nenciu we have already mentioned \cite{1982Kivelson,1998NenciuNenciu}. 
Before Nenciu-Nenciu's proof of exponential decay, Niu \cite{1991Niu} showed that eigenfunctions of $PXP$ would decay faster than any polynomial power. 
Geller and Kohn have studied generalized Wannier functions in ``nearly periodic'' materials \cite{1993GellerKohn,1993GellerKohn_2}. Nenciu and Nenciu have proved existence of generalized  Wannier functions for materials whose atomic potential is related to that of a crystal with exponentially-localized Wannier functions via an interpolation which does not close the spectral gap at the Fermi level \cite{1993NenciuNenciu}. 

More recent work by Cornean, Nenciu, and Nenciu \cite{2008CorneanNenciuNenciu} showed that the one-dimensional result of Nenciu-Nenciu \cite{1998NenciuNenciu} can be generalized to higher dimensions where $H = - \Delta + V$ and the potential $V$ is concentrated along a single axis. Hastings and Loring have 
studied Wannier functions in two dimensions 
defined as ``simultaneous approximate eigenvectors'' of the operators $PXP$ and $PYP$ \cite{2010HastingsLoring}. E and Lu proved existence and exponential localization of Wannier functions in smoothly deformed crystals in the limit where the deformation length-scale tends to infinity \cite{2011WeinanLu}. Prodan \cite{2015Prodan} showed that by diagonalizing $P e^{- R} P$ (where $R$ denotes the radial position operator $\sqrt{ X^2 + Y^2 }$) one can construct an orthogonal basis of functions which are concentrated on spherical shells in arbitrary dimension. Cornean, Herbst, and Nenciu \cite{2016CorneanHerbstNenciu} have proved that the existence and exponential localization of generalized Wannier functions is stable under perturbation by a weak (but not spatially decaying) magnetic field. In the special case where the unperturbed system is periodic and the perturbing magnetic field is constant, they prove further that the set of generalized Wannier functions remains closed under translations in the Bravais lattice.

In the presence of topological obstructions, which prevent existence of exponentially localized Wannier functions, exponentially localized overcomplete bases known as Parseval frames may nonetheless exist, see \cite{Kuchment2009,Auckly2018,2019CorneanMonacoMoscolari}.

\comment{\subsection{Connections with hybrid Wannier functions}
  \label{sec:hybrid_wannier}
  The present work can be thought of an extension of previous works on the `hermaphrodite' or `hybrid' Wannier functions which were first described in \cite{2001SgiarovelloPeressiResta}. For periodic systems, to generate a set of hybrid Wannier functions we first choose a single spatial direction and then take the Bloch-Floquet transform of the Bloch functions in that direction which a specific choice of gauge. This choice gauge is chosen so that the spread in the $x$ direction in minimized. Due to the work of Kohn, Kivelson, and Nenciu-Nenciu \cite{1959Kohn,1982Kivelson,1998NenciuNenciu}, the optimal choice of gauge is given by diagonalizing the projected position operator $PXP$. Because of this gauge choice, the hybrid Wannier functions are localized (Wannier-like) in the chosen spatial direction, and delocalized (Bloch-like) in the remaining directions. \\
  \indent The hybrid Wannier functions have been used in previous works to study the topological invariants of different systems. Starting with two papers by Soluyanov and Vanderbilt \cite{2011SoluyanovVanderbilt,2011SoluyanovVanderbilt2} there have been a series of works which use the expected position of the hybrid Wannier functions, which they refer to as the `Wannier charge centers', to numerically determine whether a system has non-trivial $\field{Z}_2$ invariant. This work was later expanded upon in \cite{2014TaherinejadGarrityVanderbilt} where they numerically determine if material is a trivial, Chern, weak, or strong topological insulator using the Wannier charge centers. More recently, the software packages \texttt{Z2Pack} \cite{2017GreschYazyevTroyerVanderbiltBernevigSoluyanov} and \texttt{WannierTools} \cite{2018WuZhangSongTroyerSoluyanov} have used the idea of Wannier charge centers as a tool to numerically identify topological materials. \\
\indent The connection between these numerical works and the present work is as follows. By construction, the hybrid Wannier functions are eigenfunctions of the operator $PXP$ \cite{2017GreschYazyevTroyerVanderbiltBernevigSoluyanov} and so the Wannier charge centers considered in previous works are exactly the eigenvalues of $PXP$. Since in two dimensions a well localized basis for the Fermi projection exists if and only if the Chern number is zero \cite{2018MonacoPanatiPisanteTeufel}, the main theorem in this work provides an alternate justification that Wannier charge centers are a sufficient way to determine if a material is a Chern insulator in two dimensions. Furthermore, the main theorem extends upon these results since it holds for all systems (periodic and non-periodic, time reversal symmetric and non-time reversal symmetric). The relationship between $PXP$ having uniform spectral gaps and the $\field{Z}_2$-invariant is the subject of ongoing work.}

\subsection*{Acknowledgements}
A.B.W. would like to thank Guillaume Bal, Christoph Sparber, and Jacob Shapiro for stimulating discussions, Michel Fruchart for pointing out the connection with hybrid Wannier functions, and Terry A. Loring for helpful comments on an early version of this manuscript. We would also like to thank the anonymous reviewers whose comments significantly improved this manuscript.

\subsection*{Declarations}
\subsubsection*{Funding}
This work is supported in part by the National Science Foundation via grant DMS-1454939 and the Department of Energy via grant DE-SC0019449. K.D.S.~is also supported in part by a National Science Foundation Graduate Research Fellowship under Grant No.~DGE-1644868.

\subsubsection*{Conflicts of interest}
We have no conflicts of interest to declare.

\subsubsection*{Availability of data and material}
Not applicable.

\subsubsection*{Code availability}
The Python code used to generate the figures for this project is available at \url{https://github.com/kstub/pxp-wannier}.

\section{Numerical Results}
\label{sec:numerical-results}

In this section we present results of implementing the numerical scheme suggested by our main theorem for generating exponentially localized generalized Wannier functions. We have further extended this numerical scheme, for example modifying the scheme to deal with different boundary conditions, and to obtain Wannier functions which respect model symmetries such as time-reversal, in \cite{Stubbs2020}. The scheme is as follows:
\begin{enumerate}
\item Choose position operators $X$ and $Y$ acting in non-parallel directions.
\item Compute the Fermi projector $P$ by diagonalizing the Hamiltonian $H$.
\item Diagonalize the operator $PXP$, and inspect $\sigma(PXP)$ for clusters of eigenvalues separated from other eigenvalues by spectral gaps.
\item Form band projectors $P_j$ onto each cluster of eigenvalues.
\item Diagonalize the operators $P_j Y P_j$ to obtain exponentially localized eigenvectors which span the Fermi projection.
\end{enumerate}
Since numerically one can only deal with a finite system, it is necessary to clarify two points compared with the infinite case.

First, note that any vector in a finite system is trivially exponentially-decaying by taking $C > 0$ sufficiently large and $\gamma > 0$ sufficiently small in \eqref{eq:exp-loc}. It is necessary to clarify, therefore, that the algorithm presented above yields exponentially-decaying eigenvectors with $C > 0$ and $\gamma > 0$ which are \emph{independent of system size}. In this sense, our algorithm yields a non-trivial result.

Second, in finite systems, all operators have purely discrete spectrum and hence (generically) there will be a spectral gap between \emph{any} pair of eigenvalues. However, to obtain localized eigenvectors it is not enough to simply form band projectors for each eigenvalue of $PXP$ alone. Hence the clarification in the algorithm that we must form band projectors from clusters of nearby eigenvalues separated from the remainder of the spectrum by clear spectral gaps. This point is clarified by our rigorous analysis in the following sections, where we show that the localization of the generalized Wannier functions produced by our scheme is related to the minimal gap between the bands of $\sigma(PXP)$ (see Section \ref{sec:pjypj-exp-loc}).

We choose to test our scheme on the Haldane model \cite{1988Haldane} at half-filling, a simple two-dimensional model whose Fermi projection, in the crystalline setting, may or may not have non-zero Chern number depending on model parameters. For this reason, the Haldane model is a natural model for investigating the connection between gaps of $PXP$ and topological triviality of $P$ in the case where the material is periodic. Historically, the Haldane model was the first model of a Chern insulator: a material exhibiting quantized Hall response without net magnetic flux through the material. We now briefly recap the essential features of this model.

\subsection{The Haldane Model}
\label{sec:haldane-def}
The Haldane model describes electrons in the tight-binding limit hopping on a honeycomb lattice. In addition to real nearest-neighbor hopping terms, the model allows for a real on-site potential difference between the $A$ and $B$ sites of the lattice, and for \emph{complex} next-nearest-neighbor hopping terms which break time-reversal symmetry without introducing net magnetic flux. 

In the crystalline case, the action of the Haldane tight-binding Hamiltonian acting on wave-functions $\psi \in \mathcal{H} := \ell^2(\field{Z}^2;\field{C}^2)$ is: 
\begin{equation}
\begin{split} \label{eq:Haldane_H}
  \begin{bmatrix}
    \left( H \psi \right)_{m,n}^A \\[1ex]
    \left( H \psi \right)_{m,n}^B 
  \end{bmatrix}
  & =
    v
    \begin{bmatrix}
      \psi_{m,n}^A \\[1ex]
      -\psi_{m,n}^B
    \end{bmatrix}
    +
    t
  \begin{bmatrix}
    \psi_{m,n}^B + \psi_{m,n-1}^B + \psi_{m-1,n}^B \\[1ex]
    \psi_{m,n}^A + \psi_{m+1,n}^A + \psi_{m,n+1}^A 
  \end{bmatrix} \\[1ex]
  & \qquad + t' e^{i\phi} 
  \begin{bmatrix}
    \psi_{m,n+1}^A + \psi_{m-1,n}^A + \psi_{m+1,n-1}^A \\[1ex]
    \psi_{m,n-1}^B + \psi_{m+1,n}^B + \psi_{m-1,n+1}^B 
  \end{bmatrix} \\
  & \qquad + t' e^{-i\phi} 
  \begin{bmatrix}
    \psi_{m,n-1}^A + \psi_{m+1,n}^A + \psi_{m-1,n+1}^A \\[1ex]
    \psi_{m,n+1}^B + \psi_{m-1,n}^B + \psi_{m+1,n-1}^B
  \end{bmatrix}.
\end{split}
\end{equation}
Here, $t, v, t'$, and $\phi$ are real parameters expressing the magnitude of nearest-neighbor hopping, the magnitude of on-site potential difference, the magnitude of complex next-nearest neighbor hopping, and the complex argument of the next-nearest neighbor hopping, respectively. 

By definition, at half-filling the Fermi level is at $0$. An explicit calculation using Bloch theory \cite{1988Haldane} (see also \cite{2013FruchartCarpentier}) shows that $H$ has a spectral gap (and hence describes an insulator) at $0$ whenever 
\begin{equation*}
    v \neq \pm 3 \sqrt{3} t' \sin \phi.
\end{equation*}
Further calculation shows that the Fermi projection has a non-trivial Chern number (equal to $1$ or $-1$) whenever 
\begin{equation} \label{eq:phase}
    |v| < 3 \sqrt{3} |t' \sin \phi|.
\end{equation}
In this case, exponentially-localized Wannier functions do not exist \cite{2018MonacoPanatiPisanteTeufel,2019MarcelliMonacoMoscolariPanati}. Whenever the parameters $t, v, t', \phi$ are such that \eqref{eq:phase} holds, we say the Haldane model is in its \emph{topological phase}.

For some of our experiments, we add a perturbation to the Hamiltonian \eqref{eq:Haldane_H} which models disorder. We replace the on-site potential $v$ in \eqref{eq:Haldane_H} by a spatially varying on-site potential $v + \eta(m,n)$, where $\eta(m,n)$ is drawn for each $m, n$ from independent Gaussian distributions with mean $0$ and variance $\sigma^2$:
\begin{equation} \label{eq:disorder}
    \eta(m,n) \sim \mathcal{N}(0,\sigma^2) \text{ for each $m,n$.}
\end{equation}
We refer to this kind of disorder as ``onsite'' disorder. Assuming $H$ \eqref{eq:Haldane_H} has a spectral gap with $\sigma^2 = 0$ (i.e. without disorder), then for sufficiently small $\sigma^2$, the spectral gap will persist almost surely and our method can be applied. 

To implement our method, we have to make a choice of position operators on the space $\ell^2(\field{Z}^2;\field{C}^2)$. The simplest choice is to define $X$ and $Y$ consistently with the crystal lattice by: 
\[
\begin{bmatrix}
\left( X \psi \right)^A_{m,n} \\[1ex]
\left( X \psi \right)^B_{m,n} 
\end{bmatrix}
= \begin{bmatrix} 
m \psi_{m,n}^A \\[1ex]
m \psi_{m,n}^B
\end{bmatrix}
\quad 
\begin{bmatrix}
\left( Y \psi \right)^A_{m,n} \\[1ex]
\left( Y \psi \right)^B_{m,n} 
\end{bmatrix}
= \begin{bmatrix} 
n \psi_{m,n}^A \\[1ex]
n \psi_{m,n}^B
\end{bmatrix}.
\]
We refer to this choice of $X$ and $Y$ as the {standard position operators}. A couple of remarks are in order. First, note that $X$ and $Y$ do not distinguish between $A$ and $B$ sites. Second, the crystal lattice vectors are not orthogonal hence eigenvalues of $X$ and $Y$ do not represent co-ordinates with respect to orthogonal axes. Since the lattice vectors are linearly independent our method can nonetheless be applied.

\subsection{Parameters for numerical tests and further remarks} For our numerical tests, we consider the Haldane model just described truncated to a $24 \times 24$ lattice (hence $\mathcal{H}$ has dimension $2 \times 24 \times 24$) under the following conditions: 
\begin{itemize}
\item Periodic boundary conditions with standard position operators, without disorder. We consider parameter values such that the system is in a non-topological phase and values such that the system is in a topological phase (Sections \ref{sec:periodic_topological} and \ref{sec:periodic}).
\item Dirichlet boundary conditions with standard position operators, without disorder (Section \ref{sec:dirichlet}).
\item Dirichlet boundary conditions with standard position operators, with weak disorder which does not close the spectral gap of $H$ (Section \ref{sec:dirichlet_weak_disorder}).
\item Dirichlet boundary conditions with standard position operators, with strong disorder (Section \ref{sec:dirichlet_strong_disorder}). Note that in this case the spectral gap assumption on $H$ is no longer valid; though $P$ will still be exponentially localized due to Anderson localization. 
\item Dirichlet boundary conditions with non-standard (rotated) position operators, without disorder (Section \ref{sec:dirichlet_rotated}).
\end{itemize}
Note that we do not consider any examples with Dirichlet boundary conditions in the topological phase. This is because $H$ does not have a spectral gap in this case due to edge states.

In each case we will display plots of the generalized Wannier functions generated by our algorithm. Specifically, given a generalized Wannier functions $\psi \in \mathcal{H} = \ell^2(\{1,...,24\}^2;\field{C}^2)$, we will plot the following matrix in a 3D surface plot:
\begin{equation}
  \label{eq:plot-mat}
  \begin{bmatrix}
    \sqrt{ |\psi_{1,1}^A|^2 + |\psi_{1,1}^B|^2} & \sqrt{ |\psi_{1,2}^A|^2 + |\psi_{1,2}^B|^2} & \cdots & \sqrt{ |\psi_{1,24}^A|^2 + |\psi_{1,24}^B|^2} \\[2ex]
    \sqrt{ |\psi_{2,1}^A|^2 + |\psi_{2,1}^B|^2} & \sqrt{ |\psi_{2,2}^A|^2 + |\psi_{2,2}^B|^2} & \cdots & \sqrt{ |\psi_{2,24}^A|^2 + |\psi_{2,24}^B|^2} \\[2ex]
    \vdots & \vdots & \ddots & \vdots \\[2ex]
    \sqrt{ |\psi_{24,1}^A|^2 + |\psi_{24,1}^B|^2}  & \sqrt{ |\psi_{24,2}^A|^2 + |\psi_{24,2}^B|^2} & \cdots & \sqrt{ |\psi_{24,24}^A|^2 + |\psi_{24,24}^B|^2} \\[2ex]
  \end{bmatrix}.
\end{equation}
To make the exponential decay of $\psi$ as clear as possible, we will also show 2D plots of the elementwise logarithm of this matrix.

We remark that while our theoretical results hold equally well in the both periodic and non-periodic cases for infinite systems, we find for finite systems our algorithm works better for systems with Dirichlet boundary conditions. This is not entirely surprising given that the position operators $X$ and $Y$ do not respect periodic boundary conditions. We have explored using different position operators, inspired by \cite{1998Resta,2000Zak,2019Aragaoetal}, in \cite{Stubbs2020}.


\subsection{Periodic Boundary Conditions, Topological versus Non-Topological} \label{sec:periodic_topological}
We have seen that for an infinite periodic system, whenever $PXP$ has uniform spectral gaps, the Chern number must vanish (see Theorems \ref{th:periodic_case} and \ref{th:Chern_thm}). In this section we provide numerical evidence that the uniform spectral gap assumption is actually equivalent to the Chern number vanishing in the case of the Haldane model by forming the Fermi projector $P$ from the Haldane Hamiltonian with periodic boundary conditions and numerically computing the spectrum of $PXP$ for different values of the Haldane model parameters. Our results are shown in Figure \ref{fig:periodic-pxp-topo-vs-non-topo}. For additional insight into these figures, see the proof of Theorem \ref{th:Chern_thm} and, in particular, Figure \ref{fig:charge_centers}.


We find that for model parameters such that the model is in a non-topological phase, $\sigma(PXP)$ shows clear gaps. For model parameters such that the model is in a topological phase, every gap of $\sigma(PXP)$ closes. This conclusion holds even when we choose model parameters such that the spectral gap of $H$ is approximately equal in either case ($\approx 2$).

Note that in the case where every gap of $\sigma(PXP)$ closes, our construction is technically well defined since the spectrum of $PXP$ is bounded on a finite domain. On the other hand, it is totally ineffective because we can only define one band projector $P_j$, which equals $P$. Hence the eigenfunctions of $P_j Y P_j$ in this case are the eigenfunctions of $P Y P$, which do not decay in $x$.

\begin{figure}
  \centering
  \includegraphics[width=.8\linewidth]{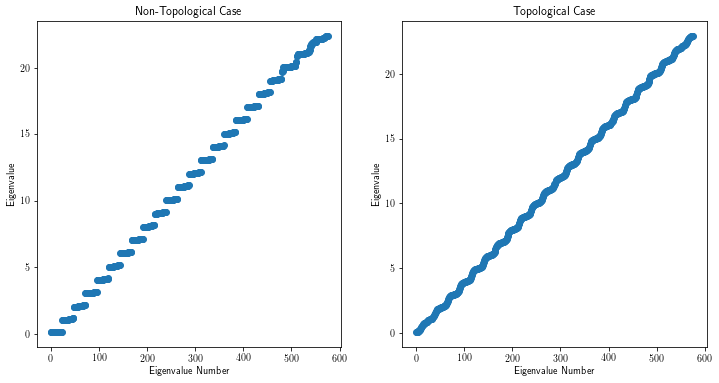}
  \caption{Plot of sorted non-zero eigenvalues of $PXP$ where $P$ is the Fermi projection and $X$ is the lattice position operator for the Haldane model on $24 \times 24$ system with periodic boundary conditions. The left plot corresponds to parameters $(t,t',v,\phi) = (1,0,1,\frac{\pi}{2})$ (non-topological phase) and the right plot corresponds to parameters $(t,t',v,\phi) = (1,\frac{1}{4},0,\frac{\pi}{2})$ (topological phase). The gap in $H$ for both non-topological and topological phase is $\approx 2$.} \label{fig:periodic-pxp-topo-vs-non-topo}
\end{figure}


\subsection{Periodic Boundary Conditions, Standard Position Operators} \label{sec:periodic}
In this section we implement our algorithm in the non-topological phase of Haldane with periodic boundary conditions, when $\sigma(PXP)$ shows clear gaps (Figure \ref{fig:periodic-pxp-topo-vs-non-topo}). Note that when we take periodic boundary conditions the last three bands of $PXP$ appear to merge together. We conjecture that this behavior is because the operator $X$ does not respect translation symmetry with respect to $x$.

Despite this, our theory still applies since we can enclose the last three bands by a single contour when we define the collection $\{ P_j \}_{j \in \cJ}$. For all bands but the last one, we find the eigenfunctions of $P_j Y P_j$ are exponentially localized like before. These results are shown in Figure \ref{fig:periodic_efuncs_notlast}. For the last band, we find that instead of the eigenfunctions of $P_j Y P_j$ being localized along a single line $x = c$ for constant $c$, they are somewhat spread across an interval of $x$ values: see Figure \ref{fig:periodic_efuncs_last}.

\begin{figure}
  \centering
  \includegraphics[width=.3\linewidth]{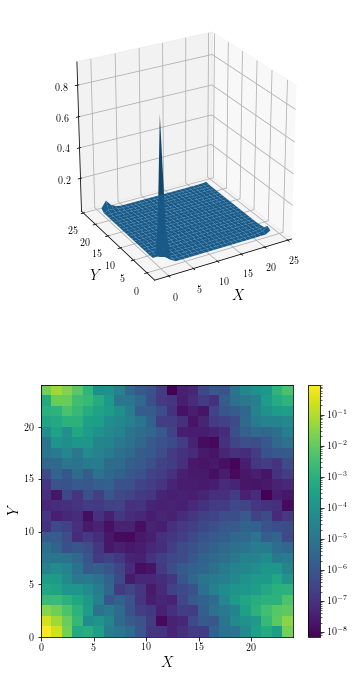}
  \includegraphics[width=.3\linewidth]{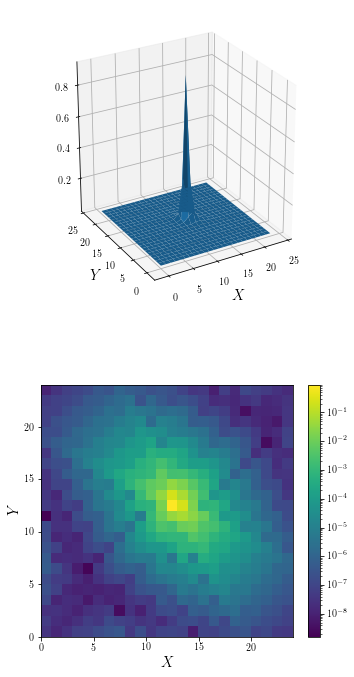}
  \includegraphics[width=.3\linewidth]{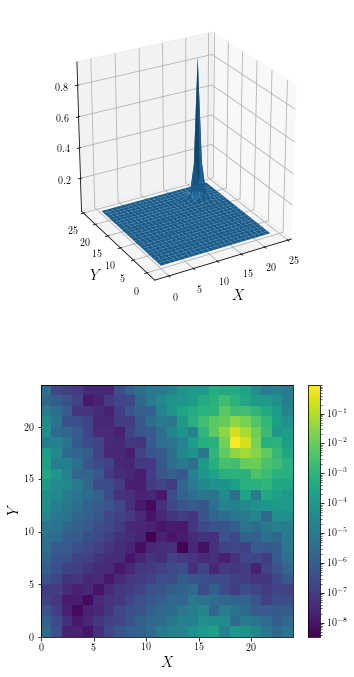} 
  \caption{Plot of eigenfunctions of the operator $P_j Y P_j$ for different values of $j$ where $\{ P_j \}_{j \in \cJ}$ are the band projectors for $PXP$, $P$ is the Fermi projection, and $X$, $Y$ are the lattice position operators. The projection $P$ comes from the Haldane model on $24 \times 24$ system with periodic boundary conditions. Parameters chosen are $(t,t',v,\phi) = (1,0,1,\frac{\pi}{2})$. Top row is a 3D surface plot of the matrix from Equation \eqref{eq:plot-mat}, bottom row is 2D log plot of the top row. For these figures we avoid the $P_j$ where a few bands of the spectrum of $PXP$ have clumped together (see Figure \ref{fig:periodic-pxp-topo-vs-non-topo}).} \label{fig:periodic_efuncs_notlast}
\end{figure}

\begin{figure}
  \centering
  \includegraphics[width=.3\linewidth]{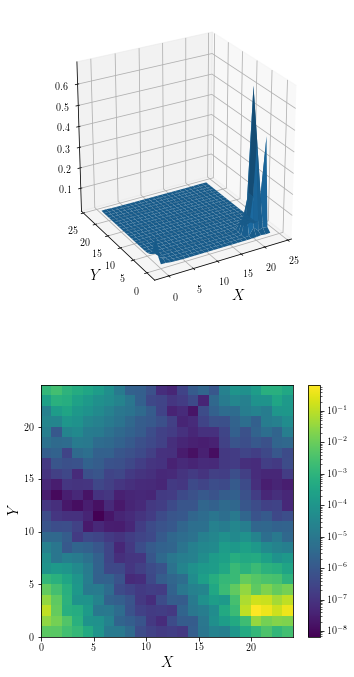}
  \includegraphics[width=.3\linewidth]{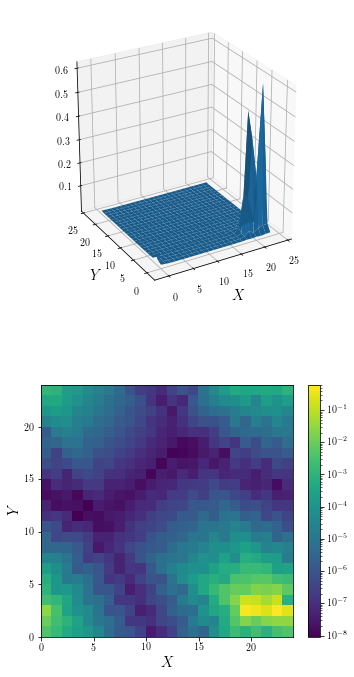} 
  \caption{Plot of eigenfunctions of the operator $P_j Y P_j$ for different values of $j$ where $\{ P_j \}_{j \in \cJ}$ are the band projectors for $PXP$, $P$ is the Fermi projection, and $X$, $Y$ are the lattice position operators. The projection $P$ comes from the Haldane model on $24 \times 24$ system with Dirichlet boundary conditions. Parameters chosen are $(t,t',v,\phi) = (1,0,1,\frac{\pi}{2})$. Top row is a 3D surface plot of the matrix from Equation \eqref{eq:plot-mat}, bottom row is 2D log plot of the top row. For these figures we consider the $P_j$ where a few bands of the spectrum of $PXP$ have clumped together (see Figure \ref{fig:periodic-pxp-topo-vs-non-topo}). Note that the generalized Wannier function generated by our method in this case has a relatively large spread in $x$ relative to those plotted in Figure \ref{fig:periodic_efuncs_notlast}.} \label{fig:periodic_efuncs_last}
\end{figure}


\subsection{Dirichlet Boundary Conditions using Standard Position Operators} \label{sec:dirichlet}
We consider the Haldane model with Dirichlet boundary conditions and parameters $(t,t',v,\phi) = (1,\frac{1}{10},1,\frac{\pi}{2})$, which correspond to the non-topological phase. For this choice of parameters the Hamiltonian $H$ has a gap of $\sim 1.006$. We plot the eigenvalues of $PXP$ in Figure \ref{fig:dirichlet-pxp-evals}, where we see $\sigma(PXP)$ shows clear gaps. We plot the eigenvectors of $PXP$ in Figure \ref{fig:dirichlet-pxp-efuncs}. We see that these eigenvectors are concentrated along lines $x = c$ for constants $c$. We finally plot the eigenfunctions of $P_j Y P_j$, which are localized with respect to $x$ and $y$, for a few different values of $j$ in Figure \ref{fig:dirichlet-pjypj-efuncs}.

\begin{figure}
  \centering
  \includegraphics[width=.8\linewidth]{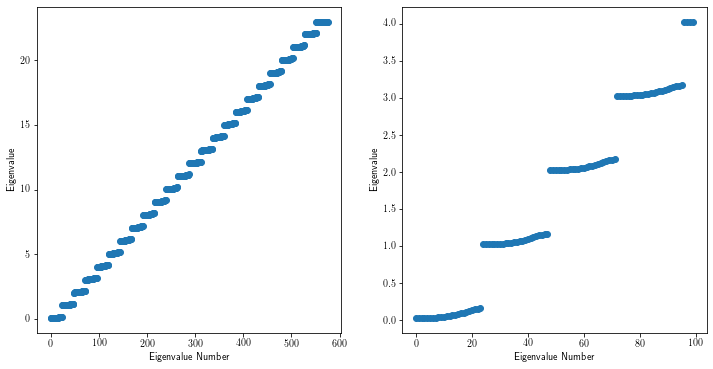}
  \caption{Plot of sorted non-zero eigenvalues of $PXP$ where $P$ is the Fermi projection and $X$ is the lattice position operator for the Haldane model on $24 \times 24$ system with Dirichlet boundary conditions. Entire spectrum (left) and first 100 eigenvalues (right). Parameters chosen are $(t,t',v,\phi) = (1,\frac{1}{10},1,\frac{\pi}{2})$. The spectrum shows clear gaps.} \label{fig:dirichlet-pxp-evals}
\end{figure}

\begin{figure}
  \centering
  \includegraphics[width=.23\linewidth]{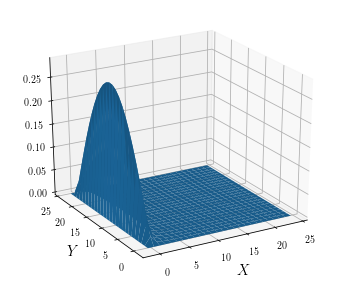} 
  \includegraphics[width=.23\linewidth]{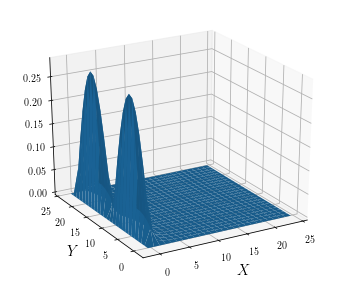} \includegraphics[width=.23\linewidth]{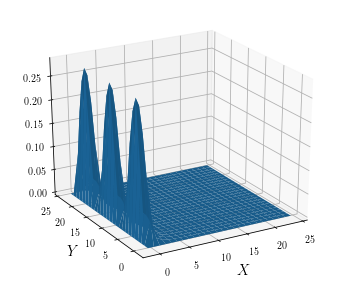} \\
  \includegraphics[width=.23\linewidth]{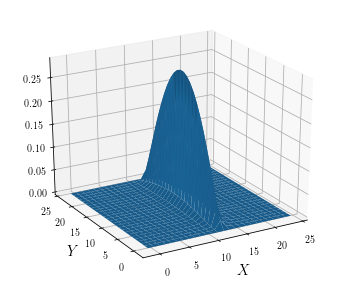} \includegraphics[width=.23\linewidth]{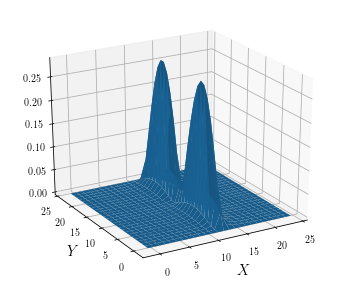} \includegraphics[width=.23\linewidth]{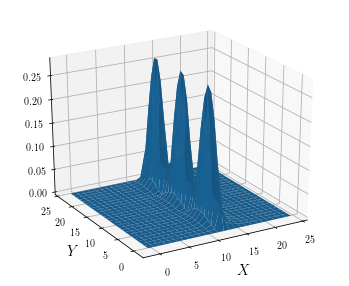} \\
  \includegraphics[width=.23\linewidth]{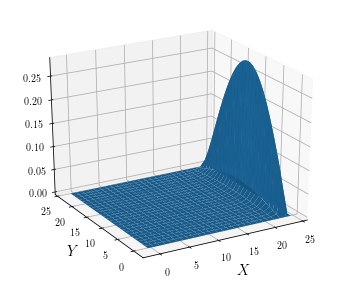} \includegraphics[width=.23\linewidth]{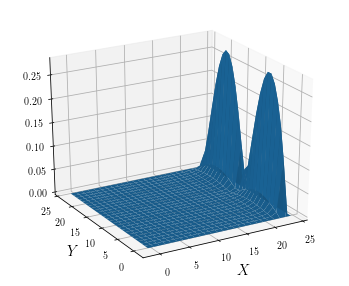} \includegraphics[width=.23\linewidth]{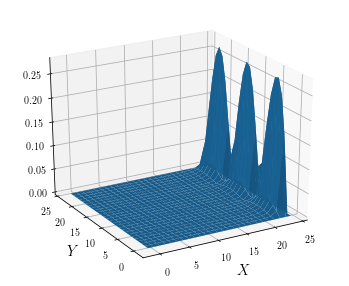} 
  \caption{Plot of eigenfunctions of the operator $PXP$ where $P$ is the Fermi projection and $X$ is the lattice position operator for the Haldane model on $24 \times 24$ system with Dirichlet boundary conditions. Parameters chosen are $(t,t',v,\phi) = (1,\frac{1}{10},1,\frac{\pi}{2})$. Each eigenvector of $PXP$ is localized along a line $x = c$ for some constant $c$.} \label{fig:dirichlet-pxp-efuncs}
\end{figure}

\begin{figure}
  \centering
  \includegraphics[width=.3\linewidth]{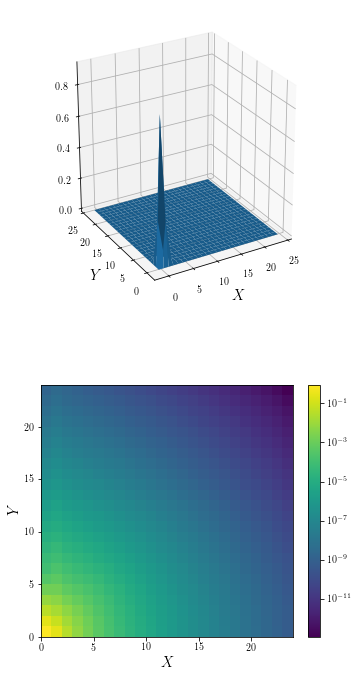}
  \includegraphics[width=.3\linewidth]{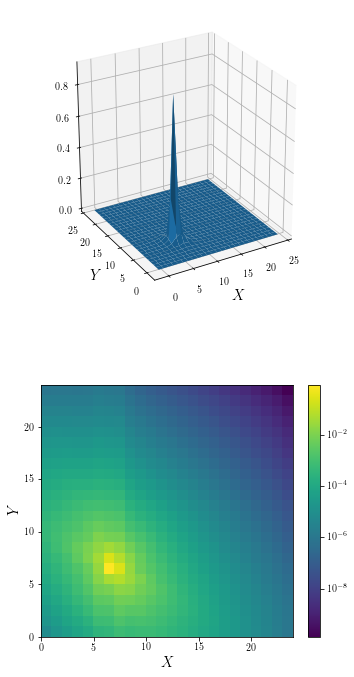} 
  \includegraphics[width=.3\linewidth]{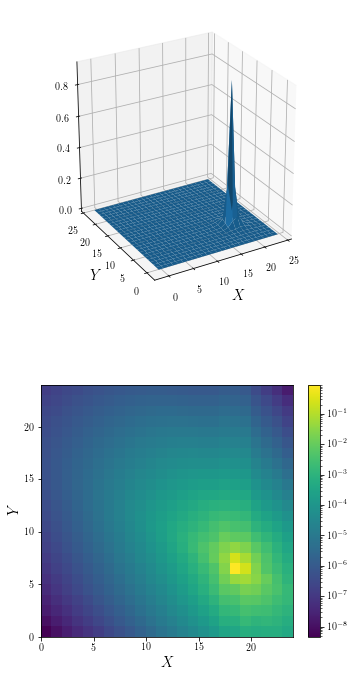} 
  \caption{Plot of eigenfunctions of the operator $P_j Y P_j$ for different values of $j$ where $\{ P_j \}_{j \in \cJ}$ are the band projectors for $PXP$, $P$ is the Fermi projection, and $X$, $Y$ are the lattice position operators. The projection $P$ comes from the Haldane model on $24 \times 24$ system with Dirichlet boundary conditions. Parameters chosen are $(t,t',v,\phi) = (1,\frac{1}{10},1,\frac{\pi}{2})$. Top row is a 3D surface plot of the matrix from Equation \eqref{eq:plot-mat}, bottom row is 2D log plot of the top row. Each eigenfunction shows clear exponential localization in line with our theoretical results.} \label{fig:dirichlet-pjypj-efuncs}
\end{figure}


\subsection{Dirichlet Boundary Conditions with Weak Disorder} \label{sec:dirichlet_weak_disorder}
We now consider a case where translational symmetry is broken even away from the edge of the material. Starting with the same parameters as in Section \ref{sec:dirichlet}, we add onsite disorder as in \eqref{eq:disorder}, with $\sigma^2 = \frac{1}{4}$. We plot results for a realization of the onsite disorder such that $H$ has a clear gap $\sim .253$. We find that the eigenvalues of $PXP$ show clear gaps despite the disorder, see Figure \ref{fig:dirichlet-disorder-pxp-evals}. We can therefore form projectors $P_j$, and the operators $P_j Y P_j$. We plot the eigenfunctions of $P_j Y P_j$ in Figure \ref{fig:dirichlet-disorder-pjypj-efuncs}. We observe that they are again exponentially localized, just as in the case without disorder (Figure \ref{fig:dirichlet-pjypj-efuncs}), in line with our theoretical results.
\begin{figure}
  \centering
  \includegraphics[width=.8\linewidth]{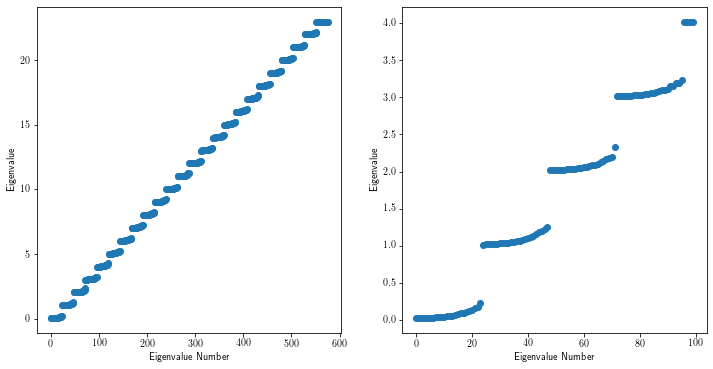}
  \caption{Plot of sorted non-zero eigenvalues of $PXP$ where $P$ is the Fermi projection and $X$ is the lattice position operator for the Haldane model on $24 \times 24$ system with Dirichlet boundary conditions. Entire spectrum (left) and first 100 eigenvalues (right). Parameters chosen are $(t,t',v,\phi) = (1,\frac{1}{10},1,\frac{\pi}{2})$, with onsite disorder drawn from a mean zero normal distribution with variance $\frac{1}{4}$. Despite the disorder, the spectrum still shows clear gaps.} \label{fig:dirichlet-disorder-pxp-evals}
\end{figure}

\begin{figure}
  \centering
  \includegraphics[width=.3\linewidth]{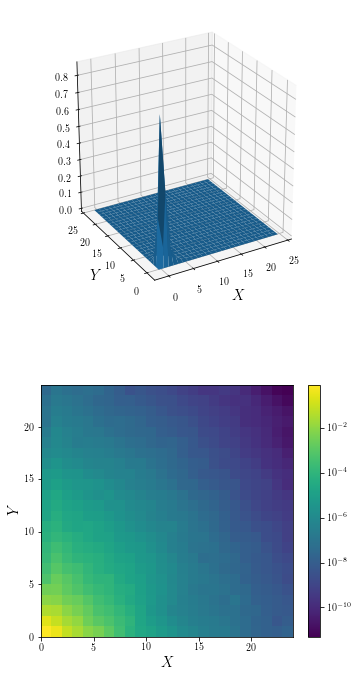}
  \includegraphics[width=.3\linewidth]{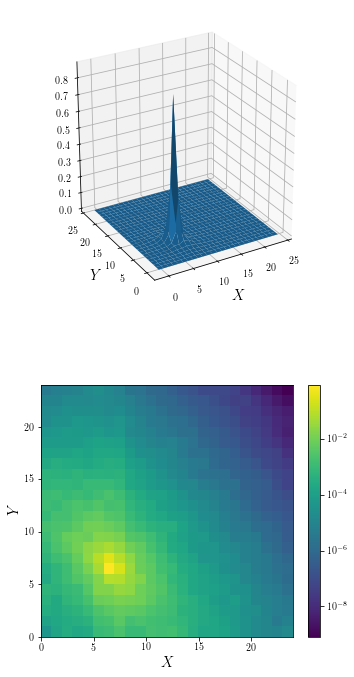}  \includegraphics[width=.3\linewidth]{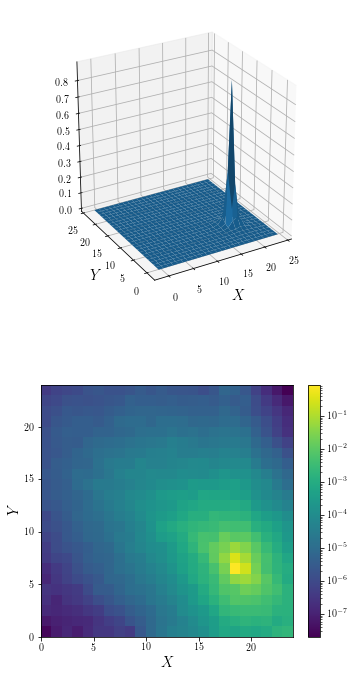} 
  \caption{Plot of eigenfunctions of the operator $P_j Y P_j$ for different values of $j$ where $\{ P_j \}_{j \in \cJ}$ are the band projectors for $PXP$, $P$ is the Fermi projection, and $X$, $Y$ are the lattice position operators. Parameters chosen are $(t,t',v,\phi) = (1,\frac{1}{10},1,\frac{\pi}{2})$, with onsite disorder drawn from a mean zero normal distribution with variance $\frac{1}{4}$. Top row is a 3D surface plot of the matrix from Equation \eqref{eq:plot-mat}, bottom row is 2D log plot of the top row. Despite the disorder, our algorithm yields exponentially-localized generalized Wannier functions.} \label{fig:dirichlet-disorder-pjypj-efuncs}
\end{figure}


\subsection{Dirichlet Boundary Conditions with Strong Disorder} \label{sec:dirichlet_strong_disorder}
We consider the same setup as the previous section, but with disorder strong enough ($\sigma^2 = 100$) to close the gap of $H$ (for the results shown in Figure \ref{fig:dirichlet_strong_disorder_pxp_evals}, the gap of $H$ $\approx .07$). Although our results do not directly apply to this case, the eigenfunctions of $H$ are themselves localized because of Anderson localization \cite{1958Anderson}. It is therefore plausible that $PXP$ may have gaps and that the eigenfunctions of $P_j Y P_j$ are localized nonetheless. 

We plot the non-zero eigenvalues in $PXP$ in Figure \ref{fig:dirichlet_strong_disorder_pxp_evals}. We find that $\sigma(PXP)$ shows clear gaps, and hence we may define projectors $P_j$ and operators $P_j Y P_j$.

We observe that the eigenfunctions of $P_j Y P_j$ are well localized. In Figure \ref{fig:dirichlet_strong_disorder_pjypj_h_compare}, we plot the eigenfunctions of $H$ in order of increasing energy value and plot an eigenfunction of $P_j Y P_j$ which has the same center. We observe that as the energy level increases, the corresponding eigenfunction of $H$ becomes less localized. In comparison, the eigenfunctions of $P_j Y P_j$ have similar rates of decay for all values of $j \in \cJ$.
\begin{figure}
  \centering
  \includegraphics[width=.8\linewidth]{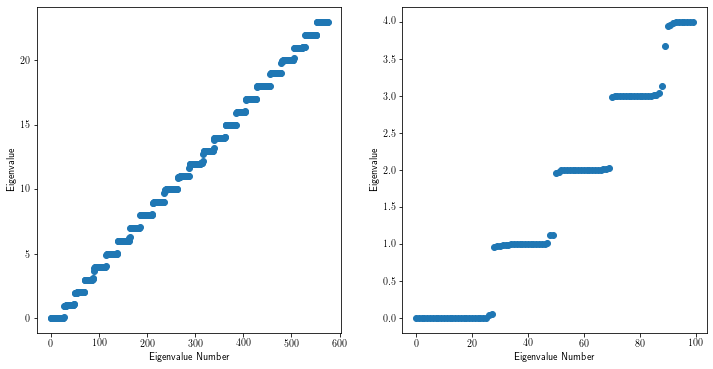}
  \caption{Plot of sorted non-zero eigenvalues of $PXP$ where $P$ is the Fermi projection and $X$ is the lattice position operator for the Haldane model on $24 \times 24$ system with Dirichlet boundary conditions. Entire spectrum (left) and first 100 eigenvalues (right). Parameters chosen are $(t,t',v,\phi) = (1,\frac{1}{10},1,\frac{\pi}{2})$, with disorder is drawn from a mean zero normal distribution with variance $100$.} \label{fig:dirichlet-strong-disorder-pxp-evals}
  \label{fig:dirichlet_strong_disorder_pxp_evals}
\end{figure}
\begin{figure}
  \centering
  \includegraphics[width=.4\linewidth]{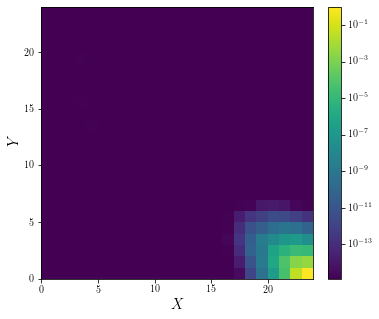} \includegraphics[width=.4\linewidth]{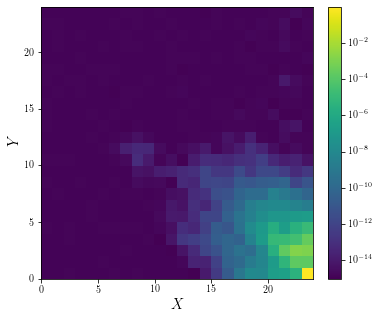} \\
  \includegraphics[width=.4\linewidth]{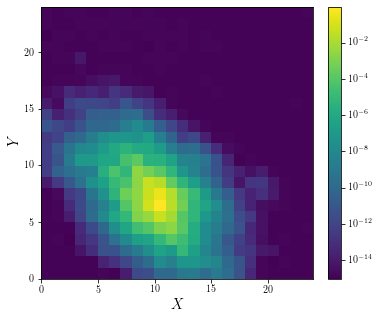} \includegraphics[width=.4\linewidth]{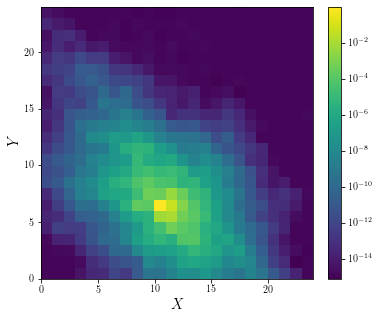} \\
  \includegraphics[width=.4\linewidth]{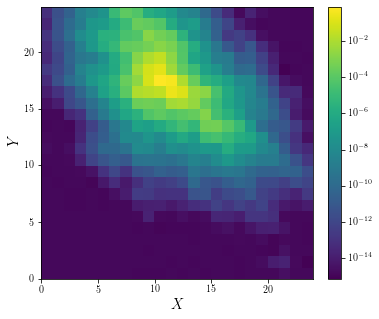} \includegraphics[width=.4\linewidth]{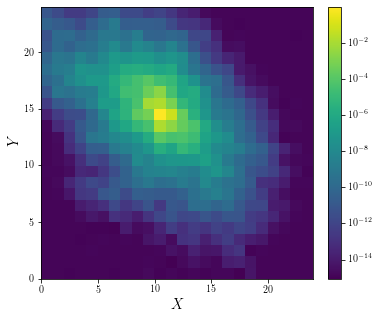} \\
        \includegraphics[width=.4\linewidth]{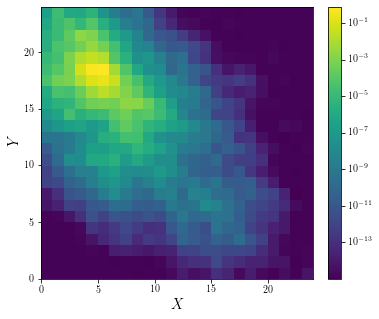} \includegraphics[width=.4\linewidth]{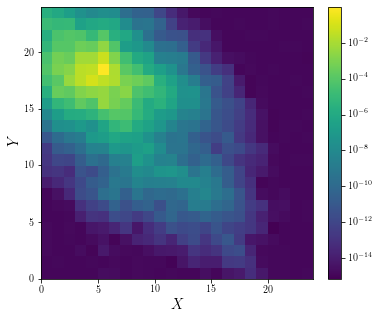} 
        \caption{Plot of eigenfunctions of $H$ (left) and $P_j Y P_j$ (right) for different values of $j$ where $\{ P_j \}_{j \in \cJ}$ are the band projectors for $PXP$, $P$ is the Fermi projection, and $X$, $Y$ are the lattice position operators.  Parameters chosen are $(t,t',v,\phi) = (1,\frac{1}{10},1,\frac{\pi}{2})$, while onsite disorder is drawn from a mean zero normal distribution with variance $100$. The eigenfunctions of $H$ are sorted in order of increasing energy (top $\rightarrow$ low energy, bottom $\rightarrow$ high energy) and eigenfunctions of $P_j Y P_j$ were chosen to have the same center as the corresponding $H$ eigenfunction.}
        \label{fig:dirichlet_strong_disorder_pjypj_h_compare}
\end{figure}


\subsection{Dirichlet Boundary Conditions using Rotated Position Operators} \label{sec:dirichlet_rotated} We now consider how our results change when we choose to work with different two-dimensional position operators (equivalently, different two-dimensional axes).
Note that, although our proofs are independent of any particular choice of position operators, we cannot rule out the possibility that the uniform spectral gap assumption on $PXP$ (Assumption \ref{def:usg}) holds only for particular choices. We also expect that different choices of position operators will yield different exponentially-localized generalized Wannier functions.

We consider the same Haldane model with Dirichlet boundary conditions but without disorder as in Section \ref{sec:dirichlet}, and introduce rotated position operators
\begin{equation}
  \label{eq:rot-pos}
  \tilde{X} := \frac{X - Y}{\sqrt{2}} \quad \tilde{Y} := \frac{X + Y}{\sqrt{2}}.
\end{equation}
The eigenvalues of $P \tilde{X}P$ are shown in Figures \ref{fig:dirichlet-rot-pxp-evals}. We find that, just like the eigenvalues of $PXP$ in Figure \ref{fig:dirichlet-pxp-evals}, the spectrum shows clear gaps.
The eigenfunctions of $P \tilde{X} P$ are shown in Figure \ref{fig:dirichlet-pxp-efuncs_rot}. They are clearly localized along lines $x + y = c$ for constant $c$. Since $P \tilde{X} P$ has gaps (Figure \ref{fig:dirichlet-rot-pxp-evals}), we can define the band projectors $P_j$ as before. The eigenfunctions of $P_j \tilde{Y} P_j$ are shown in Figure \ref{fig:pjypj-efuncs-tilde} and clearly exponentially decay similarly to those in Figure \ref{fig:dirichlet-pjypj-efuncs}.
\begin{figure}
  \centering
  \includegraphics[width=.8\linewidth]{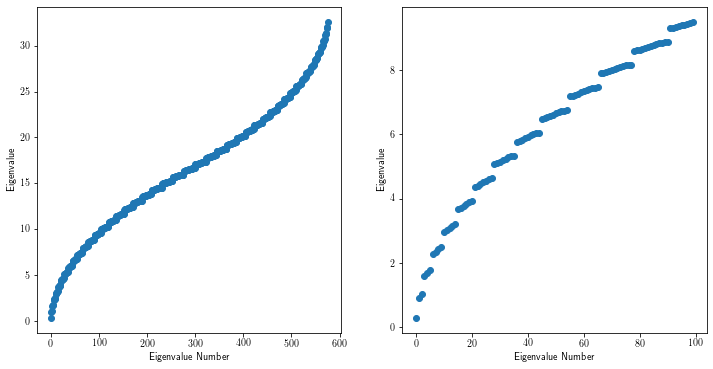}
  \caption{Plot of sorted non-zero eigenvalues of $P\tilde{X}P$ where $P$ is the Fermi projection and $\tilde{X}$ is the lattice position operator rotated by $45^\circ$ (see Equation \eqref{eq:rot-pos}) for the Haldane model on $24 \times 24$ system with Dirichlet boundary conditions. Entire spectrum (left) and first 100 eigenvalues (right). Parameters chosen are $(t,t',v,\phi) = (1,\frac{1}{10},1,\frac{\pi}{2})$. Full non-zero spectrum (left), zoom-in for the first 100 eigenvalues (right). The spectrum shows clear gaps.} \label{fig:dirichlet-rot-pxp-evals}
\end{figure}

\begin{figure}
  \centering
  \includegraphics[width=.23\linewidth]{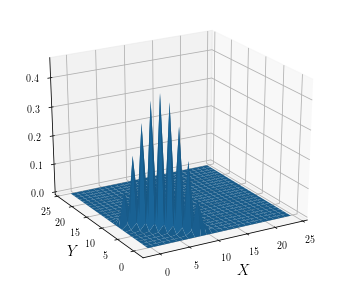} \includegraphics[width=.23\linewidth]{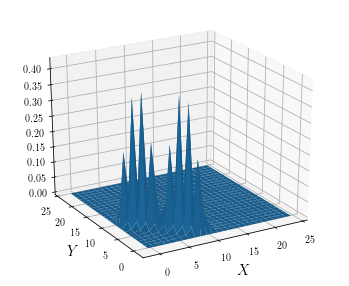} \includegraphics[width=.23\linewidth]{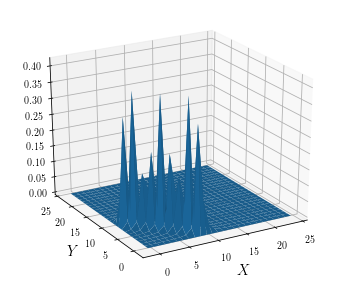} \\
  \includegraphics[width=.23\linewidth]{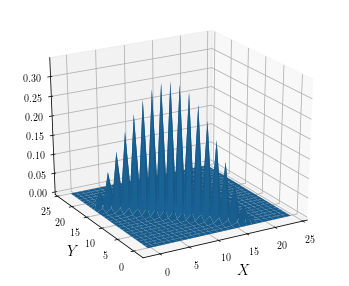} \includegraphics[width=.23\linewidth]{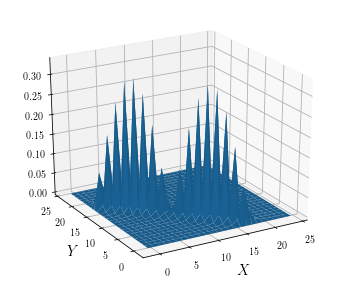} \includegraphics[width=.23\linewidth]{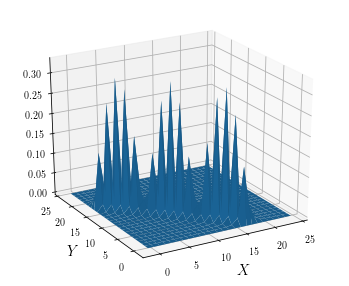} \\
  \includegraphics[width=.23\linewidth]{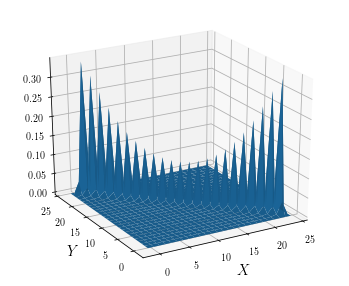} \includegraphics[width=.23\linewidth]{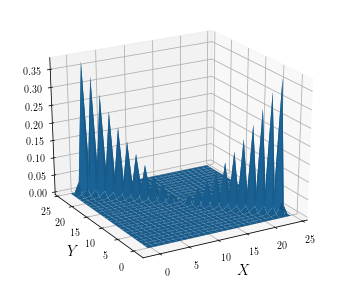} \includegraphics[width=.23\linewidth]{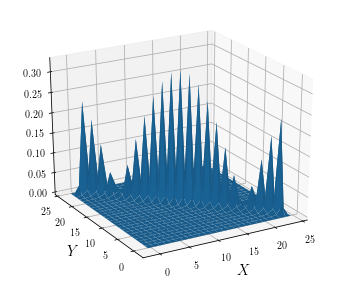} 
  \caption{Plot of eigenfunctions of the operator $P\tilde{X}P$ where $P$ is the Fermi projection and $\tilde{X}$ is the rotated lattice position operator for the Haldane model on $24 \times 24$ system with Dirichlet boundary conditions. Parameters chosen are $(t,t',v,\phi) = (1,\frac{1}{10},1,\frac{\pi}{2})$. Each eigenfunction is localized along a line $x + y = c$ for some constant $c$.} \label{fig:dirichlet-pxp-efuncs_rot}
\end{figure}

\begin{figure}
  \centering
  \includegraphics[width=.3\linewidth]{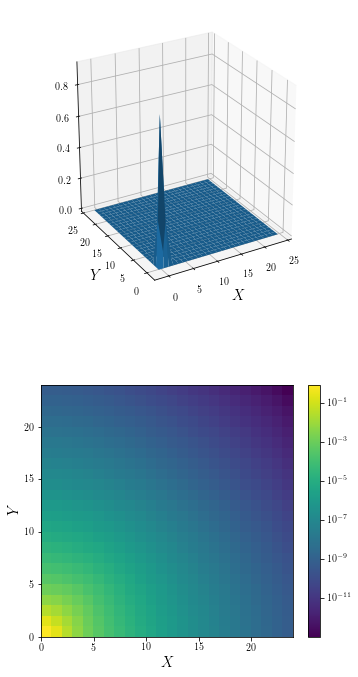} 
  \includegraphics[width=.3\linewidth]{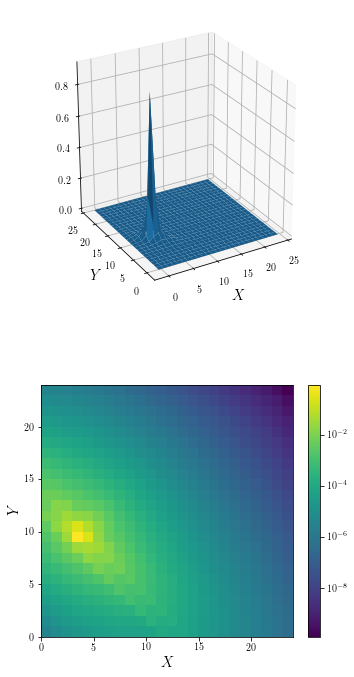} 
  \includegraphics[width=.3\linewidth]{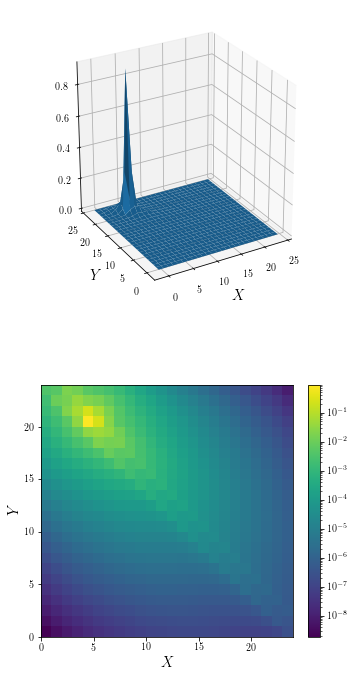} 
  \caption{Plot of eigenfunctions of the operator $P_j \tilde{Y} P_j$ for different values of $j$ where $\{ P_j \}_{j \in \cJ}$ are the band projectors for $P\tilde{X}P$, $P$ is the Fermi projection, and $\tilde{X}$, $\tilde{Y}$ are the rotated lattice position operators (see Equation \eqref{eq:rot-pos}). The projection $P$ comes from the Haldane model on $24 \times 24$ system with Dirichlet boundary conditions. Parameters chosen are $(t,t',v,\phi) = (1,\frac{1}{10},1,\frac{\pi}{2})$. Top row is a 3D surface plot of the matrix from Equation \eqref{eq:plot-mat}, bottom row is 2D log plot of the top row.} \label{fig:pjypj-efuncs-tilde}
\end{figure}

\section{Notation and Conventions} \label{sec:notations}
\label{sec:notation}
Before moving on to our proofs, we pause to review some notation. Vectors in $\field{R}^d$ will be denoted by bold face with their components denoted by subscripts. For example, $\vec{v} = (v_1, v_2, v_3, \cdots, v_d) \in \field{R}^d$. For any $\vec{v} \in \field{R}^d$, we use $\vert{} \cdot \vert{}$ to denote its Euclidean norm. That is $\vert{} \vec{v} \vert{} := \bigl(\sum_{i=1}^d v_i^2\bigr)^{1/2}$.

For any $f : \field{R}^2 \rightarrow \field{C}$ and any $p \in [1, \infty]$, we use $\|f\|_{L^p}$ to denote the $L^p$-norm of $f$. Additionally, we define $L^p_{\loc}(\field{R}^2)$ as the space of local $L^p(\field{R}^2)$ functions and $L^p_{\uloc}(\field{R}^2)$ as the set of uniformly local $L^p(\field{R}^2)$ functions
\[
    L^p_{\uloc}(\field{R}^2) := \Big\{ f \in L^p_{\loc}(\field{R}^2) : \exists M > 0 \text{ s.t. } \forall \vec{x'} \in \field{R}^2, \int_{\vert{} \vec{x} - \vec{x}'\vert{} \leq 1} \vert{}f(\vec{x})\vert{}^p \text{\emph{d}}{\vec{x}} \leq M \Big\}.
\]
For the special case of $p = 2$, we will drop the subscript $L^2$ and define $\| f \| := \| f \|_{L^2}$. For any linear operator, we will use $\| A \|$ to denote the spectral norm of $A$.

We generally define position operators $X$ and $Y$ with respect to a fixed co-ordinate system by 
\begin{equation}
\label{eq:X}
    X f(\vec{x}) = x_1 f(\vec{x}), \quad Y f(\vec{x}) = x_2 f(\vec{x}),
\end{equation}
although in the periodic setting it is natural to re-scale these operators \eqref{eq:XY_periodic}. 

Given two sets $A, B \subseteq \field{R}$ we define their diameter and distance as follows
\[
  \begin{array}{l}
    \diam{(A)} := \sup \{ \vert{}a_1 - a_2\vert{} : a_1, a_2 \in A \} \\[1.5ex]
    \dist{(A,B)} := \inf \{ \vert{}a - b\vert{} : a \in A, b \in B \}.
  \end{array}
\]
For any contour in the complex plane, $\cC$, we will use $\ell(\cC)$ to denote the length of $\cC$. 

Given a point $\vec{a} \in \field{R}^2$, and a non-negative constant $\gamma \geq 0$, we define an exponential growth operator by
\begin{equation*}
  B_{\gamma,\vec{a}} := \exp\left( \gamma \sqrt{ 1 + (X-a_1)^2 + (Y-a_2)^2 } \right).
\end{equation*}
Given a linear operator $A$, we define 
\begin{equation} \label{eq:notation_1}
  A_{\gamma,\vec{a}} := B_{\gamma,\vec{a}} A B_{\gamma,\vec{a}}^{-1}.
\end{equation}
We refer to $A_{\gamma,\vec{a}}$ as ``exponentially-tilted'' relative to $A$. We will often prove estimates where we use the notation \eqref{eq:notation_1} but omit the point $\vec{a}$. In this case the estimate should be understood as uniform in the choice of point $\vec{a}$. As a note, per our convention, when $\gamma = 0$, $A_{\gamma,\vec{a}} = A$. 

\section{Precise Statement of Main Theorem} \label{sec:results}

In this section, we present our assumptions in full detail and re-state our main theorem precisely. 
The details of the proof will be presented in Sections \ref{sec:self-adjointness}, \ref{sec:pjypj-disc-spec}, and \ref{sec:pjypj-exp-loc}, with proofs of key estimates postponed until Appendices \ref{sec:pj-props} and \ref{sec:shifting-proof}. 

Key to our proofs is the notion of an exponentially localized orthogonal projector which we now define. Note that we abuse notation in the standard way by using the same letter for the operator and integral kernel.
\begin{definition}[Exponentially Localized Projector]
  \label{def:exp-loc-kernel}
  We say that an orthogonal projector $P$ on $L^2(\field{R}^2)$ has an exponentially localized kernel if it admits an integral kernel $P(\cdot,\cdot): \field{R}^2 \times \field{R}^2 \rightarrow \field{C}$ and there exist finite, positive constants $(C, \gamma)$ so that for all $\vec{x}, \vec{x}' \in \field{R}^2$ the following bound holds almost everywhere:
  \begin{equation} \label{eq:exp_loc_proj_property}
    \vert{} P(\vec{x},\vec{x}') \vert{} \leq C e^{-\gamma \vert{}\vec{x} - \vec{x}'\vert{}}.
  \end{equation}
\end{definition}

In the context of electronic structure theory, exponentially localized projectors naturally arise as spectral projectors for gapped magnetic Schr{\"o}dinger operators. To illustrate this principle, we now state sufficient regularity conditions 
so that the Fermi projection is
an exponentially localized orthogonal projector:
\begin{assumption}[Regularity] \label{as:H_assump}
We consider the Hilbert space $\mathcal{H} := L^2(\field{R}^2)$ acted on by the Hamiltonian 
\[
  H = (-i\nabla + A(\vec{x}))^2 + V(\vec{x}).
\]
We assume that $V \in L^2_{\uloc}(\field{R}^2; \field{R})$, $A \in L^4_{\loc}(\field{R}^2; \field{R}^2)$, and $\divd A \in L^2_{\loc}(\field{R}^2; \field{R})$. 
\end{assumption}
Assumption \ref{as:H_assump} ensures that $H$ is essentially self-adjoint with core $C^\infty_c(\field{R}^2)$ \cite[Theorem 3]{1981LeinfelderSimader}, and is sufficient to prove the exponential localization of the Fermi projection. 
Note that we do not make any assumptions about periodicity of $H$ and hence our results apply for both periodic and non-periodic systems. 

We additionally assume that $H$ has a spectral gap. More formally: 
\begin{assumption}[Spectral gap] \label{as:gap_assump}
We assume that we can write
\[
  \sigma(H) = \sigma_0 \cup \sigma_1,
\]
where $\dist{(\sigma_0, \sigma_1)} > 0$ and $\diam{(\sigma_0)} < \infty$. 
\end{assumption}
Since $H$ is essentially self-adjoint, by the spectral theorem there exists an orthogonal projector $P$ associated with $\sigma_0$.  

Because of the spectral gap of $H$, it can be shown the projector $P$ admits an integral kernel which is exponentially localized.
\begin{proposition}[\cite{1982SimonSchrodinger,Moscolari2019,Marcelli2020}]
  Let $H$ satisfy the regularity and spectral gap assumptions (Assumptions \ref{as:H_assump} and \ref{as:gap_assump}) and let $P$ be the spectral projector associated with $\sigma_0$. Then $P$ is an exponentially localized orthogonal projector. 
\end{proposition}

Our main results only require that the projector $P$ is exponentially localized as in Definition \ref{def:exp-loc-kernel}. Hence we expect these results can be extended to spectral projections appearing in other contexts.

We now claim the following Lemma:
\begin{lemma}\label{lem:PXP}
Let $P$ be any exponentially localized orthogonal projector. For $X$ is as in \eqref{eq:X}, the operator $PXP$ is self-adjoint in $L^2(\field{R}^2)$ with domain $\{ \mathcal{D}(X) \cap \range(P) \} \cup \range(P)^\perp$.
\end{lemma}
We prove Lemma \ref{lem:PXP} in Section \ref{sec:self_PXP}. Lemma \ref{lem:PXP} generalizes part (i) of Theorem 1 of Nenciu-Nenciu \cite{1998NenciuNenciu} to two dimensions. The proof is similar and only relies on exponential localization of $P$. 

We are now in a position to give our precise assumption on $PXP$. When it holds, we say that $PXP$ has \emph{uniform spectral gaps}.
\begin{assumption}[Uniform Spectral Gaps]  \label{def:usg}
We assume there exist constants $(d, D)$ such that: 
  \begin{enumerate}
  \item There exists a countable set, $\cJ$, such that:
    \[
      \sigma(PXP) = \bigcup_{j \in \cJ} \sigma_j.
    \]
  \item The distance between $\sigma_j, \sigma_k$ ($j \neq k$) is uniformly lower bounded:
    \[
      d := \min_{j \neq k} \Big( \dist( \sigma_j, \sigma_k ) \Big)  > 0.
    \]
  \item The diameter of each $\sigma_j$ is uniformly bounded:
    \[
      D := \max_{j \in \cJ} \Big( \diam( \sigma_j ) \Big) < \infty.
    \]
  \end{enumerate}
\end{assumption}
If $PXP$ has uniform spectral gaps in the sense of Assumption \ref{def:usg}, we can define spectral projections associated with each subset $\left\{ \sigma_j \right\}_{j \in \cJ}$ of $\sigma(PXP)$. We will refer to these projections as \emph{band projectors}. Note that our use of ``band'' in this context should not be confused with its use in the context of Bloch eigenvalue bands of periodic operators. 
\begin{definition}[Band projectors] \label{def:band-projectors}
When $PXP$ has uniform spectral gaps with constants $(d,D)$ and decomposition $\{\sigma_j \}_{j \in \cJ}$ in the sense of Assumption \ref{def:usg}, we let
  \begin{equation} \label{eq:Pj}
    P_j := P \bigg( \frac{1}{2 \pi i} \int_{\cC_j} (\lambda - PXP)^{-1} \emph{\dee}{\lambda} \bigg) P \quad j \in \cJ,
  \end{equation}
  where $\cC_j$ encloses $\sigma_j$ and satisfies 
  \begin{equation} \label{eq:Pj_contour}
    \sup_{\lambda \in \cC_j} \|(\lambda - PXP)^{-1}\| \leq C d^{-1} \text{ and } \ell(\cC_j) \leq C'(D + d)
  \end{equation}
  for some absolute constants $C, C'$ independent of $j$. In particular, $P_j$ is an orthogonal projection onto the spectral subspace associated with $\sigma_j$.
\end{definition}
Note that the existence of a contour $\cC_j$ satisfying \eqref{eq:Pj_contour} is clearly guaranteed by the uniform spectral gap assumption (Assumption \ref{def:usg}).

The addition of $P$ in the definition \eqref{eq:Pj} of the band projectors, $P_j$, is to ensure that $\range{(P_j)} \subseteq \range{(P)}$. This is necessary to avoid the trivial null space of $PXP$ consisting of functions in the null space of $P$. Since $P$ commutes with $PXP$, we have that 
\[
  P_j = \bigg( \frac{1}{2 \pi i} \int_{\cC_j} (\lambda - PXP)^{-1} \dee{\lambda} \bigg) P.
\]
With this definition of $P_j$, we will prove the following proposition: 
\begin{proposition}
  \label{prop:pj-props}
  Suppose that $P$ is an exponentially localized orthogonal projector (Definition \ref{def:exp-loc-kernel}) and let $PXP$ satisfy the uniform spectral gap assumption (Assumption \ref{def:usg}). Let $P_j$, with $j \in \cJ$, be the band projectors as in Definition \ref{def:band-projectors}.

  There exist constants $(C, C', \gamma)$, independent of $j \in \cJ$, so that for all $j \in \cJ$:
  \begin{enumerate}
      \item The projector $P_j$ admits an kernel $P_j(\cdot,\cdot) : \field{R}^2 \times \field{R}^2 \rightarrow \field{C}$ so that for almost all $\vec{x}, \vec{x}' \in \field{R}^2$:
    \[
      \vert{} P_j(\vec{x}, \vec{x}') \vert{} \leq C e^{-\gamma \vert{} \vec{x} - \vec{x}'\vert{}} 
    \] 
          
  \item For each $\eta_j \in \sigma_j$
    \[
      \begin{split}
        & \| (X - \eta_j) P_{j,\gamma} \| \leq C' \\
        & \| P_{j,\gamma} (X - \eta_j) \| \leq C'
      \end{split}
    \]
  \end{enumerate}
\end{proposition}
\begin{proof}
  Proven in Appendix \ref{sec:pj-props}.
\end{proof}
In two dimensions, $P Y P$, like $P X P$, does not have compact resolvent. However, as a consequence of the second estimate in Proposition \ref{prop:pj-props}, one can show that in fact $P_j Y P_j$ has compact resolvent on $\range{(P)}$. Since $P_j$ additionally admits an exponentially localized kernel, by applying similar technical machinery as used in Nenciu-Nenciu \cite{1998NenciuNenciu}, we can show eigenfunctions of $P_j Y P_j$ are exponentially localized. 

We now formally outline these steps by a series of lemmas (Lemma \ref{lem:pjypj}, \ref{lem:compact_res}, \ref{lem:exp_decay})
\begin{lemma} \label{lem:pjypj}
  Suppose that $\{ P_j \}_{j \in \cJ}$ is any family of orthogonal projectors satisfying the conclusions of Proposition \ref{prop:pj-props} and let $Y$ be as in \eqref{eq:X}. Then the operators $P_j Y P_j$ are each self-adjoint on $L^2(\field{R}^2)$ with domains $\{ \mathcal{D}(Y) \cap \range( P_j ) \} \cup \range( P_j )^\perp$.
\end{lemma}
We prove Lemma \ref{lem:pjypj} in Section \ref{sec:self_PjYPj}. The proof is very similar to the proof of Lemma \ref{lem:PXP}, except that it relies on exponential localization of each $P_j$ (Proposition \ref{prop:pj-props}).

We now claim that each of the operators $P_j Y P_j$ has compact resolvent:
\begin{lemma} \label{lem:compact_res}
Each of the operators $P_j Y P_j$ $j \in \cJ$, as in Lemma \ref{lem:pjypj}, has compact resolvent.
\end{lemma}
We prove Lemma \ref{lem:compact_res} in Section \ref{sec:pjypj-disc-spec}. Lemma \ref{lem:compact_res} generalizes part (ii) of Theorem 1 of \cite{1998NenciuNenciu} to two dimensions (we sketch the generalization to higher dimensions in Section \ref{sec:higher-d}). The generalization is non-trivial since the decay induced by $Y$ alone is not sufficient for the resolvent of $P_j Y P_j$ to be compact. To prove compactness, we make use of the fact that functions in each subspace $\range(P_j)$ decay with respect to $x$. 

Our final Lemma states that eigenfunctions of the operators $P_j Y P_j$ are exponentially localized:
\begin{lemma} \label{lem:exp_decay}
Let $P_j Y P_j$, $j \in \cJ$, be as in Lemma \ref{lem:pjypj}. Then there exists a $\gamma'' > 0$, independent of $j$, such that if $\psi \in \range{(P_j)}$ and $P_j Y P_j \psi = \eta' \psi$, then for all $\eta \in \sigma_j$
    \[
      \int e^{2 \gamma'' \sqrt{1 + (x_1 - \eta)^2 + (x_2 - \eta')^2}} \vert{}\psi(\vec{x})\vert{}^2 \, \text{\emph{d}}\vec{x} \leq 16 e^{\gamma''\sqrt{1 + 2 b^2}}.
    \]
    Here $b$ is a finite positive constant (independent of $j$ and $\psi$) which depends only on the collection $\{ P_j \}_{j \in \cJ}$.
\end{lemma}
We prove Lemma \ref{lem:exp_decay} in Section \ref{sec:pjypj-exp-loc}. Lemma \ref{lem:exp_decay} generalizes part (iii) of Theorem 1 of Nenciu-Nenciu \cite{1998NenciuNenciu} to two dimensions (we sketch the generalization to higher dimensions in Section \ref{sec:higher-d}).
We note that although the proof of Lemma \ref{lem:exp_decay} has a similar structure to that of \cite{1998NenciuNenciu}, our proof involves technical innovations which are necessary to (1) generalize their proof to two dimensions and higher (2) give a proof which does not refer to the original Hamiltonian $H$, requiring only the properties of the projectors $P_j$ (Proposition \ref{prop:pj-props}).

The proof of our main theorem (Theorem \ref{th:main_theorem}) immediately follows from Lemmas \ref{lem:PXP}, \ref{lem:pjypj}, \ref{lem:compact_res}, and \ref{lem:exp_decay} and the spectral theorem for self-adjoint operators applied to each $P_j Y P_j$, $j \in \cJ$. For each $(j,m) \in \cJ \times \mathcal{M}$, we define the points $(a_{j,m},b_{j,m})$ by the $(\eta,\eta')$ appearing in Lemma \ref{lem:exp_decay}, i.e. $a_{j,m}$ can be taken as any $\eta \in \sigma_j$, while $b_{j,m}$ is the associated eigenvalue of $\psi_{j,m}$ considered as an eigenfunction of $P_j Y P_j$.

\section{Proof of self-adjointness of $P X P$ and $P_j Y P_j$ (Lemmas \ref{lem:PXP} and \ref{lem:pjypj})} \label{sec:self-adjointness}

The proofs that $PXP$ and $P_j Y P_j$ are self-adjoint in two dimensions are almost exactly the same and they both use the same argument as given in \cite{1998NenciuNenciu}. Because the proofs are so similar, we will first prove that $P X P$ is essentially self-adjoint and then note the changes which need to be made to prove that $P_j Y P_j$ is essentially self-adjoint.

\subsection{Proof of self-adjointness of $PXP$} \label{sec:self_PXP}
To prove self-adjointness of $P X P$ we require two lemmas regarding norm estimates for the projector $P$. We invite the reader to recall our convention for exponentially tilted operators (such as $P_{\gamma}$) given in Section \ref{sec:notation}.

  

Using the exponential localization of $P$, the following operator norm estimate can easily be verified. 
\begin{lemma}
  \label{lem:p-est1}
  Suppose that $P$ is an exponentially localized projector (Definition \ref{def:exp-loc-kernel}). There exists a finite positive constant $K_1$ so that
  \[
  \| [P, X] \| \leq K_1.
  \]
\end{lemma}
\begin{proof}
Since $P$ admits an exponentially localized kernel, we have that there exist finite, positive constants $(C, \gamma_1)$ so that
\[
\vert{} P(\vec{x}, \vec{x}') \vert{} \leq C e^{-\gamma_1 \vert{}\vec{x} - \vec{x}'\vert{}}.
\]
Therefore, for any $f \in L^2(\field{R}^2)$ we have
\begin{align*}
\| [P, X] f \|^2 
& = \| (X P - P X) f \|^2 \\
& = \int_{\field{R}^2} \left\vert{} \int_{\field{R}^2} (x_1 - x_1') P(\vec{x}, \vec{x}') f(\vec{x}') \dee{\vec{x'}} \right\vert{}^2 \dee{\vec{x}}.
\end{align*}
Since $P$ admits an exponentially localized kernel, one easily checks that there exists an constant $C'$ so that the following bound holds pointwise almost everywhere:
\begin{align*}
\vert{}(x_1 - x_1') P(\vec{x}, \vec{x}') \vert{}
& \leq C \vert{}x_1 - x_1'\vert{} e^{-\gamma_1 \vert{}\vec{x} - \vec{x}'\vert{}} \\
& \leq C' e^{-(\gamma_1/2) \vert{}\vec{x} - \vec{x}'\vert{}}.
\end{align*}
Combining this pointwise bound with Young's convolution inequality then gives
\begin{align*}
\| [P, X] f \|^2
& \leq C'^2  \int_{\field{R}^2} \left( \int_{\field{R}^2} e^{-(\gamma_1/2) \vert{}\vec{x} - \vec{x}'\vert{}} \left\vert{} f(\vec{x}') \right\vert{} \dee{\vec{x}'} \right)^2 \dee{\vec{x}}  \\
& \leq C'^2 \left( \int_{\field{R}^2} e^{-(\gamma_1 / 2) \vert{} \vec{x} \vert{}} \dee{\vec{x}} \right)^2 \left( \int_{\field{R}^2} \vert{}f(\vec{x})\vert{}^2 \dee{\vec{x}} \right) \\
& \leq K_1^2 \| f \|^2
\end{align*}
which completes the proof.

\end{proof}
We can now prove that $P X P$ is self-adjoint. We will first prove that $PXP$ is self-adjoint $\range(P) \cap \mathcal{D}(X) \rightarrow \range(P)$. Since we can extend $PXP$ by zero to $\range(P)^\perp$ it is trivial to see that $PXP$ is then self-adjoint in all of $L^2(\field{R}^2)$ with the domain $\{ \range(P) \cap \mathcal{D}(X) \} \cup \range(P)^\perp$. We will prove that $PXP$ is self-adjoint $\range(P) \cap \mathcal{D}(X) \rightarrow \range(P)$ by proving that $PXP \pm i \mu$ is surjective on $\range(P)$ for all $\mu > 0$. 

It is immediately clear that $PXP$ maps $\range(P) \cap \mathcal{D}(X) \rightarrow \range(P)$. Recall that $X$ is self-adjoint with domain $\mathcal{D}(X) := \{ f \in L^2 : x_1 f \in L^2 \}$, and hence $(X \pm i\mu)^{-1}$ is bounded for all $\mu > 0$ and $\|(X \pm i\mu)^{-1}\| \leq \mu^{-1}$. Since $(X \pm i \mu) P (X \pm i \mu)^{-1} P = P + [X,P] (X \pm i \mu)^{-1} P$, and since $[X,P]$ is bounded by Lemma \ref{lem:p-est1}, $P (X \pm i \mu)^{-1} P : \range(P) \rightarrow \range(P) \cap \mathcal{D}(X)$. But now, using essentially the same calculation and the fact that $P$ is a projection, we have
\begin{equation} \label{eq:self-adjointness-1}
    (P X P \pm i \mu) P (X \pm i \mu)^{-1} P = P (I + P [ X, P ] (X \pm i \mu)^{-1} P).
\end{equation}
By choosing $\mu \geq 2 \|[X,P]\|$ we have that $\| P [ X, P ] (X \pm i \mu)^{-1} P \| \leq \frac{1}{2}$, and hence
\[
  (P X P \pm i \mu) P (X \pm i \mu)^{-1} P \Big( I + P [ X, P ] (X \pm i \mu)^{-1} P \Big)^{-1} = P.
\]
We conclude that with the domain $\mathcal{D}(X) \cap \range(P)$, $PXP \pm i \mu$ is surjective on $\range(P)$ and is hence self-adjoint by the fundamental criterion of self-adjointness \cite[Chapter VIII]{1980ReedSimon}.

\subsection{Proof of self-adjointness of $P_j Y P_j$} \label{sec:self_PjYPj}
In this section we require analogous estimates on $P_j$ as we required on $P$ in Section \ref{sec:self_PXP}. Using the exponential localization of $P_j$ (Proposition \ref{prop:pj-props}(1)), the following operator norm estimate can easily be verified. 
\begin{lemma}
  \label{lem:pj-est1}
  Suppose that $P_j$ is an exponentially localized projector (Definition \ref{def:exp-loc-kernel}). There exists a finite, positive constant $K_1'$ so that
  \[
  \| [P_j, Y] \| \leq K_1'
  \]
\end{lemma}
Letting $\mathcal{D}(Y) := \{ f \in L^2 : x_2 f \in L^2 \}$, self-adjointness of $P_j Y P_j$ with the domain $\{ \mathcal{D}(Y) \cap \range(P_j) \} \cup \range(P_j)^\perp$ now follows by the exact same calculation given in the previous section. 

\section{Proof that $P_j Y P_j$ has compact resolvent (Lemma \ref{lem:compact_res})} \label{sec:pjypj-disc-spec}
For this proof we will make use of the estimate from Proposition \ref{prop:pj-props}. For the first part of this proof, we state and prove the following trivial lemma:
\begin{lemma}
\label{lem:hs-compact}
  For any $f \in L^2(\field{R}^2)$ let $M_f$ denote the multiplicative operator
  \[
    M_f g(\vec{x}) := f(\vec{x}) g(\vec{x})
  \]
  If $A$ admits an exponentially localized kernel (Definition \ref{def:exp-loc-kernel}) then $M_f A$ is a Hilbert-Schmidt operator and hence compact.
\end{lemma}
\begin{proof}
  Since $A$ has an integral kernel and $M_f$ is a multiplicative operator it's clear that $M_f A$ also has an integral kernel. Therefore, the Hilbert-Schmidt norm of $M_f A$ can be upper bounded as follows:
  \begin{align*}
    \| M_f A \|_{HS}^2
    & = \int_{\field{R}^2} \int_{\field{R}^2} \vert{}f(\vec{x})\vert{}^2 \vert{}A(\vec{x}, \vec{x}')\vert{}^2 \dee{\vec{x}} \, \dee{\vec{x}'} \\
    & \leq \int_{\field{R}^2} \int_{\field{R}^2} C^2 e^{-2\gamma \vert{} \vec{x} - \vec{x}'\vert{}} \vert{}f(\vec{x})\vert{}^2 \dee{\vec{x}} \,\dee{\vec{x}'} \\
        & \leq C' \| f \|^2
  \end{align*}
where in the last line we have used Young's convolution inequality.
\end{proof}
We now continue to the proof of Lemma \ref{lem:compact_res}. Following the main calculation of Section \ref{sec:self_PXP} with $PXP$ replaced by $P_j Y P_j$, we have that for any $\mu > 0$,
\[
  (P_j Y P_j \pm i \mu) P_j (Y \pm i \mu)^{-1} P_j \Big( I + P_j [ Y, P_j ] (Y \pm i \mu)^{-1} P_j \Big)^{-1} = P_j.
\]
Since $P_j Y P_j$ is self-adjoint we can invert $(P_{j} Y P_{j} \pm i \mu)$ for $\mu > 0$ to get:
\[
  P_{j} (Y \pm i \mu)^{-1} P_{j} \Big(I + P_{j} [ Y, P_{j}] (Y \pm i \mu)^{-1} P_{j}\Big)^{-1} = (P_{j} Y P_{j} \pm i \mu)^{-1} P_{j}.
\]
Since the product of a compact operator and a bounded operator is compact, it follows that to show that $(P_{j} Y P_{j} \pm i \mu)^{-1}$ is compact it is enough to show that $P_{j} (Y \pm i \mu)^{-1} P_{j}$ is compact.

Let $\eta \in \sigma_j$, i.e. as in Proposition \ref{prop:pj-props}. Taking $P_{j} (Y \pm i \mu)^{-1} P_{j}$ and inserting $(X - \eta + i)^{-1}(X - \eta + i)$ gives:
\begin{equation}
\label{eq:pj_yinv_pj}
  P_{j} (Y \pm i \mu)^{-1} P_{j} = \Big[ P_{j}  (Y \pm i \mu)^{-1} (X - \eta + i)^{-1}  \Big] \Big[  (X - \eta + i)  P_{j} \Big]. 
\end{equation}
Since $(X - \eta + i) P_j$ is bounded by Proposition \ref{prop:pj-props}(2), to show that $P_{j} (Y \pm i \mu)^{-1} P_{j}$ is compact it suffices to show that $P_{j}  (Y \pm i \mu)^{-1} (X - \eta + i)^{-1}$ is compact. 

To show this, first fix some $N > 0$ and let $\chi_N$ denote the cutoff function
\begin{equation*}
    \chi_N(\vec{x}) := \begin{cases} 1 & \text{ if } \vert{}x\vert{} \leq N \text{ and } \vert{}y\vert{} \leq N \\ 0 & \text{ otherwise.} \end{cases}.
\end{equation*}
Now treating $\chi_N$ as a multiplicative operator we have that
\begin{align*}
P_{j} (Y \pm i \mu)^{-1} (X - \eta + i)^{-1}
& = P_j \chi_N (Y \pm i \mu)^{-1} (X - \eta + i)^{-1} \\
& \hspace{3em} + P_j (1 - \chi_N) (Y \pm i \mu)^{-1} (X - \eta + i)^{-1}.
\end{align*}
Because of Lemma \ref{lem:hs-compact} and the fact that $P_j$ admits an exponentially localized kernel (Proposition \ref{prop:pj-props}(1)), we know that $P_j \chi_N (Y \pm i \mu)^{-1} (X - \eta + i)^{-1}$ is compact for all $N > 0$. On the other hand we have clearly that
\[
\| (1 - \chi_N) (Y \pm i \mu)^{-1} (X - \eta + i)^{-1} \| = O(N^{-1}),
\]
and hence
\[
\lim_{N \rightarrow \infty} P_j \chi_N (Y \pm i \mu)^{-1} (X - \eta + i)^{-1} = P_j (Y \pm i \mu)^{-1} (X - \eta + i)^{-1}
\]
in the operator norm topology. But now we are done since $P_j (Y \pm i \mu)^{-1} (X - \eta + i)^{-1}$ is the operator norm limit of compact operators.

\section{Eigenfunctions of $P_j Y P_j$ are exponentially localized (Lemma \ref{lem:exp_decay})}
\label{sec:pjypj-exp-loc}
We will begin this section by collecting the operator norm estimates we need for the proof of Lemma \ref{lem:exp_decay}. All of these estimates follow as a consequence of Proposition \ref{prop:pj-props}. 
\begin{lemma}
\label{lem:pj-bounds}
Let $P_j$, $j \in \cJ$, be as in Proposition \ref{prop:pj-props}. Then there exist finite, positive constants $(\gamma'', K_1'', K_2'', K_3'')$, independent of $j$, such that for all $\gamma \leq \gamma''$
\begin{enumerate}[label=(\arabic*)]
\item $\displaystyle \| P_{j,\gamma} - P_j\| \leq K_1'' \gamma$
\item $\displaystyle \| [ P_{j,\gamma}, Y ] \| \leq K_2''$
\item For all $\eta \in \sigma_j$:
  \[
    \|(X - \eta) P_{j,\gamma}\| \leq K_3'' \text{ and } \|P_{j,\gamma} (X - \eta)\| \leq K_3''.
  \]
\end{enumerate}
\end{lemma}
\begin{proof}
Estimate (3) is part of Proposition \ref{prop:pj-props}, hence we only need to show the first two estimates. Both of these estimates follow from the fact that $P_j$ admits an exponentially localized kernel using calculations similar to those used to prove Lemma \ref{lem:p-est1}.
\end{proof}
With these estimates in hand, we are now ready to prove Lemma \ref{lem:exp_decay}. The overall strategy, which follows Nenciu-Nenciu \cite{1998NenciuNenciu}, is to manipulate the eigenvalue equation $P_j Y P_j v = \eta' v$
into the form
\begin{equation} \label{eq:manipulated}
    v = \mathcal{L} v
\end{equation}
for some operator $\mathcal{L}$, multiply both sides of \eqref{eq:manipulated} by $B_{\gamma,\vec{a}}$ for some $\vec{a}$, and then use properties of $\mathcal{L}$ to deduce that the left-hand side is bounded. Our proof differs from Nenciu-Nenciu \cite{1998NenciuNenciu} in important details.

Suppose that $P_j Y P_j$ has an eigenvector $v \in \range(P_j)$ with eigenvalue $\eta'$.  Since $v \in \range{(P_j)}$ we have that
\[
  P_j Y P_j v = \eta' v \Longleftrightarrow P_j Y P_j v = \eta' P_j v \Longleftrightarrow P_j (Y - \eta') P_j v = 0.
\]
Now, assuming for a moment that $O$ is an operator such that $P_j O P_j v \in \range(P_j)$, we can add $i P_j O P_j v$ to both sides of the above equation to find
\begin{equation}
  \label{eq:nenciu-ext-2d-eq1}
  P_j (Y - \eta' + iO) P_j v = i P_j O P_j v.
\end{equation}
The main difference between our proof and the proof of Nenciu-Nenciu lies in the choice of the operator $O$. For now, let's suppose that we have chosen $O$ so that $(Y - \eta' + iO)$ is invertible and multiply both sides of Equation \eqref{eq:nenciu-ext-2d-eq1} by $(Y - \eta' + iO)^{-1}$ to get
\begin{equation}
  \label{eq:nenciu-ext-2d-eq2}
  (Y - \eta' + iO)^{-1} P_j (Y - \eta' + iO) P_j v = i (Y - \eta' + i O)^{-1} P_j O P_j v.
\end{equation}
We can simplify the left hand side of Equation \eqref{eq:nenciu-ext-2d-eq2} by commuting $P_j$ and $(Y - \eta' + iO)$ as follows
\begin{align*}
  (Y - \eta & + i O)^{-1}  P_j (Y - \eta' + i O) P_j  \\
            & = P_j + (Y - \eta' + i O)^{-1} [P_j, Y - \eta' + i O] P_j \\
            & = \Big(I + (Y - \eta' + i O)^{-1} ( [P_j, Y] + i [ P_j, O ] ) \Big) P_j.
\end{align*}
Therefore, we can write Equation \eqref{eq:nenciu-ext-2d-eq2} as follows
\[
  \Big(I + (Y - \eta' + i O)^{-1} ( [P_j, Y] + i [ P_j, O ] ) \Big) P_j v = i (Y - \eta' + i O)^{-1} P_j O P_j v.
\]
To reduce the number of terms in the next steps, let's define
\[
  A := (Y - \eta' + i O)^{-1} ( [P_j, Y] + i [ P_j, O ] ).
\]
With this definition and using that $v \in \range{(P_j)}$ we have that
\begin{equation}
  \label{eq:nenciu-ext-2d-eq3}
  (I + A) v = i (Y - \eta' + i O)^{-1} P_j O v.
\end{equation}
For the next step of the proof, we will want to show that the $(I + A)$ has bounded inverse. To do this, it is enough to show that 
\begin{equation}
  \label{eq:nenciu-ext-2d-eq4}
  \vert{}\vert{}A\vert{}\vert{} = \| (Y - \eta' + i O)^{-1} ( [P_j, Y] + i [ P_j, O ] ) \| \leq \frac{3}{4}.
\end{equation}
For Equation \eqref{eq:nenciu-ext-2d-eq4} to hold, we require a particular choice of the operator $O$ which differs from the choice made in Nenciu-Nenciu \cite{1998NenciuNenciu}. We require that $O$ satisfies the following properties:
\begin{enumerate}
\item $O$ commutes with $B_{\gamma}$. 
\item $O$ contains a cutoff in both the $X$ and $Y$ directions.
\item $(Y - \eta' + iO)$ is invertible.
\end{enumerate}
For our proof, we let $b > 0$ be a constant to be chosen later and set $O$ to be the following operator:
\begin{equation} \label{eq:def_O}
  O = b \Pi^X_{[\eta - b, \eta + b]} \Pi^Y_{[\eta' - b, \eta' + b]} + \vert{}X - \eta\vert{}
\end{equation}
where $\Pi^{X}_I$ (resp. $\Pi^Y$) is a spectral projection for $X$ (resp. $Y$) onto the interval $I$, $\vert{}X - \eta\vert{}$ is the pointwise operator $\vert{}X - \eta\vert{} f(\vec{x}) := \vert{}x_1 - \eta\vert{} f(\vec{x})$, and $\eta \in \sigma_j$ is arbitrary. Using Lemma \ref{lem:pj-bounds}(3) with $\gamma = 0$, $P_j O P_j$ is clearly bounded.

Before continuing we make three important observations about this choice for $O$:
\begin{enumerate}
\item Since $X$ and $Y$ are self-adjoint and commute, the operator $\Pi^X_{[\eta - b, \eta + b]} \Pi^Y_{[\eta' - b, \eta' + b]}$ is an orthogonal projector.
\item Due to the properties of $P_j$ we know that for all $\gamma$ sufficiently small both $\| P_{j,\gamma} \vert{}X - \eta\vert{} \|$ and $\|  \vert{}X - \eta\vert{} P_{j,\gamma} \|$ are bounded.
\item For all $b > 0$, $\|(Y - \eta' + iO)^{-1}\| \leq b^{-1}$.
\end{enumerate}
In what follows, we will abbreviate $\Pi := \Pi^X_{[\eta - b, \eta + b]} \Pi^Y_{[\eta' - b, \eta' + b]}$.

The key trick which allows us to show that $\| A \| \leq \frac{3}{4}$ is the following lemma (see Corollary 8 of \cite{2008WangDu} for an independent proof, see also \cite{2007Kittaneh}):
\begin{lemma}
  \label{lem:pos-comm-bd}
  Let $B,C$ be two bounded operators. If $B$ is positive semidefinite then
  \[
    \|[B,C]\| \leq \|B\|  \|C\|.
  \]
  If both $B$ and $C$ are positive semidefinite then
  \[
    \|[B,C]\| \leq \frac{1}{2} \|B\|  \|C\|.
  \]
\end{lemma}
\begin{proof}
  Suppose that $B$ is positive semidefinite and define $\tilde{B} := B - \frac{1}{2}\|B\| I$. Since $B$ is a positive semidefinite operator its spectrum lies in the range $[0, \|B\|]$. Therefore, the spectrum of $\tilde{B}$ lies in the range $[-\frac{1}{2} \|B\|, \frac{1}{2} \|B\vert{}]$ and hence $\|\tilde{B}\| = \frac{1}{2} \|B\|$. Since the identity commutes with every operator we have
  \[
    \|[B, C]\| = \|[\tilde{B}, C]\| = \|\tilde{B} C - C \tilde{B}\| \leq 2 \|\tilde{B}\|  \|C\| = \|B\|  \|C\|.
  \]
  If $C$ is also positive semidefinite then we can repeat the same argument using $\tilde{C} := C - \frac{1}{2}\|C\| I$ as well to get $\|[B,C]\| \leq \frac{1}{2} \|B\|  \|C\|$.
\end{proof}
We can now prove that, for $b$ sufficiently large, $\| A \| \leq \frac{3}{4}$. The following calculations are clear:
\begin{align*}
  \|A\| = \| (Y & - \lambda + iO)^{-1} [P_j, Y + iO]\| \\
                & \leq \| (Y - \lambda + iO)^{-1} \|  \Big( \|[P_j, Y]\| + \|[P_j, O]\| \Big) \\
                & \leq b^{-1} \Big( \| [P_j, Y] \| + b \|[P_j, \Pi]\| + \|[P_j, \vert{}X - \eta\vert{} ]\| \Big) \\
                &  \leq \|[P_j, \Pi]\| + b^{-1} \Big( \| [P_j, Y] \| + \|P_j\vert{}X - \eta\vert{} \| + \|\vert{}X - \eta\vert{} P_j\| \Big).
\end{align*}
Since $P_j$ and $\Pi$ are both orthogonal projectors we can apply Lemma \ref{lem:pos-comm-bd} to conclude that $\|[P_j, \Pi]\| \leq \frac{1}{2}$. It now follows that 
\[
  \|A\| \leq \frac{1}{2} + b^{-1} \Big( \| [P_j, Y] \| + \|P_j\vert{}X - \eta\vert{} \| + \|\vert{}X - \eta\vert{} P_j\| \Big),
\]
and if we choose $b$ so that
\[
  b > 4 \Big(\| [P_j, Y] \| + \|P_j\vert{}X - \eta\vert{} \| + \|\vert{}X - \eta\vert{} P_j\| \Big),
\]
we have that 
\begin{equation*}
    \|A\| \leq \frac{3}{4}.
\end{equation*}
Note that because of our estimates from Lemma \ref{lem:pj-bounds} we know that we can choose $b$ so that $b < \infty$.

Returning to Equation \eqref{eq:nenciu-ext-2d-eq3}, we now know that we can invert $(I + A)$ and get
\begin{align*}
  (I + A) v  & = i (Y - \eta' + i O)^{-1} P_j O v \\
  v & = i (I + A)^{-1} (Y - \eta' + i O)^{-1} P_j O v \numberthis{} \label{eq:nenciu-ext-2d-eq5}
\end{align*}
To reduce the number of terms in the next steps, let's define
\[
  C := (I + A)^{-1} (Y - \eta' + i O)^{-1}.
\]
With this definition Equation \eqref{eq:nenciu-ext-2d-eq5} becomes
\[
  v = i C P_j O v.
\]
Recalling that we chose $O := b \Pi + \vert{}X - \eta\vert{}$, we have that
\begin{align*}
  v & = i C P_j O v \\
    & = i C P_j (b\Pi + \vert{}X - \eta\vert{}) v \\
    & = i b C P_j \Pi v + i C P_j \vert{}X - \eta\vert{} v \\
  ( I - i C P_j \vert{}X - \eta\vert{} ) v & = ib C P_j \Pi v. \numberthis{} \label{eq:nenciu-ext-2d-eq6}
\end{align*}
Similar to before, we would like to invert the operator $( I - i C P_j \vert{}X - \eta\vert{} )$. Recall that if
\[
  b > 4 \Big(\| [P_j, Y] \| + \|P_j\vert{}X - \eta\vert{} \| + \|\vert{}X - \eta\vert{} P_j\| \Big),
\]
then $\| A \| \leq \frac{3}{4}$ so we have that
\begin{align*}
  \| C \| & = \| (I + A)^{-1} (Y - \eta' + i O)^{-1} \| \\
          & \leq \| (I + A)^{-1} \|  \| (Y - \eta' + i O)^{-1} \| \\
          & \leq 4 b^{-1}.
\end{align*}
Therefore,
\[
  \| i C P_j \vert{}X - \eta\vert{} \| \leq 4 b^{-1} \| P_j \vert{}X - \eta\vert{} \|.
\]
Since we have chosen $b > 4 \| P_j \vert{}X - \eta\vert{} \|$, the operator $( I - i C P_j \vert{}X - \eta\vert{} )$ is invertible.

Using this fact allows us to rewrite Equation \eqref{eq:nenciu-ext-2d-eq6} as
\begin{equation}
  \label{eq:nenciu-ext-2d-eq7}
  v = i b \big( I - i C P_j \vert{}X - \eta\vert{} \big)^{-1} C P_j \Pi v.
\end{equation}
After all of these algebraic steps, we have been able to derive an expression for $v$ as the product of a bounded operator and a cutoff function. The final step in this argument will be to multiply both sides of Equation \eqref{eq:nenciu-ext-2d-eq7} by the exponential growth operator $B_{\gamma}$ and show that the result is bounded. The inclusion of the cutoff function is what makes it possible to control this multiplication because $B_{\gamma} \Pi$ is bounded.

At least formally, we can multiply both sides of Equation \eqref{eq:nenciu-ext-2d-eq7} by $B_{\gamma}$ and insert copies of $B_{\gamma}^{-1} B_{\gamma}$ to get
\begin{align*}
  B_{\gamma} v & = ib B_{\gamma} \big(I - i C P_j \vert{} X - \eta \vert{} \big)^{-1} C P_j \Pi v \\[.5ex]
               & = ib B_{\gamma} \big(I - i C P_j \vert{} X - \eta \vert{} \big)^{-1} \Big(B_{\gamma}^{-1} B_{\gamma}\Big) C \Big( B_\gamma^{-1} B_\gamma \Big) P_j \Big( B_{\gamma}^{-1} B_{\gamma} \Big) \Pi v \\[.5ex]
               & = ib \bigg[ ( I - i C_{\gamma} P_{j,\gamma} \vert{}X - \eta\vert{} )^{-1} \bigg] \bigg[ C_{\gamma} P_{j,\gamma} \bigg] \bigg[B_{\gamma} \Pi\bigg] v,
\end{align*}
where we have used our convention for exponentially tilted operators $P_{j,\gamma} := B_{\gamma} P_j B_{\gamma}^{-1}$ and similarly for $C$. We will now show each of the bracketed terms are bounded.

The easiest term to bound is the last term, $B_{\gamma} \Pi$. Let's recall the definition of $B_{\gamma}$:
\[
  B_{\gamma, (a, b)} = \exp{\Big( \gamma \sqrt{1 + (X - a)^2 + (Y - b)^2} \Big)}.
\]
While we have ignored the center point $(a, b)$ thus far in the argument, here we will explicitly choose $(a, b) = (\eta, \eta')$. Since $\Pi = \Pi^X_{[\eta - b, \eta + b]} \Pi^Y_{[\eta' - b, \eta' + b]}$ we clearly have:
\[
  \| B_{\gamma, (\eta, \eta')} \Pi \| \leq e^{\gamma \sqrt{1 + 2b^2}}.
\]

To show that the first two terms are bounded we will show that, for $\gamma$ sufficiently small, $\|C_{\gamma}\| = O(b^{-1})$. Once we show this, the second term, $C_{\gamma} P_{j,\gamma}$, is clearly bounded and we may pick $b$ sufficiently large so that the first term is also bounded.

By definition, we have:
\begin{align*}
  C_{\gamma} & = B_{\gamma} ( I + A )^{-1}  (Y - \eta' + i O)^{-1} B_{\gamma}^{-1} \\
             & = B_{\gamma} \Big( I + (Y - \eta' + iO)^{-1} [P_j, Y + iO] \Big)^{-1}  (Y - \eta' + i O)^{-1} B_{\gamma}^{-1} \\
             & = \Big( I + (Y - \eta' + iO)^{-1} [P_{j,\gamma}, Y + iO] \Big)^{-1} (Y - \eta' + i O)^{-1} .
\end{align*}
For the above calculations to make sense, it suffices to show that
\[
  \| (Y - \eta' + iO)^{-1} [P_{j,\gamma}, Y + iO] \| \leq \frac{3}{4}.
\]
Since $\|P_{j,\gamma} - P_j\| \leq K''_1 \gamma$ we have that:
\begin{align*}
  \| (Y -& \eta' + iO)^{-1}  [P_{j,\gamma}, Y + iO] \|   \\
         & \leq \| (Y - \eta' + iO)^{-1} \| \Big( \| [ P_{j, \gamma}, Y] \| + b \| [ P_{j,\gamma}, \Pi ] \| + \| [ P_{j,\gamma}, \vert{}X - \eta\vert{} ] \| \Big) \\
         & \leq  b^{-1} \Big( \| [ P_{j, \gamma}, Y] \| + b \| [ P_{j,\gamma} - P + P, \Pi ] \| + \| [ P_{j,\gamma}, \vert{}X - \eta\vert{} ] \| \Big) \\
         & \leq  b^{-1} \Big( \| [ P_{j, \gamma}, Y] \| + b \| [ P_{j,\gamma} - P, \Pi ] \| + b \| [ P, \Pi ] \| + \| [ P_{j,\gamma}, \vert{}X - \eta\vert{} ] \| \Big) \\
         & \leq  b^{-1} \bigg( \frac{1}{2} b + b \|P_{j,\gamma} - P\|  +  \| [ P_{j, \gamma}, Y] \| + \| P_{j,\gamma} \vert{} X - \eta\vert{} \| + \|  \vert{} X - \eta\vert{} P_{j,\gamma}\|  \bigg) \\
         & \leq \frac{1}{2} + K''_1 \gamma + b^{-1} \Big( \| [ P_{j, \gamma}, Y] \| + \| P_{j,\gamma} \vert{} X - \eta\vert{} \| + \|  \vert{} X - \eta\vert{} P_{j,\gamma}\|  \Big).
\end{align*}
Therefore, if we pick $\gamma \leq (8 K''_1)^{-1}$ and
\[
  b \geq 8 \Big( \| [ P_{j, \gamma}, Y] \| + \| P_{j,\gamma} \vert{} X - \eta\vert{} \| + \|  \vert{} X - \eta\vert{} P_{j,\gamma}\|  \Big),
\]
we have that
\[
  \| P_{j,\gamma} (Y - \eta' + iO)^{-1} [P_{j,\gamma}, Y + iO] \| \leq \frac{3}{4}.
\]
Therefore, for these choices of $b$ and $\gamma$ we have that:
\begin{align*}
  \|C_{\gamma}\| & \leq \|\Big( I + P_{j,\gamma} (Y - \eta' + iO)^{-1} [P_{j,\gamma}, Y + iO] \Big)^{-1} \| \| (Y - \eta' + i O)^{-1} \| \\
                   & \leq 4 b^{-1}.
\end{align*}
Now recall the original equation we wanted to bound:
\[
  B_{\gamma} v = i \bigg[ (I - i  C_{\gamma} P_{j,\gamma} \vert{} X - \eta \vert{} )^{-1} \bigg] \bigg[ b C_{\gamma} P_{j,\gamma} \bigg] \bigg[ B_{\gamma} \Pi \bigg] v.
\]
Since $b \geq 8 \| P_{j,\gamma} \vert{} X - \eta\vert{} \|$ we have that
\[
  \|C_{\gamma} P_{j,\gamma} \vert{} X - \eta \vert{} \| \leq 4 b^{-1} \| P_{j,\gamma} \vert{} X - \eta\vert{} \| \leq \frac{1}{2}.
\]
Therefore,
\[
  \|(I - i  C_{\gamma} P_{j,\gamma} \vert{} X - \eta \vert{} )^{-1} \| \leq 2,
\]
so combining all of our bounds together gives:
\begin{align*}
  \|B_{\gamma, (\eta, \eta')} v\| & \leq \Big[2\Big] \Big[ 4(1 + K''_2 \gamma) \Big] \Big[ e^{\gamma\sqrt{1 + 2 b^2}} \Big] \\
                                  & \leq 16 e^{\gamma\sqrt{1 + 2 b^2}},
\end{align*}
so long as
\[
  \begin{array}{l}
    \gamma \leq (8 K''_2)^{-1} \\[1ex]
    b \geq 8 \Big( \| [ P_{j, \gamma}, Y] \| + \| P_{j,\gamma} \vert{} X - \eta\vert{} \| + \| \vert{} X - \eta\vert{} P_{j,\gamma}\|  \Big).
  \end{array}
\]
This proves Lemma \ref{lem:exp_decay}.

\section{Extensions to three dimensions and higher} \label{sec:higher-d}
The proof of Theorem \ref{th:main_theorem} generalizes to arbitrarily high dimensions under appropriate generalizations of the uniform spectral gaps assumption (Assumption \ref{def:usg}) by an inductive procedure.

Assume regularity and spectral gap assumptions analogous to Assumptions \ref{as:H_assump} and \ref{as:gap_assump}, and consider position operators $X$, $Y$, and $Z$ acting on $\mathcal{H}: = L^2(\field{R}^3)$ along directions corresponding to a three-dimensional basis. Let $P$ be the Fermi projection, and consider the operator $PXP$. Assume $PXP$ has uniform spectral gaps in the sense of Assumption \ref{def:usg}, and let $P_j$ denote spectral projections onto each of the separated components of the spectrum of $PXP$. By repeating the proof of Proposition \ref{prop:pj-props}, we can show that each $P_j$ admits an exponentially-localized kernel and is localized in $X$. Because of these properties, by assuming that $PXP$ has uniform spectral gaps, we have reduced the original three dimensional system defined by $P$, to an infinite family of quasi-two dimensional systems defined by the projectors $\{ P_j \}_{j \in \cJ}$.

If we additionally assume the operators $P_j Y P_j$ \emph{also} have uniform spectral gaps in the sense of Assumption \ref{def:usg}, and let $P_{j,k}$ denote spectral projections onto each of the separated components of the spectrum of $P_j Y P_j$. By analogous reasoning to the two dimensional case, functions in $\range(P_{j,k})$ are quasi-one dimensional in the sense that they decay away from lines $x = c_1$, $y = c_2$ for constants $c_1, c_2$. We therefore claim that the set of eigenfunctions of the operator $P_{j,k} Z P_{j,k}$ will form an exponentially localized basis of $\range(P_{j,k})$ for each $j, k$, and that the union of all of these eigenfunctions over $j$ and $k$ will form an exponentially-localized basis of $\range(P)$.

To make the above sketch rigorous, there are a few important steps in the proof which must be checked. We will discuss each step in turn.

\subsubsection*{Proving bounds on $P_{j,k}$} First, we must prove that $P_{j,k}$ satisfies estimates analogous to Proposition \ref{prop:pj-props}. To do this, note that when we proved Proposition \ref{prop:pj-props} we only used that $P$ has an exponentially localized kernel and $PXP$ has uniform spectral gaps. Since we have shown that $P_j$ has an exponentially localized kernel, it follows that under a uniform gap assumption on $P_j Y P_j$ we can prove that $P_{j,k}$ has an exponentially localized kernel and is localized in both $X$ and $Y$.

\subsubsection*{Proving $P_{j,k} Z P_{j,k}$ has compact resolvent} To prove $P_{j,k} Z P_{j,k}$ has compact resolvent, mimicking the calculations preceding \eqref{eq:pj_yinv_pj}, it is sufficient to prove that $P_{j,k} (Z + i \mu)^{-1} P_{j,k}$ is compact for each $j, k$. We will first show how a na\"ive generalization of the proof that $P_j Y P_j$ is compact in two dimensions fails, and then present a correct generalization. Just as in equation \eqref{eq:pj_yinv_pj}, we can write
\begin{equation} \label{eq:pjk_zinv}
\begin{split}
  P_{j,k} (Z \pm i \mu)^{-1} P_{j,k}
       = &\left[ P_{j,k} (Z \pm i \mu)^{-1} (X - \eta_j + i)^{-1} (Y - \eta_k + i)^{-1} \right] \\
       &\times \left[ (Y - \eta_k + i) (X - \eta_j + i) P_{j,k} \right] 
\end{split}
\end{equation}
where $\eta_j \in \sigma_j$, where $\sigma_j$ is the $j$th separated component of $\sigma(PXP)$, and $\eta_k \in \sigma_{j,k}$, where $\sigma_k$ is the $k$th separated component of $\sigma(P_j Y P_j)$. To prove $P_{j,k} (Z + i \mu )^{-1} P_{j,k}$ is compact, we must prove that the first term in \eqref{eq:pjk_zinv} is compact, while the second term is bounded. That the first term is compact follows from essentially the same argument as given in Section \ref{sec:pjypj-disc-spec}. Unfortunately, it is unclear if the last term
\begin{equation*}
    (Y - \eta_k + i) (X - \eta_j + i) P_{j,k} 
\end{equation*}
is bounded. The trick is to write, instead of \eqref{eq:pjk_zinv},
\begin{equation} \label{eq:pjk_zinv_2}
\begin{split}
  P_{j,k} (Z \pm i \mu)^{-1} P_{j,k}
       = &\left[ P_{j,k} (Z \pm i \mu)^{-1} (X - \eta_j + i)^{-1/2} (Y - \eta_k + i)^{-1/2} \right] \\
       &\times \left[ (Y - \eta_k + i)^{1/2} (X - \eta_j + i)^{1/2} P_{j,k} \right].
\end{split}
\end{equation}
That the first term of \eqref{eq:pjk_zinv_2} is compact follows from an almost identical argument as given in Section \ref{sec:pjypj-disc-spec}. We now show that the second term of \eqref{eq:pjk_zinv_2} is bounded. Note that if $f \notin \range(P_{j,k})$ then the operator acting on $f$ is clearly bounded. So let $f \in \range(P_{j,k})$. Then using the fact that the geometric mean is bounded by the arithmetic mean, we have that
\begin{equation*}
\begin{split}
   &\left\| (Y - \eta_k + i)^{1/2} (X - \eta_j + i)^{1/2} P_{j,k} f \right\|^2 \\
   &\leq \frac{1}{2} \left( \left\| (Y - \eta_k + i) P_{j,k} f \right\|^2 + \left\| (X - \eta_j + i) P_{j,k} f \right\|^2 \right).
\end{split}
\end{equation*}
The first term is bounded since $P_{j,k}$ is the projection onto a bounded subset of the spectrum of $P_j Y P_j$. The second term is bounded since $P_{j,k} = P_j P_{j,k}$ and $(X + \eta_j + i) P_j$ is bounded since $P_j$ is the projection onto a bounded subset of the spectrum of $P X P$.

\subsubsection*{Proving exponential localization of eigenfunctions of $P_{j,k} Z P_{j,k}$} The generalization of the proof of Lemma \ref{lem:exp_decay} to three dimensions is straightforward once we prove that $P_{j,k}$ admits an exponentially localized kernel. The only modification necessary is that in three dimensions the choice of the operator $O$ \eqref{eq:def_O} must be changed to
\begin{equation*}
    O = b \Pi^X_{[\eta_j - b, \eta_j + b]} \Pi^Y_{[\eta_{j,k} - b, \eta_{j,k} + b]} \Pi^Z_{[\eta_{j,k,l} - b, \eta_{j,k,l} + b]} + \vert{}X - \eta_j\vert{} + \vert{}Y - \eta_{j,k} \vert{},
\end{equation*}
where $\eta_j \in \sigma_j$, the $j$th component of $\sigma(PXP)$, $\eta_{j,k} \in \sigma_{j,k}$, the $k$th component of $\sigma(P_j Y P_j)$, and $\eta_{j,k,l}$ denotes the $l$th eigenvalue of $P_{j,k} Z P_{j,k}$.

\section{Generalization to discrete models} \label{sec:discrete}

In this section we sketch the generalization of our proofs to discrete models whose entries decay exponentially away from the diagonal. We give details for the case of models in $\ell^2(\field{Z}^2)$ when the sites are arranged on a square lattice with lattice constant $1$, but further generalizations are straightforward; see Remark \ref{rem:discrete_generalizations}. In fact, the generalization to discrete models is very simple because of the operator-theoretic structure of our proofs. Our proofs only use the specific form of the Hamiltonian in order to prove that the Fermi projection admits an exponentially localized kernel (Property \eqref{eq:exp_loc_proj_property}). For discrete space, this property is equivalent to exponential decay of the entries of the Fermi projection operator away from the diagonal. Hence our results generalize easily using the following lemma. 
\begin{lemma} \label{lem:discrete_lemma}
  For each $\vec{\lambda} = (\lambda_1, \lambda_2) \in \field{Z}^2$, let $e_{\vec{\lambda}} \in \ell^2(\field{Z}^2)$ denote a joint eigenvector of the position operators $X$ and $Y$ with eigenvalue $\lambda_1$ and $\lambda_2$ respectively:
    \begin{equation} \label{eq:XY_coords_discrete}
    X e_{\vec{\lambda}} = \lambda_1 e_{\vec{\lambda}} \qquad Y e_{\vec{\lambda}} = \lambda_2 e_{\vec{\lambda}}.
    \end{equation} 
  Furthermore, let $\vert{} \cdot \vert{}$ denote the Euclidean 2-norm on $\field{Z}^2$. That is, $\vert{} \vec{\lambda} \vert{} := \sqrt{\lambda_1^2 + \lambda_2^2}$.

  Next, let $H$ be a self-adjoint operator on $\ell^2(\field{Z}^2)$ with a spectral gap containing the Fermi level and $P$ be the Fermi projection. Suppose further that for some finite, positive constants $(\gamma', C)$ and for any $\vec{\lambda}, \vec{\mu} \in \field{Z}^2$:
    \begin{equation} \label{eq:discrete_decay}
    \vert{}\la e_{\vec{\lambda}}, H e_{\vec{\mu}} \ra\vert{} \leq C e^{-\gamma' \vert{} \vec{\lambda} - \vec{\mu} \vert{}}.
    \end{equation}
  Under these assumptions, there exist finite, positive constants $(\gamma'', K)$ depending only on $H$ so that for all $\gamma \leq \gamma''$:
\begin{equation}
    \vert{} \la e_{\vec{\lambda}}, P e_{\vec{\mu}} \ra \vert{} \leq K e^{- \gamma \vert{} \vec{\lambda} - \vec{\mu} \vert{}}.
\end{equation}
\end{lemma}
\begin{remark} \label{rem:discrete_generalizations}
    Equation \eqref{eq:XY_coords_discrete} assumes that the sites of the discrete model are arranged in a square lattice with lattice constant $1$. For sites arranged in general Bravais lattices, we replace $\vec{\lambda}$ and $\vec{\mu}$ everywhere in the lemma by the real space co-ordinates of the relevant sites. We can deal similarly with multilattices, or systems with internal degrees of freedom (where the model Hilbert space is $\ell^2(\field{Z}^2;\field{C}^N)$ for some positive integer $N$). We can replace $\field{Z}^2$ by arbitrary point sets $\mathcal{D} \subset \field{R}^2$ similarly, as long as $\mathcal{D}$ has no accumulation points. If $\mathcal{D}$ has accumulation points, the decay assumption \eqref{eq:discrete_decay} no longer guarantees rapid decay of the matrix elements of $H$ away from the diagonal. In this case, estimate \eqref{eq:disc_breaks} fails. 
\end{remark}
\begin{proof}[Proof of Lemma \ref{lem:discrete_lemma}]

To prove this result, it suffices to show that there exist a constant $C$ so that $\| P_{\gamma} \|  \leq C$. To see why, observe that by the definition of the spectral norm we have that:
\[
\| P_{\gamma} \| = \sup_{\| f \|, \| g \| = 1} \vert{} \la f, P_{\gamma} g \ra\vert{}
\]
Hence, choosing $f =  e_{\vec{\lambda}}$ and $g =  e_{\vec{\mu}}$ we have that for all $\vec{a} \in \field{R}^2$:
\begin{align*}
& \vert{} \la e_{\vec{\lambda}}, B_{\gamma,\vec{a}} P B_{\gamma,\vec{a}}^{-1} e_{\vec{\mu}} \ra\vert{} \leq C \\
& \Longrightarrow  e^{\gamma \la \vec{\lambda} - \vec{a} \ra} e^{-\gamma \la \vec{\mu} - \vec{a} \ra} \vert{} \la e_{\vec{\lambda}}, P e_{\vec{\mu}} \ra\vert{} \leq C
\end{align*}
Since the choice of $\vec{a}$ is arbitrary, we can pick $\vec{a} = \vec{\mu}$ to conclude that
\[
 \vert{} \la e_{\vec{\lambda}}, P e_{\vec{\mu}} \ra\vert{} \leq C e^{-\gamma \la \vec{\lambda} - \vec{\mu} \ra}  \leq C' e^{-\gamma \vert{} \vec{\lambda} - \vec{\mu} \vert{}} 
\]
which completes the proof.

To show that $\| P_{\gamma} \| \leq C$, we use the Riesz formula and observe that
\begin{align*}
P_{\gamma}
& = \frac{1}{2 \pi i} \int_{\cC} B_{\gamma} (\lambda - H)^{-1} B_{\gamma}^{-1} = \frac{1}{2 \pi i} \int_{\cC} (\lambda - H_{\gamma})^{-1} 
\end{align*}
Hence
\[
\| P_{\gamma} \| \leq \frac{\ell(\cC)}{2 \pi} \left( \sup_{\lambda \in \cC} \| (\lambda - H_{\gamma})^{-1}  \| \right)
\]
and it so suffices to show that $(\lambda - H_{\gamma})^{-1}$ is a bounded operator for all $\lambda \in \cC$. We observe that
\begin{align*}
(\lambda - H_{\gamma})^{-1}
& = (\lambda - H + H - H_{\gamma})^{-1} \\
& = (\lambda - H)^{-1} \Big(I - (H_{\gamma} - H) (\lambda - H)^{-1} \Big)^{-1}
\end{align*}
So since $(\lambda - H)^{-1}$ is bounded by a constant for all $\lambda \in \cC$, if we can show that $\| H_{\gamma} - H \| = O(\gamma)$ then we can pick $\gamma$ sufficiently small so that $\| (H_{\gamma} - H) (\lambda - H)^{-1} \| \leq \frac{1}{2}$ for $\lambda \in \cC$. This in turn implies $ \sup_{\lambda \in \cC} \| (\lambda - H_{\gamma})^{-1} \|$ is bounded which completes the proof.

We now show that $\| H_{\gamma} - H \| = O(\gamma)$. Since $\{ e_{\vec{\lambda}} \}_{\vec{\lambda} \in \field{Z}^2}$ forms an orthonormal basis for $\ell^2(\field{Z}^2)$, we can expand the operator $H_{\gamma} - H$ in terms of this basis. In particular,
\[
H_{\gamma} - H = \sum_{\vec{\lambda} \in \field{Z}^2} \sum_{\vec{\mu} \in \field{Z}^2}  \Big( e^{\gamma \vert{} \vec{\lambda} - \vec{a} \vert{}} e^{-\gamma \vert{}\vec{\mu} - \vec{a} \vert{}} - 1 \Big)     \la e_{\vec{\lambda}}, H e_{\vec{\mu}} \ra \ket{e_{\vec{\mu}}} \bra{e_{\vec{\lambda}}}
\]
We will now use Schur's test to show that $\| H_{\gamma} - H \| = O(\gamma)$. We define
\[
\begin{split}
& \alpha := \sup_{\lambda} \sum_{\vec{\mu} \in \field{Z}^2}  \vert{} e^{\gamma \vert{} \vec{\lambda} - \vec{a} \vert{}} e^{-\gamma \vert{}\vec{\mu} - \vec{a} \vert{}} - 1 \vert{}  \vert{}\la e_{\vec{\lambda}}, H e_{\vec{\mu}} \ra\vert{} \\
& \beta := \sup_{\mu} \sum_{\vec{\lambda} \in \field{Z}^2}  \vert{} e^{\gamma \vert{} \vec{\lambda} - \vec{a} \vert{}} e^{-\gamma \vert{}\vec{\mu} - \vec{a} \vert{}} - 1 \vert{}  \vert{}\la e_{\vec{\lambda}}, H e_{\vec{\mu}} \ra\vert{}
\end{split}
\]
and by Schur's test $\| H_{\gamma} - H \| \leq \sqrt{\alpha \beta}$. By basic calculus one can easily verify that that
\[
  \vert{} e^{\gamma \vert{} \vec{\lambda} - \vec{a} \vert{}} e^{-\gamma \vert{}\vec{\mu} - \vec{a}\vert{}} - 1 \vert{} \leq \gamma \vert{} \vec{\lambda} - \vec{\mu} \vert{} e^{\gamma \vert{}\vec{\lambda} - \vec{\mu} \vert{}}.
\]
Hence
\[
\begin{split}
& \alpha \leq \sup_{\lambda} \sum_{\vec{\mu} \in \field{Z}^2}  \gamma \vert{} \vec{\lambda} - \vec{\mu} \vert{} e^{\gamma \vert{}\vec{\lambda} - \vec{\mu} \vert{}}  \vert{}\la e_{\vec{\lambda}}, H e_{\vec{\mu}} \ra\vert{} \\
& \beta \leq \sup_{\mu} \sum_{\vec{\lambda} \in \field{Z}^2}  \gamma \vert{} \vec{\lambda} - \vec{\mu} \vert{} e^{\gamma \vert{}\vec{\lambda} - \vec{\mu} \vert{}}  \vert{}\la e_{\vec{\lambda}}, H e_{\vec{\mu}} \ra\vert{}
\end{split}
\]
Since these upper bounds on $\alpha$ and $\beta$ are symmetric in $\vec{\lambda}$ and $\vec{\mu}$, to show $\| H_{\gamma} - H \| = O(\gamma)$ it suffices to show that there exists a constant $C'$ so that
\[
\sup_{\lambda} \sum_{\vec{\mu} \in \field{Z}^2} \vert{} \vec{\lambda} - \vec{\mu} \vert{} e^{\gamma \vert{}\vec{\lambda} - \vec{\mu} \vert{}}  \vert{}\la e_{\vec{\lambda}}, H e_{\vec{\mu}} \ra\vert{} \leq C'.
\]
This however is immediate from the fact that $H$ is exponentially localized:
\begin{equation} \label{eq:disc_breaks}
\sum_{\vec{\mu} \in \field{Z}^2} \vert{} \vec{\lambda} - \vec{\mu} \vert{} e^{\gamma \vert{}\vec{\lambda} - \vec{\mu} \vert{}}  \vert{}\la e_{\vec{\lambda}}, H e_{\vec{\mu}} \ra\vert{} \leq C \sum_{\vec{\mu} \in \field{Z}^2} \vert{} \vec{\lambda} - \vec{\mu} \vert{} e^{-(\gamma' - \gamma) \vert{}\vec{\lambda} - \vec{\mu} \vert{}} 
\end{equation}
The right-hand side is clearly bounded by a constant independent of $\vec{\lambda}$ so long as $\gamma < \gamma'$ and hence the result is proved.
\end{proof}

\appendix

\section{Proof of Properties of $P_j$ (Proposition \ref{prop:pj-props})}
\label{sec:pj-props}
It turns out that both of the estimates in Proposition \ref{prop:pj-props} will easily follow by showing that the resolvent of $PXP$ is exponentially localized in the following technical sense.
\begin{proposition}
  \label{prop:pj-res}
	Suppose that $P$ is an exponentially localized orthogonal projector and $PXP$ has uniform spectral gaps with decomposition $\{ \sigma_j \}_{j \in \cJ}$ and corresponding contours $\{ \cC_j \}_{j \in \cJ}$ (see Definition \ref{def:band-projectors}). There exists finite, positive constants $(C, \gamma^*)$ so that for all $0 \leq \gamma \leq \gamma^*$
	\[
	\sup_{j \in \cJ} \sup_{\lambda \in \cC_j} \| B_{\gamma} (\lambda - PXP)^{-1} B_{\gamma}^{-1} \| = \sup_{j \in \cJ} \sup_{\lambda \in \cC_j} \| (\lambda - P_{\gamma} X P_{\gamma})^{-1} \| \leq C.
	\]	
\end{proposition}
We will prove Proposition \ref{prop:pj-res} in Section \ref{sec:exp-loc-res}. We will then use Proposition \ref{prop:pj-res} to show that $P_j$ admits an exponentially localized kernel in Section \ref{sec:exp-loc-kernel} and that $P_j$ is localized along lines where $X = \eta_j$ in Section \ref{sec:pj-x-loc}. The analysis given in Sections \ref{sec:exp-loc-res}, \ref{sec:exp-loc-kernel}, and \ref{sec:pj-x-loc} is quite general and can be generalized to choices of position operator other than $X$ and $Y$. We show how our analysis can be extended to a wide class of self-adjoint position operators in Section \ref{sec:extension-to-other}.

\subsection{Proof that the resolvent of $P_j$ is exponentially localized (Proposition \ref{prop:pj-res})}
\label{sec:exp-loc-res}
One key part of Proposition \ref{prop:pj-res} is that the bound on $\| (\lambda - P_{\gamma} X P_{\gamma})^{-1} \|$ is \textit{uniform} in the choice of $j \in \cJ$ as well as the element of the contour $\lambda \in \cC_j$. By applying a naive argument, typically one finds that many of the estimates depend on $\vert{}\lambda\vert{}$. Since $\vert{}\lambda\vert{}$ can be arbitrarily large, this naive argument does not allow us to prove Proposition \ref{prop:pj-res} with uniform constants. To correct this issue we introduce a technical result, the shifting lemma (Lemma \ref{lem:shifting}), which allows one to effectively shift the contour $\cC_j$ so that instead of getting a dependence of $\vert{}\lambda\vert{}$ we have a dependence of $\vert{} \lambda - \eta \vert{}$ where $\eta$ is an arbitrary element from $\sigma_j$. If $\lambda \in \cC_j$ and $\eta \in \sigma_j$ then the difference $\vert{} \lambda - \eta \vert{}$ is bounded by the diameter of the contour $\cC_j$ and therefore bounded by a constant uniform in $j \in \cJ$.

The core idea underlying the shifting lemma is the following simple calculation. First, recall that if we are given a projector $P$, we can define the complementary projector $Q := I - P$. Since $P$ and $Q$ act on orthogonal subspaces we should expect that for any $\eta \neq \lambda$:
\begin{align*}
(\lambda - PXP)^{-1} P 
& = \Big((\lambda - \eta + \eta) - P(X - \eta + \eta)P\Big)^{-1} P \\
& = \Big((\lambda - \eta) - P(X - \eta)P + \eta (I - P)\Big)^{-1} P \\
& = \Big((\lambda - \eta) - P(X - \eta)P + \eta Q\Big)^{-1} P \\
& = \Big((\lambda - \eta) - P(X - \eta)P\Big)^{-1} P 
\end{align*}
By using this calculation, we are able to replace $\lambda$ with $\lambda - \eta$ which leads to uniform bounds as discussed above. Importantly, the above calculation does not require periodicity of the underlying system. By using some variations on this calculation, one can show the following lemma:
\begin{lemma}[Shifting Lemma]
  \label{lem:shifting}
  Suppose $P$ admits an exponentially localized kernel and suppose that $PXP$ has uniform spectral gaps with decomposition $\{ \sigma_j \}_{j \in \cJ}$ and corresponding contours $\{ \cC_j \}_{j \in \cJ}$. 
  Then there exists a $\gamma^*$ so that the following are equivalent for all $0 \leq \gamma \leq \gamma^*$:
  \begin{enumerate}[itemsep=1ex,topsep=1.5ex]
  \item There exists a $C > 0$, independent of $j$, such that
    \[
      \sup_{\lambda \in \cC_j} \|(\lambda - P_{\gamma} X P_{\gamma})^{-1}\| \leq C.
    \]
  \item There exists a $C' > 0$, independent of $j$, such that for each $j \in \cJ$:
    \[
      \sup_{\lambda \in \cC_j} \sup_{\eta_j \in \sigma_j} \|(\lambda_{\eta_j} - P_{\gamma} X_{\eta_j} P_{\gamma})^{-1}\| \leq C'
    \]
  \end{enumerate}
  Furthermore, for any $0 \leq \gamma \leq \gamma^*$ if either $\|(\lambda - P_{\gamma} X P_{\gamma})^{-1}\|$ or $\|(\lambda_{\eta_j} - P_{\gamma} X_{\eta_j} P_{\gamma})^{-1}\|$ are bounded, we have for any $j \in \cJ$, $\lambda \in \cC_j$ and any $\eta_j \in \sigma_j$:
  \[
    (\lambda - P_{\gamma} X P_{\gamma})^{-1} P_{\gamma} = (\lambda_{\eta_j} - P_{\gamma} X_{\eta_j} P_{\gamma})^{-1} P_{\gamma}.
  \]
\end{lemma}
\begin{proof}
  Given in Appendix \ref{sec:shifting-proof}.
\end{proof}
As a consequence of Lemma \ref{lem:shifting}, to prove Proposition \ref{prop:pj-res} it is enough to fix some $j \in \cJ$, $\lambda \in \cC_j$, and $\eta_j \in \sigma_j$ and show that
\begin{equation}
\label{eq:shifted-res}
\| (\lambda_{\eta_j} - P_{\gamma} X_{\eta_j} P_{\gamma})^{-1}\| \leq C'
\end{equation}
where the constant $C'$ is independent of the choice of $j$, $\lambda$, and $\eta_j$. For the remainder of this section, we will fix a choice of $j \in \cJ$, $\lambda \in \cC_j$, and $\eta_j \in \sigma_j$ and prove Equation \eqref{eq:shifted-res}.

The path to proving Equation \eqref{eq:shifted-res} is to use the following chain of implications (where USG is an abbreviation for uniform spectral gaps):
\begin{equation}
  \label{eq:chain}
  \begin{split}
  P X &P \text{ has USG} \Longrightarrow \| (\lambda_{\eta_j} - P X_{\eta_j} P )^{-1} \| < \infty \\
  & \Longrightarrow \| (\lambda_{\eta_j} - P X_{\eta_j} P_{\gamma} )^{-1} \| < \infty \Longrightarrow \| (\lambda_{\eta_j} - P_{\gamma} X_{\eta_j} P_{\gamma} )^{-1} \| < \infty
\end{split}
\end{equation}
In words, since $P X P$ has uniform spectral gaps, we know that $\| (\lambda_{\eta_j} - P X_{\eta_j} P )^{-1} \| < \infty$ for all $\lambda \in \cC_j$. Using this fact, along with the fact that $P$ is exponentially localized, we can conclude that $\| (\lambda_{\eta_j} - P X_{\eta_j} P_{\gamma} )^{-1} \| < \infty$ for all $\gamma$ sufficiently small. Finally, once we know that $(\lambda_{\eta_j} - P X_{\eta_j} P_{\gamma} )^{-1}$ is bounded, we can use that estimate to show that $(\lambda_{\eta_j} - P_{\gamma} X_{\eta_j} P_{\gamma})^{-1}$ is bounded for all $\gamma$ sufficiently small. This completes the proof of the proposition.

The first implication, that uniform spectral gaps implies $(\lambda_{\eta_j} - P X_{\eta_j} P )^{-1}$ is bounded is an immediate consequence of Lemma \ref{lem:shifting} by choosing $\gamma = 0$. Therefore, we only need to show the last two implications. We will prove the second implication in Section \ref{sec:exp-loc-res-pt1} and the final implication in Section \ref{sec:exp-loc-res-pt2}.

\subsubsection{Proof that shifted $(\lambda - P X P_{\gamma} )^{-1}$ is bounded}
\label{sec:exp-loc-res-pt1}
By adding and subtracting $P X_{\eta_j} P$ in the shifted resolvent we have that formally:
\begin{align*}
(\lambda_{\eta_j} & - P X_{\eta_j} P_{\gamma} )^{-1}  \\
& = (\lambda_{\eta_j} - P X_{\eta_j} P + P X_{\eta_j} P - P X_{\eta_j} P_{\gamma} )^{-1} \\
& = \Big(\lambda_{\eta_j} - P X_{\eta_j} P - P X_{\eta_j} (P_{\gamma} - P) \Big)^{-1} \\
& = \Big(I - (\lambda_{\eta_j} - P X_{\eta_j} P)^{-1} P X_{\eta_j} (P_{\gamma} - P) \Big)^{-1} (\lambda_{\eta_j} - P X_{\eta_j} P)^{-1}.
\end{align*}
Since $P$ admits an exponentially localized kernel, it is easy to verify that exists a constant $C$ so that for all $\gamma$ sufficiently small
\[
\| P_{\gamma} - P \| \leq C \gamma.
\]
Therefore, if we can show that $\| (\lambda_{\eta_j} - P X_{\eta_j} P)^{-1} P X_{\eta_j} \|$ is bounded by an absolute constant, we can choose $\gamma$ sufficiently small so that
\[
\| (\lambda_{\eta_j} - P X_{\eta_j} P)^{-1} P X_{\eta_j} (P_{\gamma} - P) \| \leq \frac{1}{2}.
\]
This implies that $(\lambda_{\eta_j} - P X_{\eta_j} P_{\gamma} )^{-1}$ is bounded since
\begin{align*}
\| (\lambda_{\eta_j} & - P X_{\eta_j} P_{\gamma} )^{-1} \| \\
& \leq \| \Big(I - (\lambda_{\eta_j} - P X_{\eta_j} P)^{-1} P X_{\eta_j} (P_{\gamma} - P) \Big)^{-1} \| \|(\lambda_{\eta_j} - P X_{\eta_j} P)^{-1} \| \\
& \leq 2 \|(\lambda_{\eta_j} - P X_{\eta_j} P)^{-1} \|.
\end{align*}
To show that $\| (\lambda_{\eta_j} - P X_{\eta_j} P)^{-1} P X_{\eta_j} \|$ is bounded, we recall that $I = P + Q$ and so
\begin{align*}
(\lambda_{\eta_j} & - P X_{\eta_j} P)^{-1} P X_{\eta_j} \\
& = (\lambda_{\eta_j} - P X_{\eta_j} P)^{-1} \Big( P X_{\eta_j} P + P X_{\eta_j} Q \Big) \\
& = (\lambda_{\eta_j} - P X_{\eta_j} P)^{-1} \Big( P X_{\eta_j} P - \lambda_{\eta_j} + \lambda_{\eta_j} + P X_{\eta_j} Q \Big) \\
& = -I + (\lambda_{\eta_j} - P X_{\eta_j} P)^{-1} \Big( \lambda_{\eta_j} + P X_{\eta_j} Q \Big) 
\end{align*}
Therefore, recalling $\lambda_{\eta_j}  = \lambda - \eta_j$ and $X_{\eta_j} = X - \eta_j$ we have that
\[
\| (\lambda_{\eta_j}  - P X_{\eta_j} P)^{-1} P X_{\eta_j} \| \leq 1 + \| (\lambda_{\eta_j} - P X_{\eta_j} P)^{-1}  \| \Big( \vert{}\lambda - \eta_j \vert{} + \| P (X - \eta_j) Q \|  \Big).
\]
As discussed previously, by the choice of $\eta_j$ and the uniform spectral gaps assumption, we know that $\vert{}\lambda - \eta_j\vert{}$ is bounded by a constant independent of $j$, $\lambda$, and $\eta_j$. To see the second term is bounded by a constant, recall that $PQ = 0$ so
\[
\| P (X - \eta_j) Q \| = \| P X Q \| = \| [P, X] Q \| \leq \| [P, X ] \|
\]
Since $P$ admits an exponentially localized kernel, it is easily verified that $\| [P, X ] \|$ is bounded by an absolute constant. Therefore, $\| (\lambda_{\eta_j}  - P X_{\eta_j} P)^{-1} P X_{\eta_j} \|$ is bounded and so by choosing
\[
\gamma \leq (2 C \| (\lambda_{\eta_j}  - P X_{\eta_j} P)^{-1} P X_{\eta_j} \|)^{-1}
\]
by the previous logic $(\lambda_{\eta_j} - P X_{\eta_j} P_{\gamma} )^{-1} $ is bounded, proving the first implication.

\subsubsection{Proof that shifted $(\lambda - P_{\gamma} X P_{\gamma} )^{-1}$ is bounded}
\label{sec:exp-loc-res-pt2}
Similar to before, we begin by adding and subtracting $P X_{\eta_j} P_{\gamma}$ in the shifted resolvent:
\begin{align*}
(\lambda_{\eta_j} & - P_{\gamma} X_{\eta_j} P_{\gamma} )^{-1}  \\
& = (\lambda_{\eta_j} - P X_{\eta_j} P_{\gamma} + P X_{\eta_j} P_{\gamma} - P_{\gamma} X_{\eta_j} P_{\gamma} )^{-1} \\
& = (\lambda_{\eta_j} - P X_{\eta_j} P_{\gamma} - (P_{\gamma} - P) X_{\eta_j} P_{\gamma} )^{-1} \\
& = \Big(I - (P_{\gamma} - P) X_{\eta_j}  P_{\gamma} (\lambda_{\eta_j} - P X_{\eta_j} P_{\gamma})^{-1} \Big)^{-1} (\lambda_{\eta_j} - P X_{\eta_j} P_{\gamma})^{-1} 
\end{align*}
Similar to before, since $\| P_{\gamma} - P \| \leq C \gamma$, if we can show that $\| X_{\eta_j}  P_{\gamma} (\lambda_{\eta_j} - P X_{\eta_j} P_{\gamma})^{-1} \|$ is bounded, then we can pick $\gamma$ sufficiently small so that
\[
\| (P_{\gamma} - P) X_{\eta_j}  P_{\gamma} (\lambda_{\eta_j} - P X_{\eta_j} P_{\gamma})^{-1} \| \leq \frac{1}{2}
\]
which implies that $(\lambda_{\eta_j} - P_{\gamma} X_{\eta_j} P_{\gamma} )^{-1}$ is bounded.

To show that $X_{\eta_j}  P_{\gamma} (\lambda_{\eta_j} - P X_{\eta_j} P_{\gamma})^{-1}$ is bounded, let us adopt the shorthand $E := P_{\gamma} - P$. Since $Q_{\gamma} = I - P_{\gamma}$ and $Q = I - P$ we also have that
\[
E = P_{\gamma} - P = Q - Q_{\gamma} \Rightarrow Q = Q_{\gamma} + E.
\]
Next, we calculate
\begin{align*}
X_{\eta_j}  P_{\gamma} (\lambda_{\eta_j} - P X_{\eta_j} P_{\gamma})^{-1}
& = (P + Q) X_{\eta_j}  P_{\gamma} (\lambda_{\eta_j} - P X_{\eta_j} P_{\gamma})^{-1} \\
& = (P + Q_{\gamma} + E) X_{\eta_j}  P_{\gamma} (\lambda_{\eta_j} - P X_{\eta_j} P_{\gamma})^{-1} 
\end{align*}
Moving the term multiplied by $E$ to the left hand side then gives
\begin{align*}
(I - E) X_{\eta_j} &  P_{\gamma} (\lambda_{\eta_j} - P X_{\eta_j} P_{\gamma})^{-1} = (P + Q_{\gamma}) X_{\eta_j}  P_{\gamma} (\lambda_{\eta_j} - P X_{\eta_j} P_{\gamma})^{-1} 
\end{align*}
Since $\| E \| = \| P_{\gamma} - P \| \leq C \gamma$, we can choose $\gamma$ sufficiently small so that $I - E$ is invertible and hence we conclude that
\begin{align*}
\| X_{\eta_j} &  P_{\gamma} (\lambda_{\eta_j} - P X_{\eta_j} P_{\gamma})^{-1} \| \\
& \leq \| (I - E)^{-1} \| \Big( \| P X_{\eta_j}  P_{\gamma} (\lambda_{\eta_j} - P X_{\eta_j} P_{\gamma})^{-1} \| \\
 & \hspace{8em} + \| Q_{\gamma} X_{\eta_j}  P_{\gamma} \| \| (\lambda_{\eta_j} - P X_{\eta_j} P_{\gamma})^{-1}  \| \Big) \\
 & \leq \| (I - E)^{-1} \| \Big( 1 + \vert{}\lambda - \eta_j \vert{} \| (\lambda_{\eta_j} - P X_{\eta_j} P_{\gamma})^{-1} \| \\
 & \hspace{8em} + \| Q_{\gamma} (X_{\eta_j}  P_{\gamma} \| \| (\lambda_{\eta_j} - P X_{\eta_j} P_{\gamma})^{-1}  \| \Big)
\end{align*}
Since $\vert{}\lambda - \eta_j\vert{}$ is bounded by construction and $\| (\lambda_{\eta_j} - P X_{\eta_j} P_{\gamma})^{-1} \|$ is bounded by the proof in Section \ref{sec:exp-loc-res-pt1}, the only term to bound is $\| Q_{\gamma} X_{\eta_j}  P_{\gamma} \|$. Using that $Q_{\gamma} P_{\gamma} = 0$ we have that
\[
\| Q_{\gamma} X_{\eta_j}  P_{\gamma} \| = \| Q_{\gamma} X  P_{\gamma} \| \leq \| Q_{\gamma} \| \| [X, P_{\gamma} ] \|.
\]
Since $P$ admits an exponentially localized kernel, it is easy to verify that $\| [X, P_{\gamma} ] \|$ is bounded. Since additionally,
\[
\| Q_{\gamma} \| = \| I - P_{\gamma} \| \leq 1 + \| P_{\gamma} \|
\]
we conclude that $\| X_{\eta_j}  P_{\gamma} (\lambda_{\eta_j} - P X_{\eta_j} P_{\gamma})^{-1} \|$ is bounded by an absolute constant for all $\gamma$ sufficiently small. Therefore, by the previously discussed reasoning, we conclude that $\| (\lambda_{\eta_j} - P_{\gamma} X_{\eta_j} P_{\gamma} )^{-1} \|$ is bounded by a constant, completing the proof of the proposition.

\subsection{Proof that $P_j$ admits an exponentially localized kernel (Proposition \ref{prop:pj-props}(1))}
\label{sec:exp-loc-kernel}
Let us recall the definition of $P_j$ (Definition \ref{def:band-projectors}) 
\[
P_j = \frac{1}{2 \pi i} \int_{\cC_j} (\lambda - PXP)^{-1} P \dee{\lambda}
\]
By multiplying on the left by $B_{\gamma}$ and on the right by $B_{\gamma}^{-1}$ we therefore have that
\begin{align*}
P_{j,\gamma}
& = \frac{1}{2 \pi i} \int_{\cC_j} B_{\gamma} (\lambda - PXP)^{-1} P B_{\gamma}^{-1} \dee{\lambda} \\
& = \frac{1}{2 \pi i} \int_{\cC_j} B_{\gamma} (\lambda - PXP)^{-1} \Big( B_{\gamma}^{-1} B_{\gamma}\Big) P B_{\gamma}^{-1} \dee{\lambda} \\
& = \frac{1}{2 \pi i} \int_{\cC_j} (\lambda - P_{\gamma} X P_{\gamma} )^{-1} P_{\gamma} \dee{\lambda}
\end{align*}
Therefore, we conclude that
\[
\| P_{j,\gamma} \| \leq \frac{\ell{(\cC)}}{2 \pi} \| P_{\gamma} \| \Big( \sup_{\lambda \in \cC_j} \| (\lambda - P_{\gamma} X P_{\gamma} )^{-1} \| \Big)
\]
In Section \ref{sec:exp-loc-res} we showed that there exists an absolute constant $C$ so that
\[
\sup_{\lambda \in \cC_j} \| (\lambda - P_{\gamma} X P_{\gamma} )^{-1} \| \leq C.
\]
Therefore, we have that there exists a constant $C$ so that
\[
\| P_{j,\gamma} \| \leq C.
\]
we can now use this estimate to show that $P_j$ admits an exponentially localized kernel.
\begin{proposition}
\label{prop:pj-exp-kernel}
Suppose that $P$ is an orthogonal projector which admits an exponentially localized kernel with rate $\gamma_1$ (Definition~\ref{def:exp-loc-kernel}). Suppose further that $P_j$ is an orthogonal projector which satisfies the properties:
\begin{enumerate}
    \item $P_j P = P P_j = P_j$
    \item There exist a constant $C$ such that for all $0 \leq \gamma \leq \gamma_2$ 
    \[
    \| P_{j,\gamma} \| \leq C.
    \]
\end{enumerate}
Then $P_j$ admits an integral kernel $P_j(\cdot, \cdot) : \field{R}^2 \times \field{R}^2 \rightarrow \field{C}$ such that for all $\gamma \leq \frac{1}{2} \min\{ \gamma_1, \gamma_2 \}$
\[
\vert{}P_j(\vec{x}, \vec{x}')\vert{} \leq C' e^{-\gamma \vert{}\vec{x} - \vec{x}'\vert{}} \qquad a.e.
\]
where the constant $C'$ only depends on $\gamma_1$, $\gamma_2$, and $\| P_{j,\gamma}\|$.
\end{proposition}
\begin{proof}
For this proof we will first show that $P_{j}$ admits a measurable integral kernel and then show that $P_{j,\gamma}$ is a bounded operator from $L^1$ to $L^\infty$. Once we can show these properties, the fact that $P_j$ admits an exponentially localized kernel follows by application of the Lebesgue differentiation theorem.

Define $\gamma^* := \min\{ \gamma_1, \gamma_2 \}$, we will first show that $P_{j,\gamma}$ is a bounded operator from $L^2$ to $L^\infty$ for all $0 \leq \gamma < \gamma^*$ (i.e. that $P_{j,\gamma}$ is a Carleman operator). First, let us fix a realization of a function $f \in L^2(\field{R}^2)$. Using the fact that $P_{j,\gamma} = P_{\gamma} P_{j,\gamma}$, we have that for almost all $\vec{x}$:
\begin{align}
    \vert{}(P_{j,\gamma} f)(\vec{x})\vert{} = \vert{}(P_{\gamma} P_{j,\gamma} f)(\vec{x})\vert{}
    & = \left\vert{}\int_{\field{R}^2} P_{\gamma}(\vec{x}, \vec{x}') (P_{j,\gamma} f)(\vec{x}') \dee{\vec{x}'}\right\vert{} \\
    & \leq \int_{\field{R}^2} C e^{-(\gamma_1 - \gamma) \vert{}\vec{x} - \vec{x}'\vert{}} (P_{j,\gamma} f)(\vec{x}') \dee{\vec{x}'} \\
    & \leq C \left( \int_{\field{R}^2} e^{-2(\gamma_1 - \gamma) \vert{}\vec{x} - \vec{x}'\vert{}} \dee{\vec{x}'} \right) \| P_{j,\gamma} f \|  \label{eq:exp_decay}
\end{align}
where in the last line we have applied the Cauchy-Schwartz inequality. Taking the essential supremum over $\vec{x}$ on both sides then gives
\begin{align}
\sup_{\vec{x}}\vert{}(P_{j,\gamma} f)(\vec{x})\vert{}
& \leq C (\gamma_1 - \gamma)^{-1} \| P_{j,\gamma} \|  \| f \|.
\end{align}
Therefore, $P_{j,\gamma}$ is a bounded operator $L^2 \rightarrow L^\infty$. Since $P_{j,\gamma}$ is a bounded operator from $L^2 \rightarrow L^2$ and $L^2 \rightarrow L^\infty$, a standard result in the study of integral operators (see \cite[Corollary A.1.2]{1982SimonSchrodinger}) gives us that $P_{j,\gamma}$ admits an integral kernel which satisfies the following estimate:
\begin{equation}
  \label{eq:exp-kern-l2a}
\sup_{\vec{x}} \left[ \int \vert{}P_{j,\gamma}(\vec{x}, \vec{x}')\vert{}^2 \dee{\vec{x}'} \right]^{1/2} < \infty.
\end{equation}
By the definition of $P_{j,\gamma}$ this implies that for any $\vec{a} \in \field{R}^2$
\begin{equation}
\sup_{\vec{x}} \left[ \int \vert{}e^{\gamma \vert{} \vec{x} - \vec{a}\vert{}} P_{j}(\vec{x}, \vec{x}') e^{-\gamma \vert{} \vec{x}' - \vec{a}\vert{}}\vert{}^2 \dee{\vec{x}'} \right]^{1/2} < \infty.
\end{equation}
For reasons which will shortly be clear, we make the following observation. Since for all $\gamma$
\[
  \| P_{j,\gamma} \| = \| P_{j,\gamma}^\dagger\| = \| (B_{\gamma} P_j B_{\gamma}^{-1})^\dagger \| = \| P_{j,-\gamma} \|
\]
we can repeat the above steps replacing $\gamma$ with $-\gamma$. This implies that
\begin{equation}
\label{eq:exp-kern-l2b}
\begin{split}
\sup_{\vec{x}} & \left[ \int \vert{}e^{-\gamma \vert{} \vec{x} - \vec{a}\vert{}} P_{j}(\vec{x}, \vec{x}') e^{\gamma \vert{} \vec{x}' - \vec{a}\vert{}}\vert{}^2 \dee{\vec{x}'} \right]^{1/2} < \infty \\
& \Longleftrightarrow \sup_{\vec{x}} \left[ \int \vert{}e^{\gamma \vert{} \vec{x}' - \vec{a}\vert{}} P_{j}(\vec{x}', \vec{x}) e^{-\gamma \vert{} \vec{x} - \vec{a}\vert{}}\vert{}^2 \dee{\vec{x}'} \right]^{1/2} < \infty \\
& \Longleftrightarrow \sup_{\vec{x}} \left[ \int \vert{} P_{j,\gamma}(\vec{x}', \vec{x})\vert{}^2 \dee{\vec{x}'} \right]^{1/2} < \infty
\end{split}
\end{equation}
where we have used the fact that $P_j$ is self-adjoint and hence $P_j(\vec{x},\vec{x}') = \overline{P_j(\vec{x}',\vec{x})}$. Note that the difference between Equation~\eqref{eq:exp-kern-l2a} and~\eqref{eq:exp-kern-l2b} is that we have exchanged the arguments in the kernel for $P_{j,\gamma}$.

Having established that $P_{j,\gamma}$ has an integral kernel satisfying Equation~\eqref{eq:exp-kern-l2b}, we will now use the existence of this kernel to show that $P_{j,\gamma}$ is a bounded operator from $L^1 \rightarrow L^\infty$. Once we show this, the fact that the kernel for $P_j$ is exponentially localized will follow as a consequence of the Lebesgue differentiation theorem. 

Repeating the calculations that lead to \eqref{eq:exp_decay} we have that
\begin{align*}
    \vert{}(P_{j,\gamma} f)(\vec{x})\vert{} 
    & \leq C \int_{\field{R}^2} e^{-(\gamma_1 - \gamma) \vert{}\vec{x} - \vec{x}'\vert{}} \vert{}(P_{j,\gamma} f)(\vec{x}')\vert{} \dee{\vec{x}'} \\
    & \leq C \int_{\field{R}^2} e^{-(\gamma_1 - \gamma) \vert{}\vec{x} - \vec{x}'\vert{}} \int_{\field{R}^2} \vert{}P_{j,\gamma}(\vec{x}', \vec{y})\vert{} \vert{}f(\vec{y})\vert{} \dee{\vec{y}} \dee{\vec{x}'} \\
    & = C \int_{\field{R}^2} \vert{}f(\vec{y})\vert{}  \int_{\field{R}^2} e^{-(\gamma_1 - \gamma) \vert{}\vec{x} - \vec{x}'\vert{}} \vert{}P_{j,\gamma}(\vec{x}', \vec{y})\vert{}  \dee{\vec{x}'} \dee{\vec{y}} \\
    & \leq C \| f \|_{L^1} \left(\sup_{\vec{y}} \int_{\field{R}^2} e^{-(\gamma_1 - \gamma) \vert{}\vec{x} - \vec{x}'\vert{}} \vert{}P_{j,\gamma}(\vec{x}', \vec{y})\vert{}  \dee{\vec{x}'} \right)
\end{align*}
Taking the essential supremum over $\vec{x}$ on both sides then gives
\begin{align*}
    \sup_{\vec{x}} \vert{}(P_{j,\gamma} f)(\vec{x})\vert{} 
    & \leq C \| f \|_{L^1} \left(\sup_{\vec{x},\vec{y}} \int_{\field{R}^2} e^{-(\gamma_1 - \gamma) \vert{}\vec{x} - \vec{x}'\vert{}} \vert{}P_{j,\gamma}(\vec{x}', \vec{y})\vert{}  \dee{\vec{x}'} \right) \\[1.5ex]
    & \leq C \| f \|_{L^1} \left(\sup_{\vec{x}} \int_{\field{R}^2} e^{-2 (\gamma_1 - \gamma) \vert{}\vec{x} - \vec{x}'\vert{}} \dee{\vec{x}'} \right) \\
    & \hspace{8em} \times \left( \sup_{\vec{y}} \int_{\field{R}^2} \vert{}P_{j,\gamma}(\vec{x}', \vec{y})\vert{}^2 \dee{\vec{x}'} \right)
\end{align*}
where in the last line we have used the Cauchy-Schwarz inequality. Finally, using Equation~\eqref{eq:exp-kern-l2b} we conclude that
\[
  \sup_{\vec{x}} \vert{}(P_{j,\gamma} f)(\vec{x})\vert{}  \leq C (\gamma_1 - \gamma)^{-1} \| f \|_{L^1} 
\]
hence $P_{j,\gamma}$ is a bounded operator from $L^1 \rightarrow L^\infty$. 

To show that $P_{j}$ admits an exponentially localized kernel, let us fix some arbitrary points $\vec{a}, \vec{b} \in \field{R}^2$ and define a function $g_{\delta}$ as follows
\[
g_{\delta}(\vec{x})
=
\begin{cases}
1 & \vec{x} \in \cB_{\delta}(\vec{0}) \\
0 & \text{otherwise,}
\end{cases}
\]
where $\cB_{\delta}(\vec{0})$ is the ball of radius $\delta$ centered at $\vec{0}$. Observe that $\| g_{\delta} \|_{L^1} = \vert{}\cB_{\delta}\vert{}$ where $\vert{}\cB_{\delta}\vert{}$ is the volume of the ball of radius $\delta$. With this notation, using our previous bound for $P_{j,\gamma}$ we have that for all $\vec{a} \in \field{R}^2$ and almost all $\vec{b} \in \field{R}^2$:
\begin{align*}
    &\left\vert{} \int P_{j,\gamma}(\vec{b}, \vec{x}') g_{\delta}(\vec{x}' - \vec{a}) \dee{\vec{x}'} \right\vert{} \leq  C' (\gamma_1 - \gamma)^{-1} \| g_{\delta} \|_{L^1} \\
    & \Longrightarrow \frac{1}{\vert{}\cB_{\delta}\vert{}} \left\vert{} \int_{\cB_{\delta}(\vec{a})} P_{j,\gamma}(\vec{b}, \vec{x}') \dee{\vec{x}'} \right\vert{} \leq  C' (\gamma_1 - \gamma)^{-1} \\
    & \Longrightarrow \frac{1}{\vert{}\cB_{\delta}\vert{}} \left\vert{} \int_{\cB_{\delta}(\vec{a})} e^{\gamma \vert{} \vec{b} - \vec{a}\vert{}} P_{j}(\vec{b}, \vec{x}') e^{-\gamma \vert{} \vec{x}' - \vec{a}\vert{}} \dee{\vec{x}'} \right\vert{} \leq  C' (\gamma_1 - \gamma)^{-1}
\end{align*}
where in the last line, we have substituted in the definition for the kernel for $P_{j,\gamma} = B_{\gamma,\vec{a}} P_j B_{\gamma,\vec{a}}^{-1}$. Since for each fixed $\vec{b}$, the function $\vec{x}' \mapsto P_{j,\gamma}(\vec{b}, \vec{x}')$ is in $L^2(\field{R}^2)$ (Equation~\eqref{eq:exp-kern-l2b}), in particular this kernel is also in $L^1_{\loc}(\field{R}^2)$. Therefore, by the Lebesgue differentiation theorem \cite[Chapter 3, Theorem 1.2]{2009SteinShakarchi}, for almost every $\vec{a} \in \field{R}^2$:
\[
\lim_{\delta \rightarrow 0} \frac{1}{\vert{}\cB_{\delta}\vert{}} \left\vert{} \int_{\cB_{\delta}(\vec{a})}  e^{\gamma \vert{} \vec{b} - \vec{a}\vert{}} P_{j}(\vec{b}, \vec{x}') e^{-\gamma \vert{} \vec{x}' - \vec{a}\vert{}} \dee{\vec{x}'} \right\vert{} =  e^{\gamma \vert{} \vec{b} - \vec{a}\vert{}} \vert{}P_{j}(\vec{b}, \vec{a}) \vert{}.
\]
Hence, for almost every $\vec{a}, \vec{b} \in \field{R}^2$:
\[
\vert{}P_{j}(\vec{b}, \vec{a})\vert{} \leq  C' (\gamma_1 - \gamma)^{-1} e^{-\gamma \vert{} \vec{b} - \vec{a}\vert{}}
\]
which completes the proof.
\end{proof}

\subsection{Proof that $P_j$ is localized in $X$ (Proposition \ref{prop:pj-props}(2))}
\label{sec:pj-x-loc}
For this section, let us fix some $\eta_j \in \sigma_j$, we will prove that $\|(X - \eta_j) P_{j,\gamma}\|$ is bounded. The fact that $\|P_{j,\gamma} (X - \eta_j)\|$ is bounded follows by essentially the same steps. Recalling we define $\lambda_{\eta_j} := \lambda - \eta_j$ and $X_{\eta_j} := X - \eta_j$ and using Lemma \ref{lem:shifting} we have
\begin{align*}
  (X - \eta_j) P_{j,\gamma} & = \frac{1}{2 \pi i} \int_{\cC_j} (X - \eta_j) (\lambda - P_{\gamma} X P_{\gamma})^{-1}P_{\gamma} \dee{\lambda} \\
                          & = \frac{1}{2 \pi i} \int_{\cC_j} X_{\eta_j} (\lambda_{\eta_j} - P_{\gamma} X_{\eta_j} P_{\gamma})^{-1} P_{\gamma} \dee{\lambda} \\
                          & = \frac{1}{2 \pi i} \int_{\cC_j} X_{\eta_j} P_{\gamma} (\lambda_{\eta_j} - P_{\gamma} X_{\eta_j} P_{\gamma})^{-1} \dee{\lambda} \\
                          & = \frac{1}{2 \pi i} \int_{\cC_j} (P_{\gamma} + Q_{\gamma}) X_{\eta_j} P_{\gamma} (\lambda_{\eta_j} - P_{\gamma} X_{\eta_j} P_{\gamma})^{-1} \dee{\lambda}.
\end{align*}
where we have used that $P_{\gamma}$ commutes with $(\lambda_{\eta_j} - P_{\gamma} X_{\eta_j} P_{\gamma})^{-1}$. Therefore,
\begin{align*}
  \|(X - \eta_j) P_{j,\gamma}\| & \leq \frac{\ell(\cC_j)}{2 \pi} \bigg( \| P_{\gamma} X_{\eta_j} P_{\gamma} (\lambda_{\eta_j} - P_{\gamma} X_{\eta_j} P_{\gamma})^{-1}\| \\
  & \hspace{4em} + \| Q_{\gamma} X_{\eta_j} P_{\gamma}\|  \| (\lambda_{\eta_j} - P_{\gamma} X_{\eta_j} P_{\gamma})^{-1}\| \bigg).
\end{align*}
Note that
\[
  \| P_{\gamma} X_{\eta_j} P_{\gamma} (\lambda_{\eta_j} - P_{\gamma} X_{\eta_j} P_{\gamma})^{-1}\| \leq 1 + \vert{}\lambda_{\eta_j}\vert{} \| (\lambda_{\eta_j} - P_{\gamma} X_{\eta_j} P_{\gamma})^{-1}\|.
\]
Therefore, since $\| (\lambda_{\eta_j} - P_{\gamma} X_{\eta_j} P_{\gamma})^{-1}\|$ and $\| Q_{\gamma} X_{\eta_j} P_{\gamma}\|$ are both bounded, we can conclude that $\| (X - \eta_j) P_{j,\gamma}\|$ is bounded as we wanted to show.

\subsection{Extension to other position operators}
\label{sec:extension-to-other}
In Sections \ref{sec:exp-loc-res}, \ref{sec:exp-loc-kernel}, and \ref{sec:pj-x-loc} we show that if $PXP$ has uniform spectral gaps then if $\{ P_j \}_{j \in \cJ}$ are the band projectors for $PXP$ then each $P_j$ admits an exponentially localized kernel and each $P_j$ is localized along a line $X = \eta_j$. In the main text we used these properties of $P_j$ to show that exponentially localized Wannier functions can be constructed by diagonalizing $P_j Y P_j$ for each $j \in \cJ$. In this section we show that this argument does not rely on $X$ being precisely the position operator defined in \eqref{eq:X}. Instead, we will show that for operators $\widehat{X}$ which are, in a certain sense, close to the operator $X$, if $P \widehat{X} P$ has uniform spectral gaps, then if we define projections $P_j$ with respect to the spectral bands of $P \widehat{X} P$, the rest of our analysis goes through without modification. Specifically, we can prove that the projections $P_j$ admit exponentially localized kernels, localize along lines $X = \eta_j$, and can be used to form generalized Wannier functions by diagonalizing each $P_j Y P_j$ for $j \in \cJ$. We anticipate that this generalization may be important for two reasons.
\begin{itemize}
\item Depending on the application, it may be more natural to measure position by operators other than the standard ones.
\item It may be that $PXP$ does not have uniform spectral gaps, but $P \widehat{X} P$, where $\widehat{X}$ is an alternative position operator, does. In this case, the fact that our proofs can be generalized is essential to prove existence of exponentially localized generalized Wannier functions.
\end{itemize}
We consider various cases where alternative position operators are useful in \cite{Stubbs2020}. The problem of constructing an alternative position operator $\widehat{X}$ which ensures that $P \widehat{X} P$ has gaps has been considered in \cite{Lu2021}.



We will prove the following Lemma, which gives sufficient conditions so that the band projectors for $P \widehat{X} P$, where $\widehat{X}$ is an alternative position operator, satisfy the results of Proposition \ref{prop:pj-props}.
\begin{lemma}
\label{lem:extension}
Let $\widehat{X}$ be a symmetric operator satisfying $\| \widehat{X} - X \| \leq C$. Then $\widehat{X}$ is self-adjoint $\mathcal{D}(X) \rightarrow L^2(\field{R}^2)$ and $P \widehat{X} P$ is self-adjoint $\{ \mathcal{D}(X) \cap \range(P) \} \cup \range(P)^\perp \rightarrow L^2(\field{R}^2)$. Suppose further that
\begin{enumerate}[itemsep=.75ex]
\item $\| \widehat{X}_{\gamma} - \widehat{X} \| \leq C' \gamma$
\item $P \widehat{X} P$ has uniform spectral gaps.
\end{enumerate}
Then, if $\{ P_j \}_{j \in \cJ}$ are the band projectors of $P \widehat{X} P$, then $\{ P_j \}_{j \in \cJ}$ satisfies the results of Proposition \ref{prop:pj-props}.
\end{lemma}
Following the proof in the main text, an immediate corollary of Lemma \ref{lem:extension} is the following.
\begin{corollary}
If it is possible to construct a symmetric operator $\widehat{X}$ satisfying the assumptions of Lemma \ref{lem:extension}, then $\range{(P)}$ admits a basis of exponentially localized generalized Wannier functions.
\end{corollary}
\begin{proof}[Proof of Lemma \ref{lem:extension}]
The first claim of the lemma follows easily from the Kato-Rellich theorem since $\widehat{X} - X$ is a bounded perturbation of $X$. Assuming (2), we can define the corresponding band projectors $P_j$ via the Riesz formula
\[
    P_j := \frac{1}{2 \pi i} \int_{\cC_j} (\lambda - P \widehat{X} P)^{-1} P \dee{\lambda}.
\]
Therefore, we can write $P_{j,\gamma}$ as follows
\[
    P_{j,\gamma} := \frac{1}{2 \pi i} \int_{\cC_j} (\lambda - P_{\gamma} \widehat{X}_{\gamma} P_{\gamma})^{-1} P_{\gamma} \dee{\lambda}.
\]
Our goal is now to establish Proposition \ref{prop:pj-res} for these band projectors. Similar to before, to get bounds which are independent of the choice of $j \in \cJ$, appealing to the shifting lemma (Lemma \ref{lem:shifting}), to prove an analog of Proposition \ref{prop:pj-res}, it suffices to show that for all $j \in \cJ$
\begin{equation}
\label{eq:other-pos-res}
\sup_{\lambda \in \cC_j} \sup_{\eta_j \in \cC_j} \| (\lambda_{\eta_j} - P_{\gamma} \widehat{X}_{\eta_j,\gamma} P_{\gamma})^{-1} \| \leq C'
\end{equation}
where $\lambda_{\eta_j} = \lambda - \eta_j$ and $\widehat{X}_{\eta_j,\gamma} = \widehat{X}_{\gamma} - \eta_j$. 

Equation \eqref{eq:other-pos-res} can be proved by the following series of implications:
\begin{equation}
  \begin{split}
  P \widehat{X} &P \text{ has USG} \Longrightarrow \| (\lambda_{\eta_j} - P \widehat{X}_{\eta_j} P )^{-1} \| < \infty \\
  & \Longrightarrow \| (\lambda_{\eta_j} - P \widehat{X}_{\eta_j} P_{\gamma} )^{-1} \| < \infty \\
  &\Longrightarrow \| (\lambda_{\eta_j} - P_{\gamma} \widehat{X}_{\eta_j} P_{\gamma} )^{-1} \| < \infty \\ & \Longrightarrow \| (\lambda_{\eta_j} - P_{\gamma} \widehat{X}_{\eta_j,\gamma} P_{\gamma} )^{-1} \| < \infty.
\end{split}
\end{equation}
Using $\| \widehat{X} - X \| \leq C$, one can check that the proof in Section \ref{sec:exp-loc-res} proves all of the implications except for the last one.

For the last of these implications, we observe that $\widehat{X}_{\eta_j, \gamma} - \widehat{X}_{\eta_j} = \widehat{X}_{\gamma} - \widehat{X}$ and recall that by assumption $\| \widehat{X}_{\gamma} - \widehat{X} \| \leq C' \gamma$. Now observe that
\begin{align*}
(\lambda_{\eta_j} & - P_{\gamma} \widehat{X}_{\eta_j,\gamma} P_{\gamma} )^{-1} \\
& = (\lambda_{\eta_j} - P_{\gamma} \widehat{X}_{\eta_j} P_{\gamma} )^{-1} \Big( I - P_{\gamma} (\widehat{X}_{\gamma} - \widehat{X}) P_{\gamma} (\lambda_{\eta_j} - P_{\gamma} \widehat{X}_{\eta_j,\gamma} P_{\gamma} )^{-1} \Big)^{-1}.
\end{align*}
Since $(\lambda_{\eta_j} - P_{\gamma} \widehat{X}_{\eta_j} P_{\gamma} )^{-1} $ is bounded, we can choose $\gamma$ sufficiently small so that
\[
\| P_{\gamma} (\widehat{X}_{\gamma} - \widehat{X}) P_{\gamma} (\lambda_{\eta_j} - P_{\gamma} \widehat{X}_{\eta_j,\gamma} P_{\gamma} )^{-1}  \| \leq \frac{1}{2}.
\]
which implies that $(\lambda_{\eta_j} - P_{\gamma} \widehat{X}_{\eta_j,\gamma} P_{\gamma} )^{-1}$ completing the proof of Proposition \ref{prop:pj-res} for the band projectors of $P \widehat{X} P$. Hence following the arguments in Sections \ref{sec:exp-loc-kernel} and \ref{sec:pj-x-loc}, we can conclude the band projectors for $P \widehat{X} P$ satisfy the results of Proposition \ref{prop:pj-props} proving the lemma.
\end{proof}

\section{Proof of the Shifting Lemma (Lemma \ref{lem:shifting})}
\label{sec:shifting-proof}

The basic steps to prove Lemma \ref{lem:shifting} are the following: 
\begin{align*}
  (\lambda_{\eta} - P_{\gamma} X_{\eta} P_{\gamma})^{-1} & = (\lambda - \eta - P_{\gamma} (X - \eta) P_{\gamma})^{-1} \\
                                                         & = (\lambda - \eta - P_{\gamma} X P_{\gamma} + \eta P_{\gamma})^{-1} \\
                                                         & = (\lambda - P_{\gamma} X P_{\gamma} - \eta Q_{\gamma})^{-1}. \numberthis{} \label{eq:shifting-eq1}
\end{align*}
Since $P_{\gamma} + Q_{\gamma} = I$, because of this calculation we know that
\begin{align*}
  \|(\lambda_{\eta} - P_{\gamma} X_{\eta} P_{\gamma})^{-1}\| & = \|(\lambda - P_{\gamma} X P_{\gamma} - \eta Q_{\gamma})^{-1}\| \\
                                                             & \leq \|(\lambda - P_{\gamma} X P_{\gamma} - \eta Q_{\gamma})^{-1} P_{\gamma}\| + \|(\lambda - P_{\gamma} X P_{\gamma} - \eta Q_{\gamma})^{-1} Q_{\gamma}\|.
\end{align*}

Since $P_{\gamma} Q_{\gamma} = Q_{\gamma} P_{\gamma} = 0$, we should expect that shifting by $\eta Q_{\gamma}$ should not change what happens on $\range{(P_{\gamma})}$. Similarly, the action of $P_{\gamma} X P_{\gamma}$ should not change what happens on $\range{(Q_{\gamma})}$. This observation leads us to expect that:
\begin{align}
  &  (\lambda - P_{\gamma} X P_{\gamma} - \eta Q_{\gamma})^{-1} P_{\gamma} = (\lambda - P_{\gamma} X P_{\gamma})^{-1} P_{\gamma} \label{eq:shifting-eq2} \\
  & (\lambda - P_{\gamma} X P_{\gamma} - \eta Q_{\gamma})^{-1} Q_{\gamma} = (\lambda - \eta Q_{\gamma})^{-1} Q_{\gamma}. \label{eq:shifting-eq3}
\end{align}
By similar reasoning:
\begin{equation}
  \label{eq:shifting-eq4}
  (\lambda - \eta Q_{\gamma})^{-1} Q_{\gamma} = (\lambda - \eta + \eta P_{\gamma})^{-1} Q_{\gamma} = (\lambda - \eta)^{-1} Q_{\gamma}.
\end{equation}
Assuming Equations \eqref{eq:shifting-eq2}, \eqref{eq:shifting-eq3}, \eqref{eq:shifting-eq4} are true, we conclude that:
\begin{equation}
  \label{eq:shifting-eq5}
  \|(\lambda_{\eta} - P_{\gamma} X_{\eta} P_{\gamma})^{-1}\| \leq \|(\lambda - P_{\gamma} X P_{\gamma})^{-1}\| \|P_{\gamma}\| + \vert{}\lambda - \eta\vert{}^{-1} \|Q_{\gamma}\|.
\end{equation}
Since $P$ admits an exponentially localized kernel (Definition \ref{def:exp-loc-kernel}) we know that $\|P_{\gamma}\|$ and $\|Q_{\gamma}\|$ are bounded. Because of the uniform spectral gaps assumption on $PXP$, since we have chosen $\eta \in \sigma_j$ and $\lambda \in \cC_j$ we also know that $\vert{} \lambda - \eta \vert{}^{-1}$ is bounded by a constant independent of $j$ and $\eta$. Therefore, Equation \eqref{eq:shifting-eq5} shows that
\[
  \|(\lambda - P_{\gamma} X P_{\gamma})^{-1}\| < \infty \Longrightarrow \|(\lambda_{\eta} - P_{\gamma} X_{\eta} P_{\gamma})^{-1}\| < \infty.
\]
We can prove the reverse implication by instead starting with the calculation
\begin{align*}
  (\lambda - P_{\gamma} X P_{\gamma})^{-1} & = (\lambda - \eta + \eta - P_{\gamma} (X - \eta + \eta) P_{\gamma})^{-1} \\
                                           & = (\lambda_{\eta} - P_{\gamma} X_{\eta} P_{\gamma} + \eta Q_{\gamma})^{-1},
\end{align*}
and proceeding along similar steps.

What remains to finish the proof of Lemma \ref{lem:shifting} is to prove that Equations \eqref{eq:shifting-eq2}, \eqref{eq:shifting-eq3}, \eqref{eq:shifting-eq4} are all true. For this, we have the following technical lemma:
\begin{lemma}
  \label{lem:shifting-technical}
  Let $\tilde{P}, \tilde{Q}$ be any pair of bounded operators such that $\tilde{P} \tilde{Q} = \tilde{Q} \tilde{P} = 0$. Next, let $A,B$ be possibly unbounded operators densely defined on a common domain $\cD$. Suppose further that $\|[\tilde{P},A]\|$ both $\|[\tilde{Q},B]\|$ are bounded.

  If $\tilde{\lambda} \in \field{C}$ is any scalar such that $\|(\tilde{\lambda} + \tilde{P}A\tilde{P})^{-1}\|$ is bounded then $\|(\tilde{\lambda} + \tilde{P}A\tilde{P} + \tilde{Q}B\tilde{Q})^{-1} \tilde{P}\|$ is also bounded and
  \[
    (\tilde{\lambda} + \tilde{P}A\tilde{P})^{-1} \tilde{P} = (\tilde{\lambda} + \tilde{P}A\tilde{P} + \tilde{Q}B\tilde{Q})^{-1} \tilde{P}.
  \]
\end{lemma}
Note that applying Lemma \ref{lem:shifting-technical} three times proves that Equations \eqref{eq:shifting-eq2}, \eqref{eq:shifting-eq3}, \eqref{eq:shifting-eq4} are all true.

The assumption that $\|[\tilde{P}, A]\|$ and $\|[\tilde{Q}, B]\|$ are bounded is purely a technical assumption which ensures that $\tilde{P} A \tilde{P}$ and $\tilde{Q} B \tilde{Q}$ are well-defined operators on $\cD$. To see why, observe that
\[
  \tilde{P} A \tilde{P} = \tilde{P} [A, \tilde{P}] + \tilde{P} \tilde{P} A \text{ and } \tilde{Q} B \tilde{Q} = \tilde{Q} [B, \tilde{Q}] + \tilde{Q} \tilde{Q} B.
\]
For our purposes, the only unbounded operator we will need to be careful with is the operator $X$. Since $P$ satisfies admits an exponentially localized kernel (Definition \ref{def:exp-loc-kernel}) we know that $\|[P_{\gamma}, X]\| = \|[Q_{\gamma}, X]\| < \infty$, therefore we may apply Lemma \ref{lem:shifting-technical} without worry.
\begin{proof}[Proof of Lemma \ref{lem:shifting-technical}]
First, note that $\tilde{\lambda} + \tilde{P} A \tilde{P} + \tilde{Q} B \tilde{Q}$ is injective on $\range(\tilde{P})$ since for arbitrary non-zero $v \in \range(\tilde{P}) \cap \mathcal{D}$,
\begin{equation*}
\begin{split}
    &\left\| \left( \tilde{\lambda} + \tilde{P} A \tilde{P} + \tilde{Q} B \tilde{Q} \right) v \right\| = \left\| \left( \tilde{\lambda} + \tilde{P} A \tilde{P} + \tilde{Q} B \tilde{Q} \right) \tilde{P} v \right\|     \\
    &= \left\| \left( \tilde{\lambda} + \tilde{P} A \tilde{P} \right) v \right\| \geq \| (\tilde{\lambda} + PAP)^{-1} \|^{-1} \| v \|.
\end{split}
\end{equation*}
  Now observe that since $\tilde{P} \tilde{Q} = \tilde{Q} \tilde{P} = 0$
  \[
    [ (\tilde{\lambda} + \tilde{P}A\tilde{P} + \tilde{Q}B\tilde{Q}), (\tilde{\lambda} + \tilde{P}A\tilde{P})] = 0.
  \]
  Since $(\tilde{\lambda} + \tilde{P}A\tilde{P})^{-1}$ is well defined, this implies that
  \[
    [ (\tilde{\lambda} + \tilde{P}A\tilde{P} + \tilde{Q}B\tilde{Q}), (\tilde{\lambda} + \tilde{P}A\tilde{P})^{-1} ] = 0.
  \]
  Since $\tilde{Q}\tilde{P} = 0$ we also have that
  \begin{align*}
    & (\tilde{\lambda} + \tilde{P}A\tilde{P} + \tilde{Q}B\tilde{Q}) \tilde{P} = (\tilde{\lambda} + \tilde{P}A\tilde{P}) \tilde{P} \\[.5ex]
    & ~~~ \Longleftrightarrow (\tilde{\lambda} + \tilde{P}A\tilde{P})^{-1} (\tilde{\lambda} + \tilde{P}A\tilde{P} + \tilde{Q}B\tilde{Q}) \tilde{P} = \tilde{P} \\[.5ex]
    & ~~~ \Longleftrightarrow (\tilde{\lambda} + \tilde{P}A\tilde{P} + \tilde{Q}B\tilde{Q}) (\tilde{\lambda} + \tilde{P}A\tilde{P})^{-1} \tilde{P} =  \tilde{P}.
  \end{align*}
  The final equality implies that $\range{(\tilde{P})} \subseteq \range{(\tilde{\lambda} + \tilde{P}A\tilde{P} + \tilde{Q}B\tilde{Q})}$. Since $(\tilde{\lambda} - \tilde{P}A\tilde{P})^{-1}$ is bounded we can conclude that $(\tilde{\lambda} + \tilde{P}A\tilde{P} + \tilde{Q}B\tilde{Q})$ is invertible on $\range{(\tilde{P})}$ and moreover we have
  \[
    (\tilde{\lambda} + \tilde{P}A\tilde{P})^{-1} \tilde{P} =  (\tilde{\lambda} + \tilde{P}A\tilde{P} + \tilde{Q}B\tilde{Q})^{-1} \tilde{P}.
  \]
\end{proof}

\section{Proof that the generalized Wannier functions of Theorem \ref{th:main_theorem} are Wannier functions when $A$ and $V$ are periodic (Theorem \ref{th:periodic_case})} \label{sec:gWFs_closed}

We start with some notation. Let $\Lambda$ denote a two-dimensional lattice generated by non-parallel lattice vectors $\vec{v}_1, \vec{v}_2 \in \field{R}^2$, i.e.
\begin{equation} \label{eq:lattice_lambda}
    \Lambda = \left\{ m_1 \vec{v}_1 + m_2 \vec{v}_2 : ( m_1, m_2 ) \in \field{Z}^2 \right\}.
\end{equation}
We assume that the functions $A$ and $V$ are periodic with respect to $\Lambda$ in the sense that for all $\vec{x} \in \field{R}^2$
\begin{equation}
    A(\vec{x} + \vec{v}) = A(\vec{x}), \quad V(\vec{x} + \vec{v}) = V(\vec{x}) \text{ for all $\vec{v} \in \Lambda$.}
\end{equation}
Define $L_j := \vert{}\vec{v}_j\vert{}$, $j = 1,2$. We work with co-ordinates defined with respect to the lattice vectors $\vec{v}_1$ and $\vec{v}_2$ so that for $\vec{x} \in \field{R}^2$,
\begin{equation} \label{eq:x_coords}
    \vec{x} = x_1 \vec{v}_1 + x_2 \vec{v}_2
\end{equation}
for $\vec{x} = (x_1, x_2) \in \field{R}^2$. We require the following assumption on the position operators $X$ and $Y$. 
\begin{assumption} \label{as:XY_periodic}
We assume the position operators $X$ and $Y$ are defined with respect to the co-ordinates \eqref{eq:x_coords} by
\begin{equation} \label{eq:XY_periodic}
    X f(\vec{x}) = L_1 x_1 f(\vec{x}), \quad Y f(\vec{x}) = L_2 x_2 f(\vec{x}).
\end{equation}
\end{assumption}
Let $T_{\vec{v}_1}$ and $T_{\vec{v}_2}$ denote translation operators associated to the Bravais lattice vectors $\vec{v}_1, \vec{v}_2$, so that
\begin{equation}
    T_{\vec{v}_j} f(\vec{r}) = f(\vec{r} - \vec{v}_j), \quad j \in \{1,2\}.
\end{equation}
First, note that by the Riesz projection formula, $[T_{\vec{v}_j},P] = 0$, $j = 1,2$. Since $[T_{\vec{v}_1},X] = - L_1 T_{\vec{v}_1}$ and $[T_{\vec{v}_2},X] = 0$, we have that
\begin{equation}
    [T_{\vec{v}_1},PXP] = - L_1 P T_{\vec{v}_1}, \quad [T_{\vec{v}_2},PXP] = 0.
\end{equation}
The first of these identities implies that the spectrum of $PXP$ restricted to $\range(P)$ is the union of shifted copies of the spectrum of $PXP$ restricted to the interval $[0,L_1)$, i.e. 
\begin{equation} \label{eq:spectrum_union}
    \sigma( PXP\rvert_{\range(P)} ) = \bigcup_{j \in \field{Z}} \{ \sigma_{[0,L_1)} + L_1 j \},
\end{equation}
where $\sigma_{[0,L_1)} := \sigma( PXP\rvert_{\range(P)} ) \cap [0,L_1)$. The $L^2$ spectrum of $PXP$ is then $\{ 0 \} \cup \sigma( PXP\rvert_{\range(P)} )$. 

Using the uniform spectral gaps assumption, we can assume that, perhaps after a constant shift of $X$, the interval $(-\epsilon,\epsilon) \notin \sigma( PXP \rvert_{\range(P)} )$. Because of the symmetry \eqref{eq:spectrum_union} of the spectrum, we have that $(- \epsilon,\epsilon ) + \mathbb{Z} L_1 \notin \sigma( PXP\rvert_{\range(P)} )$, and hence the components $\{ \sigma_{[0,L_1)} + L_1 j \}$ of $\sigma( PXP \rvert_{\range(P)} )$ are separated by gaps. We now assume the following.
\begin{assumption} \label{as:Pj_periodic}
    We define the band projectors $P_j$ as in Definition \ref{def:band-projectors} with $\sigma_j := \{ \sigma_{[0,L_1)} + L_1 j \}$ for $j \in \mathbb{Z}$.
\end{assumption}
Note that other choices of $\sigma_j$ which are compatible with the uniform spectral gaps assumption are possible if $\sigma_{[0,L_1)}$ has separated components, but for simplicity we do not consider this here.  The key consequence of Assumption \ref{as:Pj_periodic} is that
\begin{equation} \label{eq:Pj_projectors_periodic}
    T_{\vec{v}_1} P_j = P_{j + 1} T_{\vec{v}_1}.
\end{equation}


Using the Riesz formula again, we have that $[T_{\vec{v}_2},P_j] = 0$. Using the facts that $[T_{\vec{v}_1},Y] = 0$ and $[T_{\vec{v}_2},Y] = - L_2 T_{\vec{v}_2}$ we have that
\begin{equation}
    \sigma (P_j Y P_j \rvert_{\range( P_j )}) = \bigcup_{m \in \field{Z}} \{ \sigma_{[0,L_2)} + L_2 m \}
\end{equation}
where $\sigma_{[0,L_2)} := \sigma (P_j Y P_j \rvert_{\range( P_j )}) \cap [0,L_2)$. We can therefore index the eigenfunctions of $P_j Y P_j$ by $m \in \field{Z}$ for every $j \in \field{Z}$.

Now let $\{ \psi_{j,m} \}_{(j,m) \in \field{Z} \times \field{Z}}$ denote the set of eigenfunctions of the operators $P_j Y P_j$. We claim that this set is closed under $T_{\vec{v}_1}$ and $T_{\vec{v}_2}$. First, since $T_{\vec{v}_2}$ commutes with $P_j$ it is clear that if $\psi_{j,m}$ is an eigenfunction of $P_j Y P_j$ with eigenvalue $\mu_{j,m}$, then
\begin{equation} \label{eq:mu_calc}
    P_j Y P_j T_{\vec{v}_2} \psi_{j,m} = T_{\vec{v}_2} P_j (Y + L_2) P_j \psi_{j,m} = ( \mu_{j,m} + L_2 ) T_{\vec{v}_2} \psi_{j,m}
\end{equation}
and hence $T_{\vec{v}_2} \psi_{j,m}$ is another eigenfunction of $P_j Y P_j$ with eigenvalue $\mu_{j,m} + L_2$, so the set is closed under $T_{\vec{v}_2}$. To see the set is closed under $T_{\vec{v}_1}$, note that applying $T_{\vec{v}_1}$ to both sides of the eigenequation $P_j ( Y - \mu ) P_j \psi_{j,m} = 0$ yields $P_{j'} ( Y - \mu ) P_{j'} T_{\vec{v}_1} \psi_{j,m} = 0$ for some $j'$ and hence $T_{\vec{v}_1} \psi_{j,m} = \psi_{j',m}$ for some $j' \in \field{Z}$.

We now consider the centers $\{ (a_{j,m},b_{j,m}) \}_{(j,m) \in \field{Z} \times \field{Z}}$. Working in the co-ordinates \eqref{eq:x_coords}, closure under translation in the lattice $\Lambda$ is equivalent to closure of the set of centers under component-wise integer addition. Now recall that for each $j$, in the statement of Theorem \ref{th:main_theorem}, the $a_{j,m}$ can be chosen as any point in $\sigma_j$, where $\sigma_j$ is the spectrum associated to the spectral projection $P_j$, for all $m$. So, since $\sigma_{j+1} = \sigma_j + L_1$ as we have already noted, in the co-ordinates \eqref{eq:x_coords} we can choose the $a_{j,m}$ such that $a_{j+1,m} = a_{j,m} + 1$, and hence the set of centers is closed under integer addition in the first component. To see the same is true of the second component, note that \eqref{eq:mu_calc} shows that the eigenvalues of $P_j Y P_j$ are closed under addition of $L_2$ for each $j$. Since the centers $b_{j,m}$ are defined as the associated eigenvalues of $P_j Y P_j$ of the $\psi_{j,m}$ we see that the set of centers is also closed under integer addition in their second component (again working in the co-ordinates \eqref{eq:x_coords}).

\section{Construction of an analytic and periodic Bloch frame when the uniform spectral gap assumption holds in a crystalline insulator}
\label{sec:chern-app}

In this section we connect the uniform spectral gap assumption (Assumption \ref{def:usg}) to the existing well-developed theory of exponentially-localized Wannier functions in crystalline insulators. Specifically, we prove that if the uniform spectral gap assumption (Assumption \ref{def:usg}) holds for a two-dimensional \emph{crystalline} insulator, then it is possible to construct an analytic and periodic Bloch frame for the Fermi projection. It is well established that existence of such a frame is equivalent to triviality of the Chern number, a topological invariant associated to the Fermi projection, and to the existence of exponentially-localized composite Wannier functions defined via integration with respect to quasi-momentum over the Brillouin zone
\cite{1964DesCloizeaux,1964DesCloizeaux2,1983Nenciu,1988HelfferSjostrand,1991Nenciu,2007BrouderPanatiCalandraMourougane,2007Panati,2018MonacoPanatiPisanteTeufel}. To our knowledge, the observation that the uniform spectral gap assumption (Assumption \ref{def:usg}) is sufficient for the construction of an analytic and periodic Bloch frame is original to the mathematical literature, although it is implicit in the works of Soluyanov and Vanderbilt \cite{2011SoluyanovVanderbilt,2011SoluyanovVanderbilt2,2014TaherinejadGarrityVanderbilt,2017GreschYazyevTroyerVanderbiltBernevigSoluyanov,2018WuZhangSongTroyerSoluyanov}, and the techniques required for the proof are basically known \cite{2016CorneanHerbstNenciu,Cances2017}.




We will make the connection between the operator $PXP$ and existence of a periodic and analytic Bloch frame through the hybrid Wannier functions (HWFs) first described in \cite{2001SgiarovelloPeressiResta} (see also \cite{1991Niu}). The key observation is that, with an appropriate Bloch frame, HWFs are exact eigenfunctions of $PXP$ \cite{1991Niu,2011SoluyanovVanderbilt,2011SoluyanovVanderbilt2,2014TaherinejadGarrityVanderbilt,2017GreschYazyevTroyerVanderbiltBernevigSoluyanov,2018WuZhangSongTroyerSoluyanov}. 
We will see that the uniform spectral gap assumption (Assumption \ref{def:usg}) gives a sufficient condition for existence of analytic and periodic Bloch frame over the whole Brillouin zone.

The structure of this section is as follows. We will first recall the aspects of Bloch theory which are necessary to state our results in Section \ref{sec:bloch}. We will then make our assumptions precise and state the theorem which we will prove (Theorem \ref{th:Chern_thm}) in Section \ref{sec:prob_stat}. We will then construct an analytic frame over the whole Brillouin zone which is periodic with respect to \emph{one} of the components of the quasi-momentum (Section \ref{sec:construct_gauge}). We will then show that the hybrid Wannier functions defined with respect to this frame diagonalize $PXP$, and show that the uniform spectral gap assumption implies a condition necessary for existence of an analytic Bloch frame which is periodic with respect to \emph{both} components of the quasi-momentum (Section \ref{sec:implies_cond}).

\subsection{Notation and Bloch theory in two spatial dimensions} \label{sec:bloch} We consider Hamiltonians with the form (recall Assumption \ref{as:H_assump})
\begin{equation} \label{eq:H_per}
    H = ( A - i \nabla )^2 + V,
\end{equation}
where $V \in L^2_{\uloc}(\R^2; \R)$, $A \in L^4_{\loc}(\R^2; \R^2)$, and $\divd A \in L^2_{\loc}(\R^2; \R)$. Let $\Lambda$ denote a two-dimensional lattice generated by non-parallel lattice vectors $\vec{v}_1, \vec{v}_2 \in \field{R}^2$, i.e.
\begin{equation} \label{eq:lattice_lambda}
    \Lambda = \left\{ m_1 \vec{v}_1 + m_2 \vec{v}_2 : ( m_1, m_2 ) \in \field{Z}^2 \right\}.
\end{equation}
We assume that the functions $A$ and $V$ are periodic with respect to $\Lambda$ in the sense that for all $\vec{x} \in \field{R}^2$
\begin{equation}
    A(\vec{x} + \vec{v}) = A(\vec{x}), \quad V(\vec{x} + \vec{v}) = V(\vec{x}) \text{ for all $\vec{v} \in \Lambda$.}
\end{equation}

Define $L_j := |\vec{v}_j|, j = 1,2$. We work with co-ordinates defined with respect to the lattice vectors $\vec{v}_1$ and $\vec{v}_2$ so that for $\vec{x} \in \field{R}^2$,
\begin{equation} \label{eq:x_coords}
    \vec{x} = x_1 \vec{v}_1 + x_2 \vec{v}_2
\end{equation}
for $(x_1, x_2) \in \field{R}^2$. We let $\Omega$ denote a fundamental cell of the lattice $\Lambda$, i.e.
\begin{equation} \label{eq:Omega}
    \Omega := \left\{ x_1 \vec{v}_1 + x_2 \vec{v}_2 : (x_1,x_2) \in [0,1]^2 \right\}.
\end{equation}

For any $\vec{k} \in \field{R}^2$, let $L^2_{\vec{k}}$ denote the space of $L^2$ functions on $\Omega$ with ``$\vec{k}$-quasi-periodic'' boundary conditions
\begin{equation} \label{eq:L2k}
    L^2_{\vec{k}} := \left\{ f(\vec{x}) \in L^2(\Omega) : f(\vec{x} + \vec{v}) = e^{i \vec{k} \cdot \vec{v}} f(\vec{x}) \quad \forall \vec{v} \in \Lambda \right\},
\end{equation}
and let $L^2_{per}$ denote the same space but with periodic boundary conditions
\begin{equation}
    L^2_{per} := \left\{ f(\vec{x}) \in L^2(\Omega) : f(\vec{x} + \vec{v}) = f(\vec{x}) \quad \forall \vec{v} \in \Lambda \right\}.
\end{equation}
The operator $H$ restricted to any of the spaces $L^2_{\vec{k}}$ is self-adjoint and has compact resolvent. Its spectrum therefore consists only of real eigenvalues which can be ordered with multiplicity as
\begin{equation} \label{eq:Bloch_bands}
    E_1(\vec{k}) \leq E_2(\vec{k}) \leq ... \leq E_n(\vec{k}) \leq ....
\end{equation}
The functions $\vec{k} \mapsto E_n(\vec{k})$ are known as Bloch band functions, and the associated $L^2_{\vec{k}}$-eigenfunctions of these eigenvalues $\Phi_n(\vec{x},\vec{k})$ are known as Bloch functions. The parameter $\vec{k}$ is known as the quasi-momentum (or crystal momentum). 

The dual lattice is defined by
\begin{equation}
    \Lambda^* := \left\{ m_1 \vec{w}_1 + m_2 \vec{w}_2 : (m_1,m_2) \in \field{Z}^2 \right\},
\end{equation}
where $\vec{w}_1$ and $\vec{w}_2$ are defined by the relations
\begin{equation}
    \vec{v}_i \cdot \vec{w}_j = 2 \pi \delta_{ij} \quad i,j = 1,2.
\end{equation}
We work with co-ordinates defined with respect to the lattice vectors $\vec{w}_1$ and $\vec{w}_2$ so that
\begin{equation} \label{eq:k_coords}
    \vec{k} = \frac{ k_1 }{ 2 \pi } \vec{w}_1 + \frac{ k_2 }{ 2 \pi } \vec{w}_2.
\end{equation}
Since $L^2_{\vec{k}+\vec{w}} = L^2_{\vec{k}}$ for any $\vec{w} \in \Lambda^*$, the Bloch band functions \eqref{eq:Bloch_bands} and associated eigenprojections $P_n(\vec{k})$ are periodic with respect to $\Lambda^*$. There is thus no loss in restricting attention to $\vec{k}$ only in a fundamental cell of the dual lattice which we denote by $\mathcal{B}$ and refer to as the Brillouin zone (see Remark \ref{rem:BZ})
\begin{equation} \label{eq:B}
    \mathcal{B} := \left\{ \frac{k_1}{2 \pi} \vec{w}_1 + \frac{k_2}{2 \pi} \vec{w}_2 : (k_1,k_2) \in [-\pi,\pi]^2 \right\}.
\end{equation}

The Bloch functions $\Phi_n(\vec{x},\vec{k})$ extend naturally to functions on $\field{R}^2$ using the boundary condition \eqref{eq:L2k}. These functions form a complete set in $L^2(\field{R}^2)$ in the following sense (see \cite{gel1950expansion,Odeh1964,papanicolau1978asymptotic} for proofs). For any $f \in L^2(\field{R}^2)$, define the Floquet-Bloch coefficients
\begin{equation}
    \tilde{f}_n(\vec{k}) := \ip{ \Phi_n(\cdot,\vec{k}) }{ f(\cdot) }_{L^2(\field{R}^2)} \quad n \geq 1, \vec{k} \in \mathcal{B}. 
\end{equation}
Then, for a suitable normalization of the $\Phi_n(\vec{x},\vec{k})$,
\begin{equation}
    f(\vec{x}) = \sum_{n \geq 1} \inty{\mathcal{B}}{}{ \tilde{f}_n(\vec{k}) \Phi_n(\vec{x},\vec{k}) }{\vec{k}}.
\end{equation}

It follows that the $L^2(\field{R}^2)$-spectrum of $H$ is simply the union of the set of 
closed intervals swept out by the maps $\vec{k} \mapsto E_n(\vec{k})$ as $\vec{k}$ varies over the Brillouin zone, i.e.
\begin{equation}
    \sigma(H) = \bigcup_{ n \in \field{N} } \bigcup_{ \vec{k} \in \mathcal{B} } E_n(\vec{k}).
\end{equation}

The eigenvalues \eqref{eq:Bloch_bands} and associated Bloch functions can equivalently be obtained by noting that $\Phi_n(\vec{x},\vec{k}) \in L^2_{\vec{k}}$ is an eigenfunction of $H$ with eigenvalue $E_n(\vec{k})$ if and only if $\chi_n(\vec{x},\vec{k}) = e^{-i \vec{k} \cdot \vec{x}} \Phi_n(\vec{x},\vec{k})$ is an eigenfunction of the operator
\begin{equation} \label{eq:Bloch_op}
    H(\vec{k}) := ( \vec{k} + A - i \nabla )^2 + V
\end{equation}
acting on $L^2_{per}$ with eigenvalue $E_n(\vec{k})$. Each $\chi_n(\vec{x},\vec{k})$ is known as a ``periodic Bloch function''. The eigenprojections $P_{per,n}(\vec{k})$ onto periodic Bloch functions satisfy the symmetry
\begin{equation} \label{eq:chi}
    P_{per,n}(\vec{k} + \vec{w}) = e^{- i \vec{w} \cdot \vec{k}} P_{per,n}(\vec{k}) e^{i \vec{w} \cdot \vec{k}} \quad \forall \vec{w} \in \Lambda^*.
\end{equation}

\begin{remark} \label{rem:BZ}
Note that generally speaking the Brillouin zone refers to the fundamental cell of the dual lattice defined as the set of points closer to the origin of the reciprocal lattice than to any other reciprocal lattice points (Wigner-Seitz cell), which does not agree with the definition \eqref{eq:B} for every Bravais lattice. Here we abuse notation because for our purposes it is simpler to work with the rectangular geometry \eqref{eq:B}.
\end{remark}


\subsection{Problem statement} \label{sec:prob_stat}

We at this point make our assumptions precise. Let $H$ be a periodic Schr\"odinger operator as in \eqref{eq:H_per}.
\begin{assumption} \label{as:per_spec_gap}
We assume that there exists a positive integer $N$ such that the $N$th and $N+1$th Bloch band functions of $H$ are separated by a uniform gap, i.e.
\begin{equation}
    E_{\text{gap}} := \min_{\vec{k} \in \mathcal{B}} E_{N+1}(\vec{k}) - \max_{\vec{k} \in \mathcal{B}} E_N(\vec{k}) > 0.
\end{equation}
\end{assumption}
Assumption \ref{as:per_spec_gap} is essentially Assumption \ref{as:gap_assump} specialized to the crystalline case. Note that we do not assume gaps between the other Bloch band functions $E_n(\vec{k})$ where $1 \leq n \leq N$. We define the Fermi projection $P : L^2(\field{R}^2) \rightarrow L^2(\field{R}^2)$ as the orthogonal projection onto the set of Bloch functions associated with the first $N$ bands, i.e.
\begin{equation} \label{eq:P_2}
    P f(\vec{x}) = \sum_{n = 1}^N \inty{\mathcal{B}}{}{ \ip{ \Phi_n(\cdot,\vec{k})}{ f(\cdot) }_{L^2(\field{R}^2)} \Phi_n(\vec{x},\vec{k}) }{\vec{k}}.
\end{equation}
Because of the gap assumption (Assumption \ref{as:per_spec_gap}), we could equivalently define $P$ via a Riesz projection. It is useful to introduce orthogonal projections $P(\vec{k}) : L^2_{\vec{k}} \rightarrow L^2_{\vec{k}}$ onto the sets of Bloch functions associated with the first $N$ bands for each $\vec{k} \in \mathcal{B}$, i.e. 
\begin{equation}
    P(\vec{k}) f(\vec{x}) = \sum_{n = 1}^N \ip{\Phi_n(\cdot,\vec{k})}{f(\cdot)}_{L^2_{\vec{k}}} \Phi_n(\vec{x},\vec{k}).
\end{equation}
We finally introduce orthogonal projections $P_{per}(\vec{k}) : L^2_{per} \rightarrow L^2_{per}$ onto the sets of periodic Bloch functions associated with the first $N$ bands for each $\vec{k} \in \mathcal{B}$, i.e.
\begin{equation}
    P_{per}(\vec{k}) f(\vec{x}) = \sum_{n = 1}^N \ip{\chi_n(\cdot,\vec{k})}{f(\cdot)}_{L^2_{per}} \chi_n(\vec{x},\vec{k}).
\end{equation}

We now introduce the precise concept of a Bloch frame.
\begin{definition} \label{def:BF_gauge}
A Bloch frame is a choice of basis for $\range P_{per}(\vec{k})$ at every $\vec{k} \in \mathcal{B}$, i.e. a collection of maps
\begin{equation}
\begin{split}
    \mathcal{B} &\rightarrow \left( L^2_{per} \right)^N  \\
    \vec{k} &\mapsto (\Xi_1(\vec{x},\vec{k}),...,\Xi_N(\vec{x},\vec{k})),
\end{split}
\end{equation}
where the set $(\Xi_1(\vec{x},\vec{k}),...,\Xi_N(\vec{x},\vec{k}))$ is a basis for $\range P_{per}(\vec{k})$ at every $\vec{k} \in \mathcal{B}$. We say the Bloch frame is analytic if every map $\vec{k} \mapsto \Xi_n(\vec{x},\vec{k}), 1 \leq n \leq N$ is (real) analytic for all $\vec{k} \in \mathcal{B}$. We say the Bloch frame is periodic if every map $\vec{k} \mapsto \Xi_n(\vec{x},\vec{k}), 1 \leq n \leq N$ satisfies 
\begin{equation} 
    \Xi_n(\vec{x},\vec{k}+\vec{w}) = e^{- i \vec{w} \cdot \vec{x}} \Xi_n(\vec{x},\vec{k}) \quad 1 \leq n \leq N, \forall \vec{w} \in \Lambda^*.
\end{equation}
\end{definition}
We will prove that
\begin{theorem} \label{th:Chern_thm}
Let $P$ be the Fermi projection of a periodic Hamiltonian as in \eqref{eq:P_2}, and assume the spectral gap Assumption \ref{as:per_spec_gap}. Recall the co-ordinates introduced in \eqref{eq:x_coords}, and define the operator $X$ by
\begin{equation} \label{eq:X_def}
    X f(\vec{x}) = L_1 x_1 f(\vec{x}).
\end{equation}
Then, if the operator $PXP$ satisfies the uniform spectral gap assumption (Assumption \ref{def:usg}), an analytic and periodic Bloch function gauge in the sense of Definition \ref{def:BF_gauge} exists.
\end{theorem}
We will prove Theorem \ref{th:Chern_thm} across the following sections. Recall that we have introduced natural $\vec{k}$-space co-ordinates $(k_1,k_2)$ \eqref{eq:k_coords}. We will first construct a Bloch function gauge which is analytic over the Brillouin zone and periodic with respect to $k_1$. We will then prove that, whenever the uniform spectral gap assumption on $PXP$ holds, a gauge which is periodic with both $k_1$ and $k_2$ can be constructed. 

\subsection{Proof of Theorem \ref{th:Chern_thm}: Parallel transport unitaries and construction of an analytic gauge periodic with respect to $k_1$} \label{sec:construct_gauge}

We first recall the notion of parallel transport of periodic Bloch functions (see \cite{2016CorneanHerbstNenciu} for more detail).

Recall we define $P_{per}(k_1,k_2)$ to be the $L^2_{per}$-projection onto the set of periodic Bloch functions at $k_1,k_2$, i.e. onto the span of $\left\{ \chi_n(\vec{x},k_1,k_2) \right\}_{1 \leq n \leq N}$. The operators $P_{per}(k_1,k_2)$ satisfy the following symmetries, which follow from symmetry of the operator $H(\vec{k})$ \eqref{eq:Bloch_op}
\begin{equation} \label{eq:P_symm}
\begin{split}
    &e^{- i \vec{w}_2 \cdot \vec{x}} P_{per}(k_1,k_2) e^{i \vec{w}_2 \cdot \vec{x}} = P_{per}(k_1,k_2 + 2 \pi) \\
    &e^{- i \vec{w}_1 \cdot \vec{x}} P_{per}(k_1,k_2) e^{i \vec{w}_1 \cdot \vec{x}} = P_{per}(k_1 + 2 \pi,k_2).
\end{split}
\end{equation}

Following Kato \cite{Kato}, we define unitaries $T_{k_2}(k_1)$ for each fixed $k_2 \in [-\pi,\pi]$ and for $k_1 \in [-\pi,\pi]$ by the ODE
\begin{equation} \label{eq:para_trans_ODE}
\begin{split}
    &i \de_{k_1} T_{k_2}(k_1) = i [ \de_{k_1} P_{per}(k_1,k_2) , P_{per}(k_1,k_2) ] T_{k_2}(k_1)  \\
    &T_{k_2}(0) = I.
\end{split}
\end{equation}
We note three properties of these unitaries. First, the identity (see \cite{Kato}, Chapter 2 $\S$4 for the proof)
\begin{equation}
    P_{per}(k_1,k_2) T_{k_2}(k_1) = T_{k_2}(k_1) P_{per}(0,k_2)
\end{equation}
shows that $T_{k_2}(k_1)$ restricted to $\range(P_{per}(0,k_2))$ maps to $\range(P_{per}(k_1,k_2))$ bijectively. Second, using the symmetries \eqref{eq:P_symm}, we have that, for example,
\begin{equation}
\begin{split}
    &i \de_{k_1} \left( e^{- i \vec{w}_2 \cdot \vec{x}} T_{k_2}(k_1) e^{i \vec{w}_2 \cdot \vec{x}} \right) = i [ \de_{k_1} P_{per}(k_1,k_2 + 2 \pi) , P_{per}(k_1,k_2 + 2 \pi) ] \left( e^{- i \vec{w}_2 \cdot \vec{x}} T_{k_2}(k_1) e^{i \vec{w}_2 \cdot \vec{x}} \right)   \\
    &e^{- i \vec{w}_2 \cdot \vec{x}} T_{k_2}(0) e^{i \vec{w}_2 \cdot \vec{x}} = I,
\end{split}
\end{equation}
and hence (since they satisfy the same initial value problems)
\begin{equation} \label{eq:T_sym}
    T_{k_2 + 2 \pi}(k_1) = e^{- i \vec{w}_2 \cdot \vec{x}} T_{k_2}(k_1) e^{i \vec{w}_2 \cdot \vec{x}}, \quad T_{k_2}(k_1 + 2\pi) = e^{- i \vec{w}_1 \cdot \vec{x}} T_{k_2}(k_1) e^{i \vec{w}_1 \cdot \vec{x}}.
\end{equation}
Third, by the method of successive approximations (see \cite{CoddingtonLevinson}), the maps $k_1 \mapsto T_{k_2}(k_1)$ and $k_2 \mapsto T_{k_2}(k_1)$ are both (real) analytic since $P_{per}(k_1,k_2)$ is analytic in both variables. It is clear that by the same calculations we can also define unitaries $T_{k_1}(k_2)$ where $k_1 \in [-\pi,\pi]$ is fixed and $k_2 \in [-\pi,\pi]$ with analogous properties.

Although our ultimate aim in this section is to construct a Bloch frame which is periodic with respect to $k_1$, we first construct a Bloch frame which is periodic with respect to $k_2$ as follows. Let $\left\{ \Xi_n(\vec{x},0,0) \right\}_{1 \leq n \leq N}$ be an arbitrary basis of the span of $\left\{ \chi_n(\vec{x},0,0) \right\}_{1 \leq n \leq N}$. Define
\begin{equation}
    \Xi_n(\vec{x},0,k_2) := T_{0}(k_2) \Xi_n(\vec{x},0,0) \quad 1 \leq n \leq N, k_2 \in [-\pi,\pi].
\end{equation}
By analyticity of $T_0(k_2)$, the $\Xi_n(\vec{x},0,k_2)$ are analytic with respect to $k_2$. Using symmetry of $P_{per}(k_1,k_2)$ \eqref{eq:P_symm} we have that
\begin{equation}
    \Xi_m(\vec{x},0,- \pi) = e^{i \vec{w}_2 \cdot \vec{x}} \sum_{n = 1}^N U_{mn} \Xi_n(\vec{x},0,\pi)
\end{equation}
where $U = \left\{ U_{mn} \right\}_{1 \leq m \leq N, 1 \leq n \leq N}$ is a unitary $N \times N$ matrix with eigenvalues $(\lambda_1,...,\lambda_N)$. By unitarity we can write $U$ as
\begin{equation}
    U = Q D Q^*, \quad D := \diag(\lambda_1,...,\lambda_N),
\end{equation}
where $Q$ is the matrix whose columns are the normalized eigenvectors of $U$, $\lambda_1,...,\lambda_N$ denotes the eigenvalues of $U$, and $\diag(v)$ denotes the diagonal matrix with the components of the vector $v$ along its diagonal. Again by unitarity, we have that each eigenvalue of $U$ can be written $\lambda_m = e^{i \Gamma_m}$ where $\Gamma_m$ is real and $-\pi \leq \Gamma_m < \pi$. We define the matrix log of $U$ in the usual way via spectral calculus so that
\begin{equation}
    - i \log U := Q \left( - i \log D \right) Q^*, \quad ( - i \log D ) := \diag(\Gamma_1,...,\Gamma_N).
\end{equation}
To obtain a periodic frame along the line $k_1 = 0$ we define
\begin{equation}
    \exp\left( - i \frac{ k_2 (- i \log U) }{ 2 \pi } \right) := Q \exp\left( - i \frac{ k_2 (- i \log D) }{ 2 \pi } \right) Q^*
\end{equation}
where
\begin{equation}
    \exp\left( - i \frac{ k_2 (- i \log D) }{ 2 \pi } \right) := \diag\left(e^{- i \frac{ \Gamma_1 k_2 }{2 \pi} }, ... ,e^{- i \frac{ \Gamma_N k_2 }{2 \pi}}\right).
\end{equation}
We then define a new frame along the line $k_1 = 0$ by
\begin{equation} \label{eq:use_U_log}
    \widetilde{\Xi}_n(\vec{x},0,k_2) := \exp\left( - i \frac{ k_2 ( - i \log U ) }{ 2 \pi } \right) \Xi_n(\vec{x},0,k_2).
\end{equation}
The $\widetilde{\Xi}_n(\vec{x},0,k_2)$ are clearly analytic with respect to $k_2$, and are now periodic in the sense that 
\begin{equation} \label{eq:til_chi_per}
    \widetilde{\Xi}_n(\vec{x},0,- \pi) = e^{i \vec{w}_2 \cdot \vec{x}} \widetilde{\Xi}_n(\vec{x},0,\pi).
\end{equation}

We now extend this frame to the whole Brillouin zone by defining
\begin{equation}
    \widetilde{\Xi}_n(\vec{x},k_1,k_2) := T_{k_2}(k_1) \widetilde{\Xi}_n(\vec{x},0,k_2) \quad 1 \leq n \leq N, (k_1,k_2) \in [-\pi,\pi]^2.
\end{equation}
Using the properties of $T_{k_2}(k_1)$ the $\widetilde{\Xi}_n(\vec{x},k_1,k_2)$ are analytic with respect to both $k_1$ and $k_2$, and periodic with respect to $k_2$ since for any $k_1 \in [-\pi,\pi]$,
\begin{equation}
\begin{split}
    \widetilde{\Xi}_n(\vec{x},k_1,-\pi) = T_{-\pi}(k_1) \widetilde{\Xi}_n(\vec{x},0,-\pi) &= e^{i \vec{w}_2 \cdot \vec{x}} T_\pi (k_1) e^{- i \vec{w}_2 \cdot \vec{x}} e^{i \vec{w}_2 \cdot \vec{x}} \widetilde{\Xi}_n(\vec{x},0,\pi) \\
    &= e^{i \vec{w}_2 \cdot \vec{x}} T_\pi(k_1) \widetilde{\Xi}_n(\vec{x},0,\pi) \\
    &= e^{i \vec{w}_2 \cdot \vec{x}} \widetilde{\Xi}_n(\vec{x},k_1,\pi),
\end{split}
\end{equation}
where the second equality is by \eqref{eq:T_sym} and \eqref{eq:til_chi_per}. Note that using \eqref{eq:para_trans_ODE} the frame satisfies the following ``parallel transport'' property with respect to $k_1$
\begin{equation} \label{eq:para_trans}
    \ip{ \widetilde{\Xi}_m(\vec{x},k_1,k_2) }{ \de_{k_1} \widetilde{\Xi}_n(\vec{x},k_1,k_2) }_{L^2_{per}} = 0 \quad 1 \leq m, n \leq N, (k_1,k_2) \in [-\pi,\pi]^2.
\end{equation}

Having constructed a frame which is analytic and periodic with respect to $k_2$, we now aim to ``mend'' the gauge so that it is also periodic with respect to $k_1$. We will see that this is not always possible while preserving periodicity with respect to $k_2$.

For each $k_2$, the gauge so far constructed will not in general be periodic with respect to $k_1$, but must satisfy
\begin{equation} \label{eq:kap_1_per}
    \widetilde{\Xi}_m(\vec{x},-\pi,k_2) = e^{i \vec{w}_1 \cdot \vec{x}} \sum_{n = 1}^N U_{mn}(k_2) \widetilde{\Xi}_n(\vec{x},\pi,k_2) \quad k_2 \in [-\pi,\pi]
\end{equation}
for a family of $N \times N$ unitary matrices $U(k_2) = \left\{ U_{mn}(k_2) \right\}_{1 \leq m \leq N, 1 \leq n \leq N}$ depending analytically and periodically on $k_2 \in [-\pi,\pi]$. 

We now require the following Lemma.
\begin{lemma} \label{lem:per_evec}
Let $\phi(\tau)$ be a family of unitary $N \times N$ matrices depending real-analytically and periodically on $\tau \in [-\pi,\pi]$. Then there exist real-analytic functions
\begin{equation}
\begin{split}
    [-\pi,\pi] &\rightarrow \field{R}^N \\
    \tau &\mapsto \vec{v}_n(\tau)
\end{split}
\end{equation}
for $1 \leq n \leq N$ such that $\vec{v}_n(\tau)$ is a normalized eigenvector of $\phi(\tau)$ for every $\tau \in [-\pi,\pi]$. The associated eigenvalues of these eigenvectors are also real-analytic. If the eigenvalues and eigenprojections of $\phi(\tau)$ are periodic then the $\vec{v}_n(\tau)$ can be chosen to be additionally periodic.
\end{lemma}
\begin{proof}
We start by noting that the roots of the characteristic polynomial of $\phi(\tau)$ are branches of analytic functions of $\tau$ with only algebraic singularities \cite{Kato}. It follows that the number of eigenvalues of $\phi(\tau)$ is a constant with the exception of finitely many points in the interval $[-\pi,\pi]$, which implies that every eigenvalue of $\phi(\tau)$ is either degenerate for all $\tau$ or non-degenerate except at finitely many $\tau$. It follows that we can find $\tau^*$ such that every eigenvalue which is not degenerate for every $\tau \in [-\pi,\pi]$ is non-degenerate. We can define projections onto each eigenvector, or degenerate family of eigenvectors, in a neighborhood of $\tau^*$ via the Riesz projection formula
\begin{equation} \label{eq:reese}
    Q_n(\tau) := \frac{1}{2 \pi i} \inty{\gamma_n}{}{ ( z - \phi(\tau) )^{-1} }{z},
\end{equation}
where $\gamma_n$ is a contour enclosing exactly one eigenvalue of $\phi(\tau^*)$. It is clear that the projection is analytic in a neighborhood of $\tau^*$. We next note that each projection and eigenvalue can be analytically continued over the whole interval $[-\pi,\pi]$, even through eigenvalue crossings (see Rellich \cite{rellich1969perturbation}). We can now define eigenvectors which are analytic in $\tau$ by defining parallel transport unitaries associated to the projections $Q_n(\tau)$ as in \eqref{eq:para_trans_ODE}. To see the last part of the Lemma, note that if the eigenprojections are additionally periodic, the eigenvectors can be made periodic by a phase shift as in \eqref{eq:use_U_log}.
\end{proof}
\begin{remark} \label{rem:periodic}
It is important to note that real-analytic and $2 \pi$-periodic matrices may not have real-analytic and $2 \pi$-periodic eigenvalues or eigenprojections. To see this, define
\begin{equation}
    \vec{v}_1(\tau) := \frac{1}{\sqrt{1 + \cos^2\left(\frac{\tau}{2}\right)}} \begin{pmatrix} 1 \\ \cos\left(\frac{\tau}{2}\right) \end{pmatrix}, \quad \vec{v}_2(\tau) := \frac{1}{\sqrt{1 + \cos^2\left(\frac{\tau}{2}\right)}} \begin{pmatrix} 1 \\ - \cos\left(\frac{\tau}{2}\right) \end{pmatrix},
\end{equation}
and consider the real-analytic and $2 \pi$-periodic matrix
\begin{equation}
    \phi(\tau) := e^{i \cos(\tau/2)} \vec{v}_1 \otimes \vec{v}_1^\top + e^{- i \cos(\tau/2)} \vec{v}_2 \otimes \vec{v}_2^\top, 
\end{equation}
which has real-analytic eigenprojections and eigenvalues which are not $2 \pi$-periodic.
\end{remark}
Using Lemma \ref{lem:per_evec} we can write $U(k_2)$ as
\begin{equation}
    U(k_2) = Q(k_2) D(k_2) Q^*(k_2), \quad D(k_2) := \diag(\lambda_1(k_2),...,\lambda_N(k_2))
\end{equation}
where $Q(k_2)$ is the matrix whose columns are normalized eigenvectors of $U(k_2)$, and $\lambda_1(k_2), ... ,\lambda_N(k_2)$ are the eigenvalues of $U(k_2)$. 
We now define
\begin{equation} \label{eq:breve_frame}
    \breve{\Xi}_m(x,k_1,k_2) = \sum_{n = 1}^N Q_{mn}(k_2) \widetilde{\Xi}_n(x,k_1,k_2) \quad 1 \leq m \leq N, (k_1,k_2) \in [-\pi,\pi]^2.
\end{equation}
The new frame $\left\{ \breve{\Xi}_n(\vec{x},k_1,k_2) \right\}_{1 \leq n \leq N}$ retains real-analyticity with respect to $k_1$ and $k_2$, and now satisfies, instead of \eqref{eq:kap_1_per},
\begin{equation} \label{eq:kap_1_per_2}
    \breve{\Xi}_n(\vec{x},-\pi,k_2) = e^{i \vec{w}_1 \cdot \vec{x}} \lambda_n(k_2) \breve{\Xi}_n(\vec{x},\pi,k_2) \quad 1 \leq n \leq N.
\end{equation}
Because of the behavior described in Remark \ref{rem:periodic} it is not necessarily true that the new frame retains periodicity with respect to $k_2$. We can nevertheless shift the frame once more as
\begin{equation} \label{eq:mending}
    \check{\Xi}_n(\vec{x},k_1,k_2) = e^{- i \frac{\Gamma_n(k_2) k_1}{2 \pi}} \breve{\Xi}_n(\vec{x},k_1,k_2) \quad 1 \leq n \leq N, (k_1,k_2) \in [-\pi,\pi]^2,
\end{equation}
where $\Gamma_n(k_2)$ is, for each $1 \leq n \leq N$ and each $k_2 \in [-\pi,\pi]$, the logarithm of $\lambda_n(k_2)$ chosen such that $k_2 \mapsto \Gamma_n(k_2)$ is analytic and $-\pi \leq \Gamma_n(0) < \pi$. Clearly the frame $\left\{ \check{\Xi}_n(\vec{x},k_1,k_2) \right\}_{1 \leq n \leq N}$ is real-analytic and periodic with respect to $k_1$ and real-analytic with respect to $k_2$. Note that even if the eigenvalues and eigenvectors of $U(k_2)$ can be chosen to be periodic in $k_2$ (so that the frame \eqref{eq:breve_frame} is actually periodic in $k_2$), the frame $\left\{ \check{\Xi}_n(\vec{x},k_1,k_2) \right\}_{1 \leq n \leq N}$ still may not be periodic with respect to $k_2$ because periodicity of the $\lambda_n(k_2)$ does not imply periodicity of the $\Gamma_n(k_2)$, only periodicity mod $2 \pi$, i.e. it is only guaranteed that 
\begin{equation} \label{eq:eq_mod_2pi}
    \Gamma_n(\pi) = \Gamma_n(-\pi) \mod 2 \pi \quad 1 \leq n \leq N.
\end{equation}
If $\Gamma_n(\pi) \neq \Gamma_n(-\pi)$, the frame will not retain periodicity with respect to $k_2$. 
\begin{remark}
A simple example of a real-analytic and $2 \pi$-periodic unitary operator, with periodic eigenprojections and eigenvalues, whose logarithm is nonetheless not periodic, is $\phi(\tau) := e^{i \tau}$. 
\end{remark}

We will see in the next section that when $PXP$ satisfies the uniform spectral gap assumption (Assumption \ref{def:usg}), it is always possible to construct an analytic and periodic gauge with respect to both $k_1$ and $k_2$.

\subsection{Proof of Theorem \ref{th:Chern_thm}: the $PXP$ gap assumption implies existence of an analytic and periodic Bloch frame} \label{sec:implies_cond}

In the previous section we constructed a Bloch frame $\left\{ \check{\Xi}_n(\vec{x},k_1,k_2) \right\}_{1 \leq n \leq N}$ which was analytic and periodic with respect to $k_1$ but not necessarily periodic with respect to $k_2$. We will see that this frame diagonalizes the operator $PXP$, allowing us to link the uniform gap assumption to the matrix $U(k_2)$ defined by \eqref{eq:kap_1_per}, and from there to the possibility of constructing an analytic and periodic Bloch frame.

We now introduce a basis of $P L^2(\field{R}^2)$ which diagonalizes the operator $PXP$. Note that via Bloch theory (the operator $PXP$ is invariant under translations with respect to $\vec{v}_2$) it suffices to consider the operator $PXP$ restricted to the family of spaces
\begin{equation}
\begin{split}
    L^2_{k_2} &:= \left\{ f(\vec{x}) \in L^2(\Upsilon) : f(\vec{x} + \vec{v}_2) = e^{i \vec{k} \cdot \vec{v}_2} f(\vec{x}) \right\}  \\
    &= \left\{ f(x_1,x_2) \in L^2(\Upsilon) : f(x_1,x_2) = e^{i k_2} f(x_1,x_2) \right\},
\end{split}
\end{equation}
where $k_2 \in [0,2\pi]$ and $\Upsilon$ denotes the strip
\begin{equation}
    \Upsilon := \left\{ x_1 \vec{v}_1 + x_2 \vec{v}_2 : x_1 \in \field{R}, x_2 \in [0,1] \right\}.
\end{equation}
Define $P(k_2) : L^2_{k_2} \rightarrow L^2_{k_2}$ by
\begin{equation}
    P(k_2) f(\vec{x}) = \sum_{n = 1}^N \inty{-\pi}{\pi}{ \ip{ \Phi_n(\cdot,k_1,k_2) }{ f(\cdot) }_{L^2_{k_2}} \Phi_n(\vec{x},k_1,k_2) }{k_1},
\end{equation}
then the restrictions of the operator $PXP$ to the spaces $L^2_{k_2}$ are the operators $P(k_2) X P(k_2) : L^2_{k_2} \rightarrow L^2_{k_2}$ where $k_2 \in [-\pi,\pi]$. In particular, we have
\begin{equation} \label{eq:spec_PXP}
    \sigma(PXP) = \bigcup_{k_2 \in [-\pi,\pi]} \sigma( P(k_2) X P(k_2) ).
\end{equation}
We define quasi-Bloch functions (functions which span $\range P(k_1,k_2)$ but are not necessarily individually eigenfunctions of $H$) by multiplying each of the components of the Bloch frame by $e^{i \vec{k} \cdot \vec{x}} = e^{i \left[ k_1 x + k_2 y \right]}$ so that
\begin{equation}
    \check{\Psi}_n(\vec{x},k_1,k_2) = e^{i \left[ k_1 x + k_2 y \right]} \check{\Xi}_n(\vec{x},k_1,k_2).
\end{equation}
In particular, we can equivalently define $P(k_2)$ in terms of the quasi-Bloch functions as
\begin{equation}
    P(k_2) f(\vec{x}) = \sum_{n = 1}^N \inty{-\pi}{\pi}{ \ip{ \check{\Psi}_n(\cdot,k_1,k_2) }{ f(\cdot) }_{L^2_{k_2}} \check{\Psi}_n(\vec{x},k_1,k_2) }{k_1}.
\end{equation}

We define hybrid Wannier functions (HWFs) for each $M \in \field{Z}$ by integrating the quasi-Bloch functions with respect to $k_1$ along lines of constant $k_2$
\begin{equation} \label{eq:HWFs}
    H_n(\vec{x},M,k_2) = \frac{1}{2 \pi} \inty{-\pi}{\pi}{ \check{\Psi}_n(\vec{x},k_1,k_2) e^{- i k_1 M} }{k_1} \quad 1 \leq n \leq N.
\end{equation}
Note that each set of hybrid Wannier functions with fixed $k_2$ is complete in $L^2_{k_2}$ and the union over $k_2$ complete in $P L^2(\field{R}^2)$ by completeness of the quasi-Bloch functions.


We claim that, if the Bloch functions are formed using the frame just constructed (i.e. after the shift \eqref{eq:mending}) and the operator $X$ is defined as in \eqref{eq:X_def}, that $H_n(\vec{x},M,k_2)$ is an exact eigenfunction of $P(k_2) X P(k_2)$ for every $k_2 \in [-\pi,\pi]$ and $1 \leq n \leq N$. To see this, first note that
\begin{equation}
    P(k_2) X P(k_2) H_n(\vec{x},M,k_2) = \frac{1}{2 \pi} P(k_2) \inty{-\pi}{\pi}{ L_1 x e^{i [ k_1 x + k_2 y ]} \check{\Xi}(\vec{x},k_1,k_2) e^{- i k_1 M} }{k_1}.
\end{equation}
Since $x e^{i [k_1 x + k_2 y]} = - i \de_{k_1} e^{i [ k_1 x + k_2 y ] }$ and integrating by parts we have
\begin{equation} \label{eq:this_eq}
\begin{split}
    = &\frac{- i L_1}{2 \pi} P(k_2) \left[ \check{\Psi}(\vec{x},k_1,k_2) e^{- i k_1 M} \right]_{-\pi}^\pi    \\
    &+ \frac{i L_1}{2 \pi} P(k_2) \int_{-\pi}^{\pi} \sum_{m = 1}^N \ip{ \check{\Xi}_m(\cdot,k_1,k_2) }{ \de_{k_1} \check{\Xi}_n(\cdot,k_1,k_2) } \check{\Psi}_m(\vec{x},k_1,k_2) e^{- i k_1 M} \\
    &\quad \quad \quad \quad \quad \quad \quad \quad \quad + \check{\Psi}_n(\vec{x},k_1,k_2) (- i M) e^{- i k_1 M} \textrm{d}k_1,
\end{split}
\end{equation}
where we have expanded $\de_{k_1} \check{\Xi}_n(\vec{x},k_1,k_2)$ in terms of the other elements of the Bloch frame, with $\ip{}{}$ denoting the $L^2_{per}$ inner product. Since $\check{\Psi}_n(\vec{x},k_1,k_2)$ is $2 \pi$-periodic with respect to $k_1$, the first term of \eqref{eq:this_eq} vanishes. Now recall that using \eqref{eq:mending},
\begin{equation}
    \de_{k_1} \check{\Xi}_n(\vec{x},k_1,k_2) = - i \frac{\Gamma_n(k_2)}{2 \pi} \check{\Xi}_n(\vec{x},k_1,k_2) + e^{- i \frac{\Gamma_n(k_2) k_1}{2 \pi} } \de_{k_1} \breve{\Xi}_n(\vec{x},k_1,k_2),
\end{equation}
where $\ip{\breve{\Xi}_m(\cdot,k_1,k_2)}{\de_{k_1} \breve{\Xi}_n(\vec{x},k_1,k_2)} = 0$ for all $1 \leq m, n \leq N$ by \eqref{eq:para_trans}. Hence \eqref{eq:this_eq} simplifies to 
\begin{equation}
\begin{split}
    &P(k_2) X P(k_2) H_n(\vec{x},M,k_2) \\
    &= \left( L_1 M + \overline{x}(k_2) \right) H_n(\vec{x},M,k_2), \\
\end{split}
\end{equation}
where 
\begin{equation}
    \overline{x}_n(k_2) := \frac{ L_1 \Gamma_n(k_2) }{ 2 \pi }
\end{equation}
is known as the $n$th Wannier charge center ``at $k_2$''. Since $M$ is arbitrary, we see that the spectrum of $P(k_2) X P(k_2)$ is the union of shifted copies of $L_1 \field{Z}$
\begin{equation} \label{eq:spec_PXP}
    \sigma(P(k_2) X P(k_2)) = \bigcup_{1 \leq n \leq N} \left\{ L_1 \field{Z} + \overline{x}_n(k_2) \right\}.
\end{equation}

We now use the uniform spectral gap assumption to deduce information about the matrix $U(k_2)$. The uniform spectral gap assumption implies that $\sigma(PXP)$ is made up of intervals separated by gaps. Since $\sigma(PXP)$ is $L_1$-periodic, this implies that $\sigma(PXP) \cap [0,L_1)$ must have a gap and hence there exists an open subinterval $I \subset [0,L_1)$ such that
\begin{equation} \label{eq:must_gap}
\begin{split}
     &I \cap \bigcup_{k_2 \in [-\pi,\pi]} \bigcup_{1 \leq n \leq N} \{ L_1 \field{Z} + \overline{x}_n(k_2) \} = \emptyset    \\
     &\implies I \cap \bigcup_{k_2 \in [-\pi,\pi]} \bigcup_{1 \leq n \leq N} \left\{ L_1 \field{Z} + \frac{L_1 \Gamma_n(k_2)}{2 \pi} \right\} = \emptyset.
\end{split}
\end{equation}
But recall that the eigenvalues of the matrix $U(k_2)$ are given by $e^{i \Gamma_n(k_2)}$ for $1 \leq n \leq N$. Hence \eqref{eq:must_gap} implies the existence of $\lambda := e^{i \Gamma}$ which \emph{is not in the spectrum of $U(k_2)$ for any $k_2 \in [-\pi,\pi]$}. We now recall the following Lemma; for the proof, see Section 2.7.3 of \cite{2016CorneanHerbstNenciu}.
\begin{lemma} \label{lem:nenciu_lemma}
Let $\phi(\tau)$ be a family of unitary $N \times N$ matrices depending real-analytically and periodically on $\tau \in [-\pi,\pi]$, and assume that there exists $\lambda \in \field{C}$ with $|\lambda| = 1$ such that $\lambda \notin \bigcup_{\tau \in [-\pi,\pi]} \sigma( \phi(\tau) )$. Then there exists a self-adjoint Hermitian matrix $h(\tau)$, real-analytic and $2 \pi$-periodic in $\tau$, such that $\phi(\tau) = e^{i h(\tau)}$ for all $\tau$.
\end{lemma}

We can now prove Theorem \ref{th:Chern_thm}. Let $U(k_2)$ be as in \eqref{eq:kap_1_per}. By constructing a Bloch frame which diagonalizes $PXP$, we have seen that whenever $PXP$ has gaps, it must be that there exists a $\lambda$ in the complex unit circle which is not in the spectrum of $U(k_2)$ for any $k_2 \in [-\pi,\pi]$. But Lemma \ref{lem:nenciu_lemma} implies that in this case $U(k_2) = e^{i h(k_2)}$ for a family of Hermitian, real-analytic, and $2 \pi$-periodic matrices $h(k_2)$. So define a new Bloch frame by
\begin{equation}
    \acute{\Xi}_n(\vec{x},k_1,k_2) := \exp\left( - i \frac{k_1 h(k_2)}{2 \pi} \right) \widetilde{\Xi}_n(\vec{x},k_1,k_2) \quad 1 \leq n \leq N, (k_1, k_2) \in [-\pi, \pi]^2.
\end{equation}
Using the properties of the frame $\{ \widetilde{\Xi}_n(\vec{x},k_1,k_2) \}_{1 \leq n \leq N}$ and of the family $h(k_2)$, the Bloch frame $\{ \acute{\Xi}_n(\vec{x},k_1,k_2) \}_{1 \leq n \leq N}$ is a periodic and real-analytic Bloch frame in the sense of Definition \ref{def:BF_gauge} and we are done.

\begin{remark}
The proof of Theorem \ref{th:Chern_thm} clearly generalizes without difficulty to tight-binding models under analogous assumptions.
\end{remark}

Numerical computations of the spectra of $\im \log P(k_2) \exp\left(\frac{2 \pi i X}{M}\right) P(k_2)$ where $k_2 \in [-\pi,\pi]$ for the Haldane model (where $N = 1$) are shown in Figure \ref{fig:charge_centers}. Here $M$ denotes the number of cells in the $X$ direction. Since formally
\begin{equation}
    \frac{M}{2 \pi} \im \log \exp\left( \frac{2 \pi i X}{M} \right) = X \text{ mod } M,
\end{equation}
we expect the spectra of $\im \log P(k_2) \exp\left(\frac{2 \pi i X}{M}\right) P(k_2)$ approximate those of $P(k_2) X P(k_2)$ while respecting periodic boundary conditions (see Resta \cite{1998Resta}).
The computations illustrate clearly that the uniform spectral gap assumption fails when the model is in its topological phase because each eigenvalue of $\im \log P(k_2) \exp\left(\frac{2 \pi i X}{M}\right) P(k_2)$ sweeps out an interval of width $\frac{2 \pi}{M}$ as $k_2$ is varied over the interval $[-\pi,\pi]$, so that the union of the spectra is the whole interval $[-\pi,\pi]$. Similar figures appear in the works of Soluyanov, Vanderbilt and coauthors \cite{2011SoluyanovVanderbilt,2011SoluyanovVanderbilt2,2014TaherinejadGarrityVanderbilt,2017GreschYazyevTroyerVanderbiltBernevigSoluyanov,2018WuZhangSongTroyerSoluyanov}.

\begin{figure}
  \centering
  \includegraphics[width=.4\linewidth]{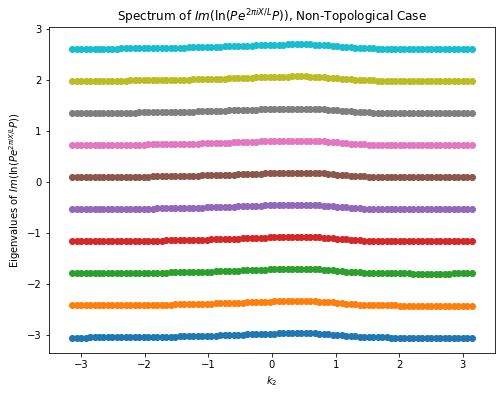}
  \includegraphics[width=.4\linewidth]{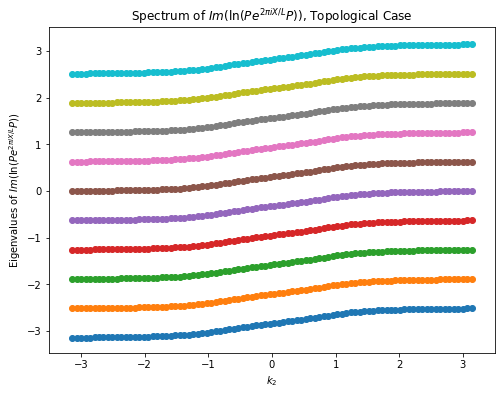} 
  \caption{Eigenvalues of $\im \log P(k_2) \exp\left(\frac{2 \pi i X}{10}\right) P(k_2)$ as a function of $k_2$, computed from the Haldane model on a $10 \times 90$ lattice with periodic boundary conditions in its non-topological (left) phase and topological (right) phase. Note that for large $M$, $\im \log P(k_2) \exp\left(\frac{2 \pi i X}{M}\right) P(k_2) \approx \frac{1}{M} P(k_2) X P(k_2)$. Eigenvalues of $\im \log P(k_2) \exp\left(\frac{2 \pi i X}{10}\right) P(k_2)$ sweep out closed intervals of width $\frac{2 \pi}{10}$ as $k_2$ is varied from $0$ to $2 \pi$ so that the spectrum of $\im \log P \exp\left(\frac{2 \pi i X}{10}\right) P$ does not have gaps. In the non-topological phase, eigenvalues of $\im \log P \exp\left(\frac{2 \pi i X}{10}\right) P$ return to their original values after $k_2$ is varied from $0$ to $2 \pi$.
 } \label{fig:charge_centers}
\end{figure}


\printbibliography

\end{document}
